\documentclass[final,3p,times,authoryear]{elsarticle}

\usepackage{amsmath, amssymb, amsthm}
\usepackage{siunitx}
\usepackage{bbm} 
\usepackage{booktabs,multirow}
\usepackage{placeins}
\usepackage{enumitem}
\usepackage{caption}
\newcommand{\ind}[1]{\mathbbm{1}_{\lrb{#1}}}
\newcommand{\lrb}[1]{\left\{#1\right\}}

\theoremstyle{definition}
\newtheorem{theorem}{Theorem}

\newtheorem{proposition}{Proposition}
\newtheorem{lemma}{Lemma}
\newtheorem{definition}{Definition}
\newtheorem{example}{Example}
\newtheorem{remark}{Remark}

\journal{Journal of Multivariate Analysis}

\begin{document}
	
	\begin{frontmatter}

		\title{Kernel-based independence and mean independence tests for weakly dependent data} 
		
		\author[1]{Daniel Diz-Castro}
		\author[1]{Manuel Febrero-Bande}
		\author[1]{Wenceslao Gonz\'alez-Manteiga}
		\affiliation[1]{organization={Department of Statistics, Mathematical Analysis and Optimization, Universidade de Santiago de Compostela},
			city={Santiago de Compostela},
			postcode={15782},
			country={Spain}}

		\begin{abstract}
			We provide a unified framework for independence and mean independence tests 
			based on the Hilbert-Schmidt independence criterion, extending some previous results in the literature to hold in general topological spaces. We also present a complete theoretical analysis of the 
			test statistic asymptotic behavior when the observed sample corresponds to a partial sample path of some stationary and ergodic stochastic process under near epoch dependence assumptions. In particular, we explore the test statistic consistency and limit distribution under both fixed and local hypothesis. The finite sample performance of the test(s) is illustrated with a succinct simulation study involving functional data.
		\end{abstract}

		
	\end{frontmatter}
	
	\section{Introduction}
	\label{sec1}
	
	The literature of statistical tests has received some notable contributions over the last two decades, coinciding with the rise of distance and kernel-based methods in statistical inference. In particular, the Hilbert-Schmidt independence criterion (HSIC) introduced in \cite{Gretton2007} is currently one of the most relevant family of tests of statistical independence between two random variables, $X$ and $Y$, since it is able to characterize this property (if the kernels involved are appropiately chosen) and, from a computational point of view, it can be implemented quite efficiently. Even though it was originally proposed as a test based on independent and identically distributed (i.i.d.) samples, an extension of the HSIC to allow for sample paths of absolutely regular stochastic processes taking values in Polish spaces can be found in \cite{Chwialkowski2014}, whereas in \cite{Chwialkowski2014a} the authors adress the Euclidean and $\tau$-dependent case based on some of the asymptotic results for $V$-statistics derived in \cite{Leucht2013}. Moreover, in \cite{Wang2021} the HSIC is adapted to test for the independence of two sequences of i.i.d. innovations coming from a parametric time series model in Euclidean spaces.
	
	In parallel, the distance correlation and distance covariance (DCOV) were presented in \cite{Szekely2007} as a tool for measuring lack of independence between Euclidean valued random variables and lead to another, widely used, family of independence tests based on i.i.d. samples. Several extensions of these ideas were established since then, including independence tests built upon i.i.d. samples of random variables taking values in metric spaces \citep{Lyons2013} as well as serial independence tests for time series and independence tests between variables of two (or more) time series in Euclidean spaces \citep{Zhou2012, Fokianos2018, Davis2018}, Hilbert spaces \citep{Hlavka2020, Meintanis2022}  and metric spaces \citep{Jiang2023}. A general theory for distance correlation based on dependent samples in metric spaces can be found in \cite{Kroll2021}.
	
	These two approaches, the distance-based and the kernel-based, are closely related, as the results in \cite{Sejdinovic2013} suggest. Furthermore, both of them can be adapted to test if $Y$ is mean independent of $X$; that is, to test if $\mathbb{E}[Y|X] = \mathbb{E}[Y]$ almost surely (a.s.). The martingale difference divergence (MDD) was introduced in \cite{Shao2014} as a distance covariance analogue to test the null hypothesis of mean independence in the i.i.d. Euclidean setting and latter extended to Hilbert spaces \citep{Lee2020} and to Euclidean time series \citep{Lee2017}. See also \cite{Li2023} for a more general version of the MDD in Euclidean spaces. The corresponding kernel-based test of mean independence was proposed in \cite{Lai2021} for Hilbert space valued i.i.d. samples. To the best of our knowledge, an adaptation of the kernel approach to mean independence testing with dependent samples is still missing.

	In this contribution, we take advantage of the theoretical proximity of the characterization of independence and mean independence by means of kernel-based null expectations to provide, in Section \ref{sec2}, a unified framework for independence and mean independence tests. The asymptotic behaviour of the test statistic, i.e., the empirical estimator of the HSIC, is analysed, in detail, in Section \ref{sec3},  where we assume that $X$ and $Y$ take values in a  Polish space (such as functional or directional data) and also that we observe a partial sample path of some discretely indexed stationary and ergodic stochastic process $\{(X_i,Y_i)\}_{i\in \mathbb{Z}}$. Distributional results are derived under near epoch dependence assumptions, which provide a substantially more general framework for modeling data dependence than the mixing conditions employed in all the aforementioned references, except for \cite{Zhou2012} and \cite{Chwialkowski2014}. We also propose a consistent wild bootstrap scheme to calibrate tests in practice, which is based on a new wild bootstrap consistency theorem for the sample mean of near epoch dependent data in Hilbert spaces. A simulation study is conducted in Section \ref{sec4} to assess the finite sample performance of the tests and conclusions are presented in Section \ref{sec5}. A real-data application is showcased in the supplementary material.
	
	\section{Kernel measures of independence and mean independence}
	\label{sec2}
	
	Throughout this document we assume that all random variables are defined on a common probability space $(\Omega, \mathcal{F},\mathbb{P})$ and take values in measurable spaces endowed with the corresponding Borel $\sigma$-algebras. In particular, $X \in \mathcal{X}$ and $Y\in \mathcal{Y}$ where $\mathcal{X}$ and  $\mathcal{Y}$ are two nonempty topological spaces that admit the definition of Borel probability measures.  Additional structure on those spaces will be made explicit when needed.
	
	As its name suggest, kernel-based methods rely on the special properties of positive definite bivariate functions that are symmetric in its arguments; i.e., kernels. We introduce in this section some basic notions about kernels and reproducing kernel Hilbert spaces (RKHS) as well as the theoretical foundings of the HSIC and the kernel-based conditional mean dependence measure (KCMD) of \cite{Lai2021}, the latter of which we extend to allow for non-Hilbert space valued $X$ and $Y$.
	
	\subsection{Independence}
	\label{sec2.1}
	
	Let $S$ be any nonempty topological space that admits the definition of Borel probability measures. The RKHS associated with a kernel $c_{\mathcal{S}}: \mathcal{S}\times \mathcal{S} \rightarrow \mathbb{R}$, which we denote by $\mathcal{C}_{\mathcal{S}}$, is a Hilbert space of functions from $\mathcal{S}$ to $\mathbb{R}$ obtained as the completion of the linear span of $\{c_{\mathcal{S}}(s,\cdot): s\in \mathcal{S}\}$ with respect to the metric induced by inner product which arises from the following reproducing property: $\langle c_{\mathcal{S}}(s,\cdot), \zeta_{\mathcal{S}} \rangle_{\mathcal{C}_{\mathcal{S}}} = \zeta_{\mathcal{S}}(s)$ for all $\zeta_{\mathcal{S}}\in \mathcal{C}_{\mathcal{S}}$. 
	The map $\Phi_{\mathcal{S}} : \mathcal{S}\rightarrow \mathcal{C}_{\mathcal{S}}$ such that $\Phi_{\mathcal{S}}(s) = c_{\mathcal{S}}(s,\cdot)$ is refered to as the canonical feature map of $c_{\mathcal{S}}$ and it is continuous (and hence, measurable) provided that $c_{\mathcal{S}}$ is continuous, since
	$$ ||\Phi_{\mathcal{S}}(s) - \Phi_{\mathcal{S}}(s') ||^2_{\mathcal{C}_{\mathcal{S}}} =  c_{\mathcal{S}}(s,s) -  c_{\mathcal{S}}(s,s')- c_{\mathcal{S}}(s',s)  + c_{\mathcal{S}}(s',s').$$
	Moreover, if $\Phi_{\mathcal{S}}$ is measurable and $S$ is Polish (separable and completely metrizable), the RKHS $\mathcal{C}_{\mathcal{S}}$ is separable (see, e.g., \citeauthor{Owhadi2016}, \citeyear{Owhadi2016}).
	
	Let $\mathcal{C}_{\mathcal{X}}$ and $\mathcal{C}_{\mathcal{Y}}$ be the RKHS associated with some kernels $c_{\mathcal{X}}:\mathcal{X}\times \mathcal{X} \rightarrow \mathbb{R}$ and $c_{\mathcal{Y}}:\mathcal{Y}\times \mathcal{Y} \rightarrow \mathbb{R}$ and let $\Phi_{\mathcal{X}}$ and $\Phi_{\mathcal{Y}}$ be the corresponding canonical feature maps. Considering tensor products of the form $\zeta_{\mathcal{X}} \otimes \zeta_{\mathcal{Y}}:\mathcal{C}_{\mathcal{Y}} \rightarrow \mathcal{C}_{\mathcal{X}}$, defined as linear maps such that $\zeta_{\mathcal{X}} \otimes \zeta_{\mathcal{Y}}(\zeta_{\mathcal{Y}}') = \zeta_{\mathcal{X}}\langle\zeta_{\mathcal{Y}}, \zeta'_{\mathcal{Y}}\rangle_{\mathcal{Y}}$ for all $\zeta_{\mathcal{X}}\in \mathcal{C}_{\mathcal{X}}$ and all $\zeta_{\mathcal{Y}},\zeta'_{\mathcal{Y}}\in \mathcal{C}_{\mathcal{Y}}$, we can construct a tensor product Hilbert space $\mathcal{HS} = \mathcal{C}_{\mathcal{X}} \otimes \mathcal{C}_{\mathcal{Y}}$ as the completion of the linear span of $\{\zeta_{\mathcal{X}}\otimes \zeta_{\mathcal{Y}}: \zeta_{\mathcal{X}}\in \mathcal{C}_{\mathcal{X}}, \zeta_{\mathcal{Y}}\in \mathcal{C}_{\mathcal{Y}}\}$ with respect to the inner product given by $\langle \zeta_{\mathcal{X}}\otimes \zeta_{\mathcal{Y}},\zeta'_{\mathcal{X}}\otimes \zeta'_{\mathcal{Y}} \rangle_{\mathcal{HS}} = \langle\zeta_{\mathcal{X}}, \zeta'_{\mathcal{X}}\rangle_{\mathcal{C}_{\mathcal{X}}}\langle\zeta_{\mathcal{Y}}, \zeta'_{\mathcal{Y}}\rangle_{\mathcal{C}_{\mathcal{Y}}}$. From this tensor product Hilbert space structure, we can define the HSIC as follows:
	
	\begin{definition}\label{def1}
		The Hilbert-Schmidt independence criterion associated with $X$, $Y$, $c_\mathcal{X}$ and $c_\mathcal{Y}$, is given by the expression
		\begin{equation}\label{1}
			\text{\normalfont HSIC}(X,Y;c_\mathcal{X},c_\mathcal{Y}) = \left|\left| \mathbb{E}\Big[(\Phi_{\mathcal{X}}(X)-\mathbb{E}[\Phi_{\mathcal{X}}(X)])\otimes (\Phi_{\mathcal{Y}}(Y)-\mathbb{E}[\Phi_{\mathcal{Y}}(Y)])\Big]\right|\right|^2_{\mathcal{HS}},
		\end{equation}	
		which is well-defined provided that
		$\Phi_{\mathcal{X}}(X)$, $\Phi_{\mathcal{Y}}(Y)$ and $\Phi_{\mathcal{X}}(X)\otimes \Phi_{\mathcal{Y}}(Y)$ are Pettis integrable. Unless we want to emphasize the specific choice of kernels $c_\mathcal{X}$ and $c_\mathcal{Y}$, we drop the explicit reference to them and use the simplified notation $\text{\normalfont HSIC}(X,Y)$.
	\end{definition}
	
	According to Definition \ref{def1}, $\text{\normalfont HSIC}(X,Y)$ is the squared norm in $\mathcal{HS}$ of the cross-covariance operator of $\Phi_{\mathcal{X}}(X)$ and $\Phi_{\mathcal{Y}}(Y)$ which, if $\mathcal{HS}$ is separable, coincides with its Hilbert-Schmidt norm. By definition of $\Phi_{\mathcal{X}}$ and $\Phi_{\mathcal{Y}}$, sufficient conditions for the existence of Pettis integrals (\citeauthor{Diestel1977}, \citeyear{Diestel1977}, Ch. 2) in (\ref{1}) are the continuity of $c_{\mathcal{X}}$ and $c_{\mathcal{Y}}$ and the finiteness of $\mathbb{E}[c_{\mathcal{X}}(X,X)^{1/2}]$, $\mathbb{E}[c_{\mathcal{Y}}(Y,Y)^{1/2}]$ and $\mathbb{E}[(c_{\mathcal{X}}(X,X)c_{\mathcal{Y}}(Y,Y))^{1/2}]$. Pettis integrals implicitly define a weak form of expectation by means of unidimensional projections: $\mathbb{E}[\Phi_{\mathcal{X}}(X)]$ and $\mathbb{E}[\Phi_{\mathcal{Y}}(Y)]$ are the unique elements of $\mathcal{C}_{\mathcal{X}}$ and $\mathcal{C}_{\mathcal{Y}}$ such that $\mathbb{E}[\langle \Phi_{\mathcal{X}}(X), \zeta_{\mathcal{X}}\rangle_{\mathcal{C}_{\mathcal{X}}}] = \langle \mathbb{E}[\Phi_{\mathcal{X}}(X)], \zeta_{\mathcal{X}}\rangle_{\mathcal{C}_{\mathcal{X}}}$ and $\mathbb{E}[\langle \Phi_{\mathcal{Y}}(Y), \zeta_{\mathcal{Y}}\rangle_{\mathcal{C}_{\mathcal{Y}}}]=\langle \mathbb{E}[\Phi_{\mathcal{Y}}(Y)], \zeta_{\mathcal{Y}}\rangle_{\mathcal{C}_{\mathcal{Y}}}$ for all $\zeta_{\mathcal{X}}\in \mathcal{C}_{\mathcal{X}}$ and all $\zeta_{\mathcal{Y}}\in \mathcal{C}_{\mathcal{Y}}$. Under the additional assumption of separability, Pettis integrals of measurable functions can be identified with Bochner integrals (\citeauthor{Diestel1977}, \citeyear{Diestel1977}, Ch. 2).
	
	Let $\Phi_{\mathcal{X}c}(x) = \Phi_{\mathcal{X}}(x) - \mathbb{E}[\Phi_{\mathcal{X}}(X)]$, $\Phi_{\mathcal{Y}c}(y) = \Phi_{\mathcal{Y}}(y) - \mathbb{E}[\Phi_{\mathcal{Y}}(Y)]$ and,  given two i.i.d. copies $(X,Y)$ and $(X',Y')$, define $c_{\mathcal{X}c}$ and $c_{\mathcal{Y}c}$ as 
	\begin{equation}\label{2}
		\begin{cases}
			c_{\mathcal{X}c}(x,x') = c_{\mathcal{X}}(x,x') - \mathbb{E}[c_{\mathcal{X}}(x,X')] - \mathbb{E}[c_{\mathcal{X}}(X,x')] + \mathbb{E}[c_{\mathcal{X}}(X,X')],\\
			c_{\mathcal{Y}c}(y,y')\hspace{0.15mm} = c_{\mathcal{Y}}(y,y') - \mathbb{E}[c_{\mathcal{Y}}(y,Y')] - \mathbb{E}[c_{\mathcal{Y}}(Y,y')] + \mathbb{E}[c_{\mathcal{Y}}(Y,Y')],
		\end{cases}
	\end{equation}
	so that $c_{\mathcal{X}c}(x,x') = \langle \Phi_{\mathcal{X}c}(x),\Phi_{\mathcal{X}c}(x') \rangle_{\mathcal{C}_{\mathcal{X}}}$ and $c_{\mathcal{Y}c}(y,y') = \langle \Phi_{\mathcal{Y}c}(y),\Phi_{\mathcal{Y}c}(y') \rangle_{\mathcal{C}_{\mathcal{Y}}}$ for all $x,x'\in \mathcal{X}$ and $y,y'\in \mathcal{Y}$ . Then, it is trivial to show that
	$$ \text{\normalfont HSIC}(X,Y) = \mathbb{E}[c_{\mathcal{X}c}(X,X')c_{\mathcal{Y}c}(Y,Y')].$$
	If $X$ and $Y$ are independent, $\text{\normalfont HSIC}(X,Y)$ factorizes into $\mathbb{E}[c_{\mathcal{X}c}(X,X')]$ and $\mathbb{E}[c_{\mathcal{Y}c}(Y,Y')]$, which are both equal to $0$ by construction. The reciprocal ($\text{\normalfont HSIC}(X,Y) = 0$ implies that $X$ and $Y$ are independent) holds if $c_{\mathcal{X}}$ and $c_{\mathcal{Y}}$ are characteristic.
	
	\begin{definition}\label{def2}
		We say that a kernel $c_{\mathcal{S}}$ is characteristic if the canonical feature map $\Phi_{\mathcal{S}}$ is measurable and, for all random variables $S,S'\in \mathcal{S}$ such that $\Phi_{\mathcal{S}}(S)$ and $\Phi_{\mathcal{S}}(S')$ are Pettis integrable, $\mathbb{E}[\Phi_{\mathcal{S}}(S)] = \mathbb{E}[\Phi_{\mathcal{S}}(S')]$ implies that $S$ and $S'$ are identically distributed.
	\end{definition}
	
	Some examples of characteristic kernels in separable Hilbert spaces are the distance-induced kernel $c_{\mathcal{S}}(s,s') = 0.5(||s||_{\mathcal{S}} + ||s'||_{\mathcal{S}} -||s-s'||_{\mathcal{S}})$ \citep{Sejdinovic2013} and the Gaussian kernel $c_{\mathcal{S}}(s,s') = \exp(-||s-s'||^2_{\mathcal{S}}/\sigma^2_{c_{\mathcal{S}}})$, where $\sigma^2_{c_{\mathcal{S}}}>0$ \citep{Ziegel2024}. For general Polish spaces we have the following result, the first part of which is a reformulation of Proposition 5.2 in \cite{Ziegel2024}.
	
	\begin{proposition}\label{pr1}
		Let $\mathcal{S}$ be any nonempty Polish space. Then:
		\begin{enumerate}[
			label = (\roman*),
			labelindent=\parindent,
			leftmargin=*,
			labelsep= 0em,
			align=left
			]
			\item If $\Xi : \mathcal{S} \rightarrow \tilde{\mathcal{S}}$ is an injective and measurable map taking values in a separable Hilbert space $\tilde{\mathcal{S}}$ and $c_{\tilde{\mathcal{S}}}:\tilde{\mathcal{S}}\times \tilde{\mathcal{S}} \rightarrow \mathbb{R}$ is a characteristic kernel, the kernel $c_{\mathcal{S}}(s,s') = c_{\tilde{\mathcal{S}}}(\Xi(s), \Xi(s'))$ is also characteristic.
			\item There always exists an injective and continuous map $\Xi : \mathcal{S} \rightarrow \ell^2$, where $\ell^2$ denotes the Hilbert space of real valued square summable sequences.
		\end{enumerate}
	\end{proposition}
	
	Canonical feature maps of characteristic kernels are injective. To prove this, it suffices to consider random variables which take a single value a.s. in Definition \ref{def2}. Moreover, by virtue of Corollary 16 in \cite{Sejdinovic2013}, if $c_\mathcal{S}$ is injective, it induces a semimetric on $S$:
	\begin{equation}\label{3}
		d_{\Phi_{\mathcal{S}}}(s,s') = ||\Phi_{\mathcal{S}}(s) - \Phi_{\mathcal{S}}(s') ||^2_{\mathcal{C}_{\mathcal{S}}}.
	\end{equation}
	In general, $d_{\Phi_{\mathcal{S}}}$ is not a metric since the triangle inequality mail fail to hold. For example, taking $c_{\mathbb{R}}(x,y) = \exp(-(x-y)^2)$ (the Gaussian kernel in $\mathbb{R}$), it holds that $d_{\Phi_{\mathbb{R}}} = 2-2\exp(-(x-y)^2)$. Thus, $d_{\Phi_{\mathbb{R}}}(0,1) \approx 0.63$ and $d_{\Phi_{\mathbb{R}}}(0,0.5) = d_{\Phi_{\mathbb{R}}}(0.5,1)  \approx 0.22$, so the triangle inequality does not hold.
	
	A direct consequence of Proposition 29 in \cite{Sejdinovic2013} and Theorem 3.11 in \cite{Lyons2013} is that, if $c_{\mathcal{X}}$ and $c_{\mathcal{Y}}$ are characteristic kernels and the induced semimetrics $d_{\Phi_{\mathcal{X}}}$ and $d_{\Phi_{\mathcal{Y}}}$ are, in fact, metrics, then $\text{\normalfont HSIC}(X,Y) = 0$ if and only if $X$ and $Y$ are independent. Alternatively, it is sufficient for $c_{\mathcal{X}}c_{\mathcal{Y}}$ to be a characteristic and continous kernel from $\mathcal{X}\times \mathcal{Y}$ to $\mathbb{R}$ \citep{Sejdinovic2013}. If $\mathcal{X}$ and $\mathcal{Y}$ are second countable topological spaces and $c_{\mathcal{X}}$ and $c_{\mathcal{Y}}$ are characteristic and bounded with measurable canonical feature maps, $\text{\normalfont HSIC}(X,Y)=0$ also implies independence \citep{Szabo2018}. According to the following result, the characteristic property of $c_{\mathcal{X}}$ and $c_{\mathcal{Y}}$ suffices to ensure that $\text{\normalfont HSIC}(X,Y)$ characterizes independence for any nonempty topological spaces $\mathcal{X}$ and $\mathcal{Y}$ where Borel probability measures can be defined. 
	
	\begin{proposition}\label{pr2}
		Let $c_{\mathcal{X}}$ and $c_{\mathcal{Y}}$ be characteristic and assume that
		$\Phi_{\mathcal{X}}(X)$, $\Phi_{\mathcal{Y}}(Y)$ and $\Phi_{\mathcal{X}}(X)\otimes \Phi_{\mathcal{Y}}(Y)$ are Pettis integrable. Then, the corresponding $\text{\normalfont HSIC}$ characterizes independence; that is,
		$$\text{\normalfont HSIC}(X,Y) = 0 \iff X \text{ and } Y \text{ are independent}.$$
	\end{proposition}
	
	\begin{remark}
		It follows from Proposition \ref{pr2} and also Proposition 29 in \cite{Sejdinovic2013} that Theorem 3.11 in \cite{Lyons2013} can be extended to semimetric spaces.
	\end{remark}
	
	\subsection{Mean independence}
	\label{sec2.2}
	
	A kernel-based measure of departure from the hypothesis of mean independence of $Y$ on $X$ was introduced in \cite{Lai2021} for random variables taking values in separable Hilbert spaces. It turns out that the kernel-based conditional mean dependence measure (KCMD) that the authors proposed is nothing but the HSIC with a linear kernel on $\mathcal{Y}$. The same applies to the martingale difference divergence \citep{Shao2014, Lee2020}.
	
	Let $\mathcal{Y}$ is a Hilbert space and let $\langle\star,\cdot\rangle_{\mathcal{Y}}$ denote the linear kernel defined as $\langle\star,\cdot\rangle_{\mathcal{Y}}(y,y') = \langle y,y'\rangle_{\mathcal{Y}}$. 
	The Pettis integrability of $Y$ suffices to guarantee that $\mathbb{E}[\langle Y, y' \rangle_{\mathcal{Y}}|X]$ exists as a conditional expectation with respect to the $\sigma$-algebra generated by $X$, which we denote by $\sigma(X)$, and is finite for all $y'\in \mathcal{Y}$. However, $\mathbb{E}[Y|X]$ may fail to exist as a $\sigma(X)$-measurable, $\mathcal{Y}$-valued and Pettis integrable random variable verifying $\mathbb{E}[\mathbb{E}[Y|X]\ind{X\in B_{\mathcal{X}}}] = \mathbb{E}[Y\ind{X\in B_{\mathcal{X}}}]$ for all $B_{\mathcal{X}}\in \mathcal{B}(\mathcal{X})$, the Borel $\sigma$-algebra of $\mathcal{X}$ \citep{Rybakov1971}. If $\mathbb{E}[Y|X]$ exists, it is a.s. unique and $\mathbb{E}[\langle Y, y' \rangle_{\mathcal{Y}}|X] =\langle  \mathbb{E}[Y|X],y' \rangle_{\mathcal{Y}}$ a.s. Moreover, a sufficient condition for the existence of $\mathbb{E}[Y|X]$ is that $Y$ is Bochner integrable.
	
	It is straightforward to verify that $\text{\normalfont HSIC}(X,Y;c_{\mathcal{X}}, \langle\star,\cdot\rangle_{\mathcal{Y}}) = 0$ whenever $\mathbb{E}[Y|X] = \mathbb{E}[Y]$ a.s. The reciprocal was proved in \cite{Lai2021} assuming that $\mathcal{X}$ and $\mathcal{Y}$ are separable Hilbert spaces, the Bochner integrability of $Y$ and, in addition, that $c_{\mathcal{X}}$ is a characteristic kernel. The following proposition extends this result to any Hilbert space $\mathcal{Y}$ and any nonempty topological space $\mathcal{X}$ on which Borel probability measures can be defined.
	
	\begin{proposition}\label{pr3}
		Let $\mathcal{Y}$ be a Hilbert space and let $c_{\mathcal{X}}$ be a characteristic kernel. Assume that   $\Phi_{\mathcal{X}}(X)$, $Y$ and $\Phi_{\mathcal{X}}(X)\otimes Y$ are Pettis integrable. Then,
		$$ \text{\normalfont HSIC}(X,Y;c_{\mathcal{X}}, \langle\star,\cdot\rangle_{\mathcal{Y}}) = 0 \iff \mathbb{E}[\langle Y, y'\rangle_{\mathcal{Y}}|X] = \langle \mathbb{E}[Y], y'\rangle_{\mathcal{Y}} \ \ \text{a.s. for all }y'\in \mathcal{Y}.$$ 
		Suppose, in addition, that $\mathbb{E}[Y|X]$ exists as Pettis integrable conditional expectation with respect to $\sigma(X)$. Then,
		$$ \text{\normalfont HSIC}(X,Y;c_{\mathcal{X}}, \langle\star,\cdot\rangle_{\mathcal{Y}}) = 0 \iff \mathbb{E}[Y|X] = \mathbb{E}[Y] \ \ \text{a.s}.$$ 
	\end{proposition}
	
	We may wonder if Proposition \ref{pr3} can be further extended to the case of $Y$ taking values in a Banach space. The answer is affirmative, provided that there is a continuous linear embedding $\Xi$ from $\mathcal{Y}$ into some Hilbert space $\tilde{\mathcal{S}}$. Indeed, if $\Xi$ is injective, linear and continous (and hence, bounded), it follows that $\Xi^*(\tilde{\mathcal{S}}^*)\subset \mathcal{Y}^*$ separates points of $\mathcal{Y}$, where $\Xi^*$ denotes the dual operator of $\Xi$, $\tilde{\mathcal{S}}^*$ is the dual space of $\tilde{\mathcal{S}}$ and $\mathcal{Y}^*$ is the dual space of $Y$. Moreover, $\Xi^*(\tilde{\mathcal{S}}^*) = \{\langle \Xi(\cdot), \tilde{s}\rangle_{\tilde{\mathcal{S}}}\mid\tilde{s}\in \tilde{\mathcal{S}} \}$. Thus, it follows from the definitions of conditional expectation with respect to $\sigma(X)$ and Pettis integrability that
	\begin{align*}
		\mathbb{E}[y^{\prime *}(Y)|X] =  y^{\prime *}(\mathbb{E}[Y])  \ \ \text{a.s. for all }y^{\prime *}\in \mathcal{Y}^* 
		&\iff y^{\prime *}\big(\mathbb{E}\big[(Y - \mathbb{E}[Y])\ind{X\in B_{\mathcal{X}}}\big]\big) = 0  \ \ \text{for all }y^{\prime *}\in \mathcal{Y}^*, B_{\mathcal{X}}\in \mathcal{B}(\mathcal{X})\\
		&\iff \big\langle \Xi\big(\mathbb{E}\big[(Y - \mathbb{E}[Y])\ind{X\in B_{\mathcal{X}}}\big]\big), \tilde{s}\big\rangle_{\tilde{\mathcal{S}}} = 0  \ \ \text{for all }\tilde{s}\in \tilde{\mathcal{S}}, B_{\mathcal{X}}\in \mathcal{B}(\mathcal{X})\\
		&\iff\mathbb{E}[ \langle \Xi(Y), \tilde{s}\rangle_{\tilde{\mathcal{S}}}|X] =  \langle \mathbb{E}[\Xi(Y)], \tilde{s}\rangle_{\tilde{\mathcal{S}}}  \ \ \text{a.s. for all }\tilde{s}\in \tilde{\mathcal{S}}.\\
	\end{align*}	
	where we have implicitly used that Pettis integrals, defined by dual space function evaluations, commute with bounded linear operators. Similarly, if $\mathbb{E}[Y|X]$ exists, $\mathbb{E}[\Xi(Y)|X] = \Xi(\mathbb{E}[Y|X])$ a.s. and
	$$\mathbb{E}[Y|X] = \mathbb{E}[Y]  \ \ \text{a.s.} \iff \Xi(\mathbb{E}[Y|X])  = \Xi(\mathbb{E}[Y]) \ \ \text{a.s.} \iff \mathbb{E}[\Xi(Y)|X]  = \mathbb{E}[\Xi(Y)] \ \ \text{a.s.},$$
	Then, we can determine if $\Xi(Y)$ is mean independent of $X$ via $\text{\normalfont HSIC}(X,\Xi(Y);c_{\mathcal{X}}, \langle\star,\cdot\rangle_{\tilde{\mathcal{S}}})$ and transfer the conclusions to the original pair $(X,Y)$.
	\begin{example}
		The $L^p$ spaces of functions in the $(0,1)$ interval, $L^p(0,1)$, with $2\le p<\infty$ and the space of continuous functions in the unit interval endowed with the supremum norm, $\mathcal{C}([0,1])$, are separable Banach spaces that can be embedded into $L^2(0,1)$ by choosing $\Xi$ as the inclusion function.
		Furthermore, such linear embeddings always exist provided that $\mathcal{Y}$ is separable since, by virtue of the Banach–Mazur Theorem (see, e.g., Theorem 1.4.3 in \citeauthor{Albiac2006}, \citeyear{Albiac2006}), all separable Banach spaces are linearly and isometrically embedded into a closed subset of $\mathcal{C}([0,1])$. 
	\end{example}
	
	\section{Kernel-based independence and mean independence tests}
	\label{sec3}

	Given a sample $\{(X_i,Y_i)\}_{i=1}^n$ with the same distribution as the pair $(X,Y)$, the empirical estimator of $\text{\normalfont HSIC}(X,Y)$ is given by the expression
	\begin{equation}\label{4}
		\text{\normalfont HSIC}_n(X,Y) = \left|\left|\frac{1}{n}\sum_{i=1}^n (\Phi_{\mathcal{X}}(X_i)- \bar{\Phi}_{\mathcal{X}}(X)) \otimes  (\Phi_{\mathcal{Y}}(Y_i)- \bar{\Phi}_{\mathcal{Y}}(Y))\right|\right|^2_{\mathcal{HS}},
	\end{equation}
	where $\bar{\Phi}_{\mathcal{X}}(X)$ and $\bar{\Phi}_{\mathcal{Y}}(Y)$ denote the sample averages of $\{\Phi_{\mathcal{X}}(X_i)\}_{i=1}^n$ and $\{\Phi_{\mathcal{Y}}(Y_i)\}_{i=1}^n$, respectively. Moreover, if we denote by $\hat{c}_{\mathcal{X}c}$ and $\hat{c}_{\mathcal{Y}c}$ the empirical analogue of $c_{\mathcal{X}c}$ and $c_{\mathcal{Y}c}$, defined as in (\ref{2}), then
	$$\text{\normalfont HSIC}_n(X,Y) = \frac{1}{n^2}\sum\limits_{i,j=1}^n \hat{c}_{\mathcal{X}c}(X_i,X_j)\hat{c}_{\mathcal{Y}c}(Y_i,Y_j).$$
	Based on $\text{\normalfont HSIC}_n(X,Y)$, a test of independence (or mean independence) can be constructed by rejecting the corresponding null hypothesis whenever this statistic, with a proper choice of kernels $c_{\mathcal{X}}$ and $c_{\mathcal{Y}}$ according to the discussion in Section \ref{sec2}, is significantly different from $0$. Alternatively, an U-centering approach \citep{Szekely2014} could be followed to rely on an analogous, unbiased (in the i.i.d. case) estimator of $\text{\normalfont HSIC}(X,Y)$. The consistency and limit distribution of $\text{\normalfont HSIC}_n(X,Y)$ based on i.i.d. samples, as well as the distance covariance counterpart, are studied in \cite{Gretton2007}, \cite{Lyons2013} and \cite{Sejdinovic2013}, among others.
	We provide in this section analogous results under stationarity, ergodicity (for the test statistic consistency) and near epoch dependence on $\beta$-mixing underlying processes (to derive limit distributions), assuming $\mathcal{X}$ and $\mathcal{Y}$ to be standard Borel spaces (Polish spaces endowed with Borel $\sigma$-algebras). Thus, our dependence framework is more general than those previously considered in the literature and this is not a minor issue given that mixing conditions can fail to hold for simple autoregressive models \citep{Andrews1984}, whereas, for example, causal linear time series models are near epoch dependent under some mild moment existence assumptions (\citeauthor{Davidson2002}, \citeyear{Davidson2002}, Ch. 17).

	\subsection{Consistency and limit distribution}
	\label{sec3.1}
	
	Let $(\mathcal{X}\times \mathcal{Y})^\mathbb{Z}$ denote the space of sequences of $\mathcal{X}\times \mathcal{Y}$ valued random variables endowed with the usual product topology and product Borel $\sigma$-algebra, which is a standard Borel space.  Let $T:(\mathcal{X}\times \mathcal{Y})^\mathbb{Z} \rightarrow (\mathcal{X}\times \mathcal{Y})^\mathbb{Z}$ denote the shift function defined as $T(\{(x_i,y_i)\}_{i\in \mathbb{Z}}) = \{(x_{i+1},y_{i+1})\}_{i\in \mathbb{Z}}$ for all $\{(x_i,y_i)\}_{i\in \mathbb{Z}}$. We assume that $\{(X_i,Y_i)\}_{i\in \mathbb{Z}}\in (\mathcal{X}\times \mathcal{Y})^\mathbb{Z}$ is a strictly stationary process, that is, the finite dimensional distributions of the process are invariant to joint index shifts. We also require $\{(X_i,Y_i)\}_{i\in \mathbb{Z}}$ to be ergodic, i.e., the probability of any measurable set $B_{(\mathcal{X}\times \mathcal{Y})^\mathbb{Z}}$ such that $T^{-1}(B_{(\mathcal{X}\times \mathcal{Y})^\mathbb{Z}}) = B_{(\mathcal{X}\times \mathcal{Y})^\mathbb{Z}}$ is equal to either $0$ or $1$ (see, e.g., Chapter 1 in \citeauthor{Krengel1985}, \citeyear{Krengel1985}). Under this assumptions, we can prove 
	the consistency of $\text{\normalfont HSIC}_n(X,Y)$ as an estimator of $\text{\normalfont HSIC}(X,Y)$ under minimal moment and measurability conditions.
	
	\begin{theorem}\label{Th1}
		Let $\{(X_i,Y_i)\}_{i\in \mathbb{Z}}$ be strictly stationary and ergodic. In addition, let $\Phi_{\mathcal{X}}$ and $\Phi_{\mathcal{Y}}$ be measurable and assume that $\mathbb{E}[c_{\mathcal{X}}(X,X)^{1/2}]$, $\mathbb{E}[c_{\mathcal{Y}}(Y,Y)^{1/2}]$ and $\mathbb{E}[(c_{\mathcal{X}}(X,X)c_{\mathcal{Y}}(Y,Y))^{1/2}]$ exist and are finite. Then,
		$$ \text{\normalfont HSIC}_{n}(X,Y) \xrightarrow{a.s.} \text{\normalfont HSIC}(X,Y).$$
	\end{theorem}	
	
	\begin{remark}\label{remark2}
		If both $\Phi_{\mathcal{X}}$ and $\Phi_{\mathcal{Y}}$ are measurable, the tensor product $\Phi_{\mathcal{X}}\otimes \Phi_{\mathcal{Y}}$ is also measurable. Moreover, since $\mathcal{X}$ and $\mathcal{Y}$ are Polish, $\mathcal{C}_{\mathcal{X}}$, $\mathcal{C}_{\mathcal{Y}}$ and $\mathcal{HS}$ are separable \citep{Owhadi2016}. If we further assume that $\mathbb{E}[c_{\mathcal{X}}(X,X)^{1/2}]$, $\mathbb{E}[c_{\mathcal{Y}}(Y,Y)^{1/2}]$ and $\mathbb{E}[(c_{\mathcal{X}}(X,X)c_{\mathcal{Y}}(Y,Y))^{1/2}]$ are finite, it follows from Pettis's measurability theorem (see, e.g., Theorem 2 in \citeauthor{Diestel1977}, \citeyear{Diestel1977}) that $\Phi_{\mathcal{X}}(X)$, $\Phi_{\mathcal{Y}}(Y)$ and $\Phi_{\mathcal{X}}(X)\otimes \Phi_{\mathcal{Y}}(Y)$ are Bochner integrable. 
	\end{remark}
	
	By virtue of the continuous mapping theorem and the representation of $\text{\normalfont HSIC}_{n}(X,Y)$ as the squared norm of an empirical cross-covariance operator in (\ref{4}), it is clear that the asymptotic distribution of $\text{\normalfont HSIC}_{n}(X,Y)$ can be elegantly derived from a central limit theorem for dependent data in Hilbert spaces, such as Theorem 1.1 in \cite{Dehling2015}, which requires some finite moment restrictions on $\Phi_{\mathcal{X}c}(X)$ and $\Phi_{\mathcal{Y}c}(Y)$ and, moreover, the assumption that $\{\Phi_{\mathcal{X}c}(X_i)\otimes \Phi_{\mathcal{Y}c}(Y_i)\}_{i\in \mathbb{Z}}$ is $L^1$ near epoch dependent on some underlying strictly stationary process $\{\eta_{i}\}_{i\in \mathbb{Z}}$, taking values in some nonempty Polish space $E$, that is $\beta$-mixing with coefficients  $\{\beta_j\}_{j=0}^\infty$.
	\begin{definition}
		We say that a strictly stationary process $\{\eta_{i}\}_{i\in \mathbb{Z}} \in E^{\mathbb{Z}}$ is $\beta$-mixing (also known as absolutely regular), with coefficients $\{\beta_j\}_{j=0}^\infty$, if
		$$\beta_j = \mathbb{E}\!\left[\sup_{A\in \sigma(\{\eta_{k}\}_{k\ge j})}\!\left|\mathbb{P}(A|\sigma(\{\eta_{k}\}_{k\le 0})) - \mathbb{P}(A) \right|\right]\xrightarrow{j\rightarrow\infty} 0,$$
		where $\sigma(\{\eta_{k}\}_{k\ge j})$ and $\sigma(\{\eta_{k}\}_{k\le 0})$ denote the $\sigma$-algebras generated by $\{\eta_{k}\}_{k\ge j}$ and $\{\eta_{k}\}_{k\le 0}$.
	\end{definition}
	
	\begin{definition}\label{def4}
		Let $\mathcal{S}$ be a nonempty Polish space, let $d_{\mathcal{S}}$ be some metric that induces the topology on $\mathcal{S}$ and let $p>0$. We say that $\{S_{i}\}_{i\in \mathbb{Z}} \in \mathcal{S}^{\mathbb{Z}}$ is $L^p$ near epoch dependent ($L^p$ $\text{\normalfont NED}$) on $\{\eta_{i}\}_{i\in \mathbb{Z}} \in E^{\mathbb{Z}}$, with respect to $d_{\mathcal{S}}$, with coefficients $\{a_j\}_{j=0}^\infty$, if there is a measurable function $\Pi_{\mathcal{S}}:E^{\mathbb{Z}} \rightarrow \mathcal{S}$ such that $S_i = \Pi_{\mathcal{S}}(\{\eta_{i+j}\}_{j\in \mathbb{Z}})$ for all $i\in \mathbb{Z}$ and, moreover, there is a sequence of measurable functions $\{\Pi_{\mathcal{S}j}:E^{2j+1} \rightarrow \mathcal{X}\}_{j\in \mathbb{N}}$ verifying
		$$\mathbb{E}\big[d_{\mathcal{S}}(S_0, \Pi_{\mathcal{S}j}(\eta_{-j},\ldots, \eta_{j}))^p\big]^\frac{1}{p} \le a_j \xrightarrow{j\rightarrow\infty} 0.$$
		In addition, we say that $\{S_{i}\}_{i\in \mathbb{Z}}$ is near epoch dependent in probability ($\mathbb{P}$ $\text{\normalfont NED}$) on $\{\eta_{i}\}_{i\in \mathbb{Z}}$ if 
		$$d_{\mathcal{S}}(S_0, \Pi_{\mathcal{S}j}(\eta_{-j},\ldots, \eta_{j})) \xrightarrow{j\rightarrow \infty} 0,$$
		in probability.
	\end{definition}
	\begin{remark}\label{remark3}
		As some results in \cite{Dehling2016} imply, $L^p$ $\text{\normalfont NED}$ follows, under some finite moment conditions, from the weaker, and more general, notion of $\mathbb{P}$ \textrm{$\text{\normalfont NED}$} introduced in Definition \ref{def4}. Moreover, the choice of metric $d_{\mathcal{S}}$ is not relevant to establish the $\mathbb{P}$ $\text{\normalfont NED}$ property since inclusion functions between metric spaces endowed with topologically equivalent metrics are continuous and, therefore, convergence in probability is preserved by virtue of the continuous mapping theorem. What is not preserved when we change metrics, however, is the rate of convergence.
	\end{remark}
	
	We assume that $\{X_i\}_{i\in \mathbb{Z}}$ and $\{Y_i\}_{i\in \mathbb{Z}}$ are $\mathbb{P}$ $\text{\normalfont NED}$ on a $\beta$-mixing process $\{\eta_{i}\}_{i\in \mathbb{Z}}$ and that $\Phi_{\mathcal{X}}$ and $\Phi_{\mathcal{Y}}$ are continuous, so that $\{\Phi_{\mathcal{X}c}(X_i)\}_{i\in \mathbb{Z}}$, $\{\Phi_{\mathcal{Y}c}(Y_i)\}_{i\in \mathbb{Z}}$ and $\{\Phi_{\mathcal{X}c}(X_i)\otimes \Phi_{\mathcal{Y}c}(Y_i)\}_{i\in \mathbb{Z}}$ are also $\mathbb{P}$ $\text{\normalfont NED}$ on $\{\eta_{i}\}_{i\in \mathbb{Z}}$ as a consequence of the continuous mapping theorem. Furthermore, following Appendix A in \cite{Dehling2016}, we have
	$$ \inf\Big\{\epsilon\mid	\mathbb{P}\big(d_{\mathcal{X}}(X_0, \Pi_{\mathcal{X}j}(\eta_{-j},\ldots, \eta_{j})) >\epsilon\big)\le \epsilon, 	\mathbb{P}\big(d_{\mathcal{Y}}(Y_0, \Pi_{\mathcal{Y}j}(\eta_{-j},\ldots, \eta_{j})) >\epsilon\big)\le \epsilon\Big\} \xrightarrow{j\rightarrow \infty} 0,$$
	for any metrics $d_{\mathcal{X}}$ and $d_{\mathcal{Y}}$ that induce the topologies on $\mathcal{X}$ and $\mathcal{Y}$, respectively.
	This, in turn, implies that there always exist a sequence $\{\epsilon_{j}\}_{j=0}^\infty$, depending on $d_{\mathcal{X}}$ and $d_{\mathcal{Y}}$, such that $\epsilon_{j}\rightarrow 0$ as $j\rightarrow \infty$ and
	\begin{equation}\label{5}
		P_{j} = \max\Big\{\mathbb{P}\big(d_{\mathcal{X}}(X_0, \Pi_{\mathcal{X}j}(\eta_{-j},\ldots, \eta_{j})) >\epsilon_{j}\big), 	\mathbb{P}\big(d_{\mathcal{Y}}(Y_0, \Pi_{\mathcal{Y}j}(\eta_{-j},\ldots, \eta_{j})) >\epsilon_{j}\big)\Big\} \xrightarrow{j\rightarrow \infty} 0.
	\end{equation}
	In some cases, it is possible to quantify the rate of convergence to zero in equation (\ref{5}) for any sequence $\{\epsilon_{j}\}_{j=0}^\infty$.
	\begin{example}
		Let us assume that $\{X_i\}_{i\in \mathbb{Z}}$ and $\{Y_i\}_{i\in \mathbb{Z}}$ are also $L^p$ $\text{\normalfont NED}$ on $\{\eta_{i}\}_{i\in \mathbb{Z}}$, for some $p>0$, with respect to $d_{\mathcal{X}}$ and $d_{\mathcal{Y}}$ and with coefficients $\{a_{j}\}_{j=0}^\infty$, that is,
		$$a_{j} = \max\Big\{\mathbb{E}\big[ d_{\mathcal{X}}(X_0, \Pi_{\mathcal{X}j}(\eta_{-j},\ldots, \eta_{j}))^p\big]^\frac{1}{p}, \mathbb{E}\big[ d_{\mathcal{Y}}(Y_0, \Pi_{\mathcal{Y}j}(\eta_{-j},\ldots, \eta_{j}))^p\big]^\frac{1}{p}\Big\}\xrightarrow{j\rightarrow \infty} 0.$$
		Then, by virtue of Markov's inequality,
		$P_{j} \le a^p_{j}/\epsilon_{j}^{p}$.
	\end{example}
	As the following technical conditions imply, we require the existence of some metrics $d_{\mathcal{X}}$ and $d_{\mathcal{Y}}$ such that the corresponding sequences $\{\epsilon_{j}\}_{j=0}^\infty$ and $\{P_{j}\}_{j=0}^\infty$ converge to zero at an appropriate rate.
	\begin{enumerate}[
		label = (T\arabic*),
		labelindent=\parindent,
		leftmargin=*,
		align=left
		]
		\item Both $\Phi_{\mathcal{X}}$ and $\Phi_{\mathcal{Y}}$ are continuous. Moreover,
		let $d_{\Phi_{\mathcal{X}}}$ and $d_{\Phi_{\mathcal{Y}}}$ be the semimetrics induced by $c_{\mathcal{X}}$ and $c_{\mathcal{Y}}$, respectively, defined as in (\ref{3}). Then,
		$$\varrho_{}(\epsilon) = \max\!\left\{\mathbb{E}\!\left[\sup_{x\in \mathcal{X}: d_{\mathcal{X}}(X,x)\le \epsilon}d_{\Phi_{\mathcal{X}}}(X,x)\right]^\frac{1}{2}, \mathbb{E}\!\left[\sup_{y\in \mathcal{Y}:d_{\mathcal{Y}}(Y,y)\le \epsilon}d_{\Phi_{\mathcal{Y}}}(Y,y)\right]^\frac{1}{2}\right\}\xrightarrow{\epsilon\rightarrow 0} 0.$$
		\item There is some $\delta>0$ such that $\mathbb{E}[c_{\mathcal{X}}(X,X)^{(2+\delta)/2}]$, $\mathbb{E}[c_{\mathcal{Y}}(Y,Y)^{(2+\delta)/2}]$ and 
		$\mathbb{E}[(c_{\mathcal{X}}(X,X)c_{\mathcal{Y}}(Y,Y))^{(2+\delta)/2}]$ are finite. Moreover, $\sum_{j=0}^\infty \beta_j^{\delta/(2+\delta)}< \infty$, $\sum_{j=0}^\infty P_{j}^{\delta/(2+\delta)} < \infty$ and $\sum_{j=0}^\infty \varrho^{\delta/(1+\delta)}(\epsilon_{j}) < \infty$.
	\end{enumerate}
	
	Condition (T1) holds if $\Phi_{\mathcal{X}}$ and $\Phi_{\mathcal{Y}}$ are uniformly continuous as functions defined on the metric spaces $(\mathcal{X},d_{\mathcal{X}})$ and $(\mathcal{Y},d_{\mathcal{Y}})$, respectively. If $\mathcal{X}$ and $\mathcal{Y}$ are Hilbert spaces and $d_{\mathcal{X}}$ and $d_{\mathcal{Y}}$ are the corresponding usual metrics, (T1) is verified by the most prevalent kernels in the literature, such as the distance-induced kernel (with $\varrho_{}(\epsilon) = \epsilon^{1/2}$) and the Gaussian kernel (with $\varrho_{}(\epsilon) \le \sqrt{2}\epsilon$). Condition (T2) connects the kernel regularity imposed in (T1) with the $\mathbb{P}$ $\text{\normalfont NED}$ assumptions on $\{X_i\}_{i\in \mathbb{Z}}$ and $\{Y_i\}_{i\in \mathbb{Z}}$. The limit distribution of $\text{\normalfont HSIC}_{n}(X,Y)$ is described in the following result.
	\begin{theorem}\label{Th2}
		Let $\{X_i\}_{i\in \mathbb{Z}}$ and $\{Y_i\}_{i\in \mathbb{Z}}$ be $\mathbb{P}$ $\text{\normalfont NED}$ on $\{\eta_{i}\}_{i\in \mathbb{Z}}$ and suppose (T1) and (T2) hold. Let $N_{\mathcal{HS}}$ denote a centred Gaussian random variable taking values in $\mathcal{HS}$ with covariance operator $\Gamma$ verifying
		$$\langle s, \Gamma (s')\rangle_{\mathcal{HS}} = \sum\limits_{i\in \mathbb{Z}} \mathbb{E}\!\left[\langle S_{0} , s \rangle_{\mathcal{HS}} \langle S_{i}, s' \rangle_{\mathcal{HS}}  \right],$$
		for all $s,s' \in \mathcal{HS}$ and $S_i =  \Phi_{\mathcal{X}c}(X_i)\otimes\Phi_{\mathcal{Y}c}(Y_i) - \mathbb{E}[\Phi_{\mathcal{X}c}(X)\otimes\Phi_{\mathcal{Y}c}(Y)]$ for all $i\in \mathbb{Z}$. If $\text{\normalfont HSIC}(X,Y) = 0$, we have
		$$ n\text{\normalfont HSIC}_{n}(X,Y) \xrightarrow{\mathcal{D}} \left|\left| N_{\mathcal{HS}}\right|\right|^2_{\mathcal{HS}},$$
		where $\xrightarrow{\mathcal{D}}$ denotes convergence in distribution. If $\text{\normalfont HSIC}(X,Y) > 0$, it holds
		$$ \sqrt{n}(\text{\normalfont HSIC}_{n}(X,Y) - \text{\normalfont HSIC}(X,Y)) \xrightarrow{\mathcal{D}} N(0,4\sigma^2_{\mathcal{HS}}), $$
		where $\sigma^2_{\mathcal{HS}} =  \langle\mathbb{E}[\Phi_{\mathcal{X}c}(X)\otimes\Phi_{\mathcal{Y}c}(Y)],\Gamma(\mathbb{E}[\Phi_{\mathcal{X}c}(X)\otimes\Phi_{\mathcal{Y}c}(Y)])\rangle_{\mathcal{HS}}$.
	\end{theorem}

	Let $\{\lambda_j\}_{j=1}^\infty$ denote the eigenvalues of the covariance operator of the random variable $N_{\mathcal{HS}}$ defined in Theorem \ref{Th2}. We have the following equivalent representation of the limit distribution when $\text{\normalfont HSIC}(X,Y) = 0$:
	$$\left|\left| N_{\mathcal{HS}}\right|\right|^2_{\mathcal{HS}} = \sum\limits_{j=1}^\infty \lambda_j |N_j|^2,$$
	where $\{N_j\}_{j\in \mathbb{Z}}$ is a collection of i.i.d. real valued standard Gaussian random variables. This limit distribution trivially coincides with the one obtained for i.i.d. data if we further assume that, for all $i>0$, $\mathbb{E}[\Phi_{\mathcal{Y}}(Y_i)| X_i, X_0, Y_0] = \mathbb{E}[\Phi_{\mathcal{Y}}(Y)]$ a.s., which holds if the data is serially independent and $\text{\normalfont HSIC}(X,Y) = 0$. 
	
	\begin{remark}
		The proof of Theorem \ref{Th2} can be easily extended to any type of dependence structure provided there is an associated central limit theorem for Hilbert space valued random variables. See, e.g., \cite{Chen1998} for the case of $L^2$ $\text{\normalfont NED}$ assumptions on strong mixing underlying processes.
	\end{remark}
	
	\subsection{Local properties}
	\label{sec3.2}
	
	Let us study the asymptotic effect on $\text{\normalfont HSIC}_n(X,Y)$ caused by small perturbations on the process $\{(X_i,Y_i)\}_{i\in \mathbb{Z}}$. To do so, we assume:
	\begin{enumerate}[
		label = (T\arabic*),
		labelindent=\parindent,
		leftmargin=*,
		align=left
		]
		\setcounter{enumi}{2}
		\item We have a triangular array $\{(X_{in},Y_{in})\}_{1\le i \le n}^{n\in \mathbb{N}^+}$ such that, for all $n\in  \mathbb{N}^+$, $(X_{in},Y_{in})$ is identically distributed to some pair $(X_{\cdot n},Y_{\cdot n})\in \mathcal{X}\times \mathcal{Y}$ for all $i\in\{1,\ldots,n\}$ and $(X_{in},Y_{in})$ is a measurable function of $\{\eta_{i+j}\}_{j\in \mathbb{Z}}$ for all $i,n$. Furthermore, it holds that
		$$
		\begin{cases}
			\Phi_{\mathcal{X}}(X_{in}) = 	\Phi_{\mathcal{X}}(X_{i}) + Z_{\mathcal{X}i}/k_n+ R_{\mathcal{X}in}\ \ \text{a.s},& \\
			\Phi_{\mathcal{Y}}(Y_{in}) = \Phi_{\mathcal{Y}}(Y_i) + Z_{\mathcal{Y}i}/k_n+ R_{\mathcal{Y}in}\ \ \text{a.s},&
		\end{cases}$$
		where $k_n\rightarrow\infty$ as $n\rightarrow\infty$ and $Z_{\mathcal{Y}i},Z_{\mathcal{X}i}$ (which do not depend on $n$) and $R_{\mathcal{Y}in},R_{\mathcal{X}in}$ are all measurable functions of $\{\eta_{i+j}\}_{j\in \mathbb{Z}}$ for all $i,n$.
		\item $\mathbb{E}[||Z_{\mathcal{X}i}||_{\mathcal{C}_\mathcal{X}}]$,  $\mathbb{E}[||Z_{\mathcal{Y}i}||_{\mathcal{C}_\mathcal{Y}}]$, $\mathbb{E}[||Z_{\mathcal{X}i}||_{\mathcal{C}_\mathcal{X}}||Z_{\mathcal{Y}i}||_{\mathcal{C}_\mathcal{Y}}]$, $\mathbb{E}[c_{\mathcal{X}}(X_{i},X_{i})^{1/2}||Z_{\mathcal{Y}i}||_{\mathcal{C}_\mathcal{Y}}]$ and $\mathbb{E}[c_{\mathcal{Y}}(Y_{i},Y_{i})^{1/2}||Z_{\mathcal{X}i}||_{\mathcal{C}_\mathcal{X}}]$ are finite and, in addition, $\mathbb{E}[||R_{\mathcal{X}in}||_{\mathcal{C}_\mathcal{X}}]$, $\mathbb{E}[||R_{\mathcal{Y}in}||_{\mathcal{C}_\mathcal{Y}}]$,   $\mathbb{E}[||R_{\mathcal{X}in}||_{\mathcal{C}_\mathcal{X}}||R_{\mathcal{Y}in}||_{\mathcal{C}_\mathcal{Y}}]$, $\mathbb{E}[||Z_{\mathcal{X}i}||_{\mathcal{C}_\mathcal{X}}||R_{\mathcal{Y}in}||_{\mathcal{C}_\mathcal{Y}}]$, $\mathbb{E}[||R_{\mathcal{X}in}||_{\mathcal{C}_\mathcal{X}}||Z_{\mathcal{Y}i}||_{\mathcal{C}_\mathcal{Y}}]$,
		$\mathbb{E}[c_{\mathcal{X}}(X_{i},X_{i})^{1/2}||R_{\mathcal{Y}in}||_{\mathcal{C}_\mathcal{Y}}]$ and $\mathbb{E}[c_{\mathcal{Y}}(Y_{i},Y_{i})^{1/2}||R_{\mathcal{X}in}||_{\mathcal{C}_\mathcal{X}}]$ are all $o(1/k_n)$ as $n\rightarrow\infty$.
	\end{enumerate}
	
	Condition (T3) generalises the usual local alternative hypothesis in Pitman's sense studied in \cite{Lee2020} and \cite{Lai2021} in the context of mean independence tests. Condition (T4) is required to apply ergodic theorems and bound some asymptotically neglibible terms in $\text{\normalfont HSIC}_{n}(X_{\cdot n},Y_{\cdot n})$. Under these conditions, $Z_{\mathcal{X}i}$ and $Z_{\mathcal{Y}i}$ are the limits in $L^1$ sense of $\{k_n(	\Phi_{\mathcal{Y}}(Y_{in}) - \Phi_{\mathcal{Y}}(Y_i))\}_{n\in \mathbb{N}^+}$ and $\{k_n(	\Phi_{\mathcal{X}}(X_{in}) - \Phi_{\mathcal{X}}(X_i))\}_{n\in \mathbb{N}^+}$, respectively, which further implies that $\{\Phi_{\mathcal{Y}}(Y_{in})\}_{n\in \mathbb{N}^+}$ and $\{\Phi_{\mathcal{X}}(X_{in})\}_{n\in \mathbb{N}^+}$ converge in $L^1$ norm to $\Phi_{\mathcal{Y}}(Y_i)$ and $\Phi_{\mathcal{X}}(X_i)$ at rate $k_n$. The next theorem describes the limit distribution of $\text{\normalfont HSIC}_{n}(X_{\cdot n},Y_{\cdot n})$ in terms of the convergence rate $k_n$.
	
	\begin{theorem}\label{Th3}
		Let $\{X_i\}_{i\in \mathbb{Z}}$ and $\{Y_i\}_{i\in \mathbb{Z}}$ be $\mathbb{P}$ $\text{\normalfont NED}$ on $\{\eta_{i}\}_{i\in \mathbb{Z}}$ and suppose (T1)-(T4) hold. In addition, let $N_{\mathcal{HS}}$ and $\sigma^2_{\mathcal{HS}}$ be defined as in Theorem \ref{Th2}, let $0\le M<\infty$ and let  $\mu = \mathbb{E}[\Phi_{\mathcal{X}c}(X_i)\otimes Z_{\mathcal{Y}i}] + \mathbb{E}[Z_{\mathcal{X}i}\otimes \Phi_{\mathcal{Y}c}(Y_i)]$.
		Then, assuming $\text{\normalfont HSIC}(X,Y) = 0$, we have:
		\begin{enumerate}[
			label = (\roman*),
			labelindent=\parindent,
			leftmargin=*,
			labelsep= 0em,
			align=left
			]
			\item If $\sqrt{n}/k_n \rightarrow M$, $ n\text{\normalfont HSIC}_{n}(X_{\cdot n},Y_{\cdot n}) \xrightarrow{\mathcal{D}} || M\mu+N_{\mathcal{HS}}||^2_{\mathcal{HS}}$,
			\item If $\sqrt{n}/k_n\rightarrow \infty$, $ n\text{\normalfont HSIC}_{n}(X_{\cdot n},Y_{\cdot n}) = n||\mu ||^2_{\mathcal{HS}}/k_n^2  + o_{\mathbb{P}}(n/k_n^2)$.
		\end{enumerate}
		Furthermore, denoting by $\theta = \langle \mu, \mathbb{E}[\Phi_{\mathcal{X}c}(X)\otimes \Phi_{\mathcal{Y}c}(Y)]\rangle_{\mathcal{HS}}$ and assuming $\text{\normalfont HSIC}(X,Y)>0$, it holds: 
		\begin{enumerate}[
			label = (\roman*),
			labelindent=\parindent,
			leftmargin=*,
			labelsep= 0em,
			align=left
			]
			\setcounter{enumi}{2}
			\item If $\sqrt{n}/k_n \rightarrow M$, $ \sqrt{n}(\text{\normalfont HSIC}_{n}(X_{\cdot n},Y_{\cdot n}) - \text{\normalfont HSIC}(X,Y)) \xrightarrow{\mathcal{D}} N(2M\theta,4\sigma^2_{\mathcal{HS}})$,
			\item If $\sqrt{n}/k_n\rightarrow \infty$, $ \sqrt{n}(\text{\normalfont HSIC}_{n}(X_{\cdot n},Y_{\cdot n}) - \text{\normalfont HSIC}(X,Y)) = \sqrt{n}2\theta/k_n  + o_{\mathbb{P}}(\sqrt{n}/k_n)$.
		\end{enumerate}
	\end{theorem}	
	
	From Theorem \ref{Th3} we can derive two fundamental conclusions. On the one hand, $\text{\normalfont HSIC}_{n}(X_{\cdot n},Y_{\cdot n})$ is still consistent, in probability, as an estimator of $\text{\normalfont HSIC}(X,Y)$. If we further assume that $\{k_n(||R_{\mathcal{X}\cdot n}||+||R_{\mathcal{Y}\cdot n}||)\}_{n\in \mathbb{N^+}}$ is majorized by some integrable random variable independent of $n$, consistency can be strengthened to hold a.s. On the other hand, the noncentral chi squared distribution is stochastically increasing in the noncentrality parameter. Therefore, if we compare the limit distributions when $\text{\normalfont HSIC}(X,Y)=0$ in Theorems \ref{Th2} and \ref{Th3}, it is clear that a test based on $\text{\normalfont HSIC}_{n}(X_{\cdot n},Y_{\cdot n})$ will asymptotically reject the null hypothesis of $\text{\normalfont HSIC}(X,Y)=0$, with non-trivial power, provided that $\mu \neq 0$ and $k_n = O(\sqrt{n})$.

	\subsection{Bootstrap calibration}
	\label{sec3.3}
	
	The distribution of $||N_{\mathcal{HS}}||^2_{\mathcal{HS}}$, defined as in Theorem \ref{Th2}, is non-pivotal. To calibrate a test of independence or mean independence based on the empirical estimator of the HSIC, in the i.i.d. sample case, there are many proposals such as gamma distribution-based approximations \citep{Zhang2024}, estimating the eigenvalues of the covariance operator of $N_{\mathcal{HS}}$ and approximating the quantiles of sampling distributions by simulation \citep{Sejdinovic2013}, resampling schemes \citep{Sejdinovic2013} and permutation approaches \citep{Gretton2007}. However, the gamma-based approximation is not consistent for all significance levels and the resampling schemes usually fail for mean independence tests.
	
	For dependent data, the proposal in \cite{Chwialkowski2014}, which is only valid to calibrate independence tests, is based on shifting the indexes of $\{Y_{i}\}_{i=1}^n$ while keeping those of $\{X_{i}\}_{i=1}^n$ fixed, whereas \cite{Chwialkowski2014a} proved the consistency of wild bootstrap schemes to calibrate kernel-based tests for Euclidean data under $\tau$-dependence. For the equivalent distance covariance-based tests, authors also rely on bootstrap schemes and blockwise permutations \citep{Fokianos2018,Meintanis2022}. In our $\mathbb{P}$ $\text{\normalfont NED}$ sample setting, we suggest to consider the following wild bootstrap approach:
	
	\begin{enumerate}[
		label = Step \arabic*.,
		labelindent=\parindent,
		leftmargin=*,
		align=left
		]
		\item Draw $n$ real valued Gaussian random variables $\{r_{in}\}_{i=1}^{n}$, independent of the data (and also independent of $\{\eta_{i}\}_{i\in \mathbb{Z}}$), such that $\mathbb{E}[r_{in}] = 0$ and $\text{Cov}(r_{in},r_{jn}) = \rho_n(|i-j|)$, where $\rho_n:\mathbb{N} \rightarrow [0,1]$, $\rho_n(0) = 1$, $\rho_n(k)\rightarrow 1$ as $n \rightarrow\infty$ for all $k\ge1$ and $\sum_{k=1}^n \rho_n(k) = O(\ell_n)$ for some $\{\ell_n\}_{n\in \mathbb{N}^+}$ such that $\ell_n \rightarrow \infty$ and $\ell_n =o(n)$.
		\item Compute the statistic $\text{\normalfont HSIC}^*_{n}(X,Y)$ defined as
		$$\text{\normalfont HSIC}^*_{n}(X,Y) = \frac{1}{n^2}\sum_{i,j=1}^{n}(r_{in}-\bar{r}_{\cdot n}) (r_{jn}-\bar{r}_{\cdot n}) \hat{c}_{\mathcal{X}c}(X_i,X_j)\hat{c}_{\mathcal{Y}c}(Y_i,Y_j),$$
		where $\bar{r}_{\cdot n}$ denotes the sample average of $\{r_{in}\}_{i=1}^{n}$. 
		\item Repeat the previous two steps $n_B$ times to obtain $n_B$ conditionally independent copies of $\text{\normalfont HSIC}^*_{n}(X,Y)$. For any $\alpha \in (0,1)$, compute the $(1-\alpha)$th sample quantile of the conditional distribution of  $n\text{\normalfont HSIC}^*_{n}(X,Y)$, which we denote by $Q_{n,n_B}^*(1-\alpha)$.
		\item If $n\text{\normalfont HSIC}_{n}(X,Y)\ge Q_{n,n_B}^*(1-\alpha)$, reject the null hypothesis of $\text{\normalfont HSIC}(X,Y) = 0$ at significance level $\alpha$.
	\end{enumerate}
	
	An example of valid data generating process for the bootstrap multipliers $\{r_{in}\}_{i=1}^{n}$, which is our suggestion to implement the test in practice, is a scaled moving average process given by
	$$ r_{in} = \sum_{k=1}^{\ell_n}\frac{\Delta_{k\ell_n}}{\sqrt{\sum_{k=1}^{\ell_n}\Delta^2_{k\ell_n}}}\varepsilon_{i-k+1}  ,$$
	where $\{\varepsilon_i\}_{i\in \mathbb{Z}}$ is a collection of i.i.d. standard Gaussian random variables and $\Delta_{k\ell_n} = 0.5-|(k-0.5)/\ell_n-0.5|$  (see the comparison of several data generating processes for the multipliers in \citeauthor{Doukhan2015}, \citeyear{Doukhan2015}). Before stating the consistency results and the required technical conditions, let us introduce the notion of weak convergence in probability. 
	
	\begin{definition}\label{Def5}
		Let $\{S_n\}_{n\in \mathbb{N}^+}$ be a sequence of measurable functions of $(\{r_{in}\}_{i=1}^{n},\{\eta_{i}\}_{i\in \mathbb{Z}})$ taking values in some nonempty Polish space $\mathcal{S}$. We say that $S_n$ converges weakly to $S\in \mathcal{S}$ in probability $\mathbb{P}$, and we denote it by
		$$S_n \xrightarrow{\mathcal{D}^*} S \ \text{in }\mathbb{P},$$
		if the conditional probability law of $S_n$ given $\{\eta_{i}\}_{i\in \mathbb{Z}}$ (which can be defined by means of regular conditional probabilities since $\{\eta_{i}\}_{i\in \mathbb{Z}}$ takes values in a Polish space) converges to the law of $S$ in probability. Equivalently, for any subsequence of the natural numbers $\{n_k\}_{k\in \mathbb{N^+}}$ there is a further subsequence $\{n_{k_l}\}_{l\in \mathbb{N^+}}$ such that the conditional probability law of $S_{n_{k_l}}$ given $\{\eta_{i}\}_{i\in \mathbb{Z}}$ converges (as $l\rightarrow \infty$) to the law of $S$, for almost all realizations of $\{\eta_{i}\}_{i\in \mathbb{Z}}$.
	\end{definition}
	
	According to the following results, which required the derivation of a new bootstrap consistency theorem for the sample mean of Hilbert space valued random variables under $L^1$ $\text{\normalfont NED}$ assumptions (see Theorem \ref{Th11} in \ref{app2}), this bootstrap scheme is asymptotically valid, in the sense implied by Definition \ref{Def5}, for any test based on the HSIC under some additional, mild technical conditions:
	
	\begin{enumerate}[
		label = (T\arabic*),
		labelindent=\parindent,
		leftmargin=*,
		align=left
		]
		\setcounter{enumi}{4}
		\item $\sum_{j=0}^\infty jP^{(1+\delta)/(2+\delta)}_{j} < \infty$, $\sum_{j=0}^\infty j\varrho_{}(\epsilon_{j}) < \infty$ and $\sum_{j=0}^\infty j\beta_j< \infty$.
		\item $\mathbb{E}[||Z_{\mathcal{X}i}||^2_{\mathcal{C}_\mathcal{X}}]$,  $\mathbb{E}[||Z_{\mathcal{Y}i}||^2_{\mathcal{C}_\mathcal{Y}}]$, $\mathbb{E}[||Z_{\mathcal{X}i}||^2_{\mathcal{C}_\mathcal{X}}||Z_{\mathcal{Y}i}||^2_{\mathcal{C}_\mathcal{Y}}]$, $\mathbb{E}[c_{\mathcal{X}}(X_{i},X_{i})||Z_{\mathcal{Y}i}||^2_{\mathcal{C}_\mathcal{Y}}]$ and $\mathbb{E}[c_{\mathcal{Y}}(Y_{i},Y_{i})||Z_{\mathcal{X}i}||^2_{\mathcal{C}_\mathcal{X}}]$ are finite and, in addition, $\mathbb{E}[||R_{\mathcal{X}in}||^2_{\mathcal{C}_\mathcal{X}}]$, $\mathbb{E}[||R_{\mathcal{Y}in}||^2_{\mathcal{C}_\mathcal{Y}}]$,  $\mathbb{E}[c_{\mathcal{X}}(X_{i},X_{i})||R_{\mathcal{Y}in}||^2_{\mathcal{C}_\mathcal{Y}}]$, $\mathbb{E}[c_{\mathcal{Y}}(Y_{i},Y_{i})||R_{\mathcal{X}in}||^2_{\mathcal{C}_\mathcal{X}}]$, $\mathbb{E}[||R_{\mathcal{X}in}||^2_{\mathcal{C}_\mathcal{X}}||R_{\mathcal{Y}in}||^2_{\mathcal{C}_\mathcal{Y}}]$, $\mathbb{E}[||Z_{\mathcal{X}i}||^2_{\mathcal{C}_\mathcal{X}}||R_{\mathcal{Y}in}||^2_{\mathcal{C}_\mathcal{Y}}]$ and $\mathbb{E}[||R_{\mathcal{X}in}||^2_{\mathcal{C}_\mathcal{X}}||Z_{\mathcal{Y}i}||^2_{\mathcal{C}_\mathcal{Y}}]$ are all $o(1/k^2_n)$ as $n\rightarrow\infty$.
	\end{enumerate}
	
	Condition (T5) establishes additional restrictions on the decay rates of $\{\epsilon_{j}\}_{j=0}^\infty$, $\{P_{j}\}_{j=0}^\infty$ and $\{\beta_{j}\}_{j=0}^\infty$.
	Moreover, condition (T6) imposes bounds on second order moments of $\{\Phi_{\mathcal{X}}(X_{\cdot n}) - \Phi_{\mathcal{X}}(X)\}_{n\in \mathbb{N}^+}$ and $\{\Phi_{\mathcal{Y}}(Y_{\cdot n}) - \Phi_{\mathcal{Y}}(Y)\}_{n\in \mathbb{N}^+}$. 
	
	\begin{theorem}\label{Th4}
		Let $\{X_i\}_{i\in \mathbb{Z}}$ and $\{Y_i\}_{i\in \mathbb{Z}}$ be $\mathbb{P}$ $\text{\normalfont NED}$ on $\{\eta_{i}\}_{i\in \mathbb{Z}}$. If (T1), (T2) and (T5) hold, we have
		$$ n\text{\normalfont HSIC}^*_{n}(X,Y) \xrightarrow{\mathcal{D}^*} || N_{\mathcal{HS}}||^2_{\mathcal{HS}} \ \text{in }\mathbb{P},$$
		where $N_{\mathcal{HS}}$ is defined as in Theorem \ref{Th2}. In addition, let $Q_n^*(1-\alpha)$ denote the $(1-\alpha)$th quantile of the conditional distribution of  $n\text{\normalfont HSIC}^*_{n}(X,Y)$ given $\{\eta_{i}\}_{i\in \mathbb{Z}}$ and assume that the covariance operator of $N_{\mathcal{HS}}$ has one or more positive eigenvalues. Then:
		\begin{enumerate}[
			label = (\roman*),
			labelindent=\parindent,
			leftmargin=*,
			labelsep= 0em,
			align=left
			]
			\item If $\text{\normalfont HSIC}_{n}(X,Y) = 0$, $\mathbb{P}(n\text{\normalfont HSIC}_{n}(X,Y)>Q_n^*(1-\alpha)) \rightarrow \alpha$,
			\item If $\text{\normalfont HSIC}_{n}(X,Y)>  0$, $\mathbb{P}(n\text{\normalfont HSIC}_{n}(X,Y)>Q_n^*(1-\alpha)) \rightarrow 1$.
		\end{enumerate}
	\end{theorem}

	\begin{theorem}\label{Th5}
		Let $\{X_i\}_{i\in \mathbb{Z}}$ and $\{Y_i\}_{i\in \mathbb{Z}}$ be $\mathbb{P}$ $\text{\normalfont NED}$ on $\{\eta_{i}\}_{i\in \mathbb{Z}}$ and assume (T1)-(T6) hold.
		Then, if $\ell_n/k^2_n \rightarrow 0$ (which holds if $\sqrt{n}/k_n = O(1)$), we have 
		$$ n\text{\normalfont HSIC}^*_{n}(X_{\cdot n},Y_{\cdot n}) \xrightarrow{\mathcal{D}^*} || N_{\mathcal{HS}}||^2_{\mathcal{HS}} \ \text{in } \mathbb{P},$$
		where $N_{\mathcal{HS}}$ be defined as in Theorem \ref{Th2}. In addition, let $Q_n^*(1-\alpha)$ denote the $(1-\alpha)$th quantile of the conditional distribution of $n\text{\normalfont HSIC}^*_{n}(X_{\cdot n},Y_{\cdot n})$ given $\{\eta_{i}\}_{i\in \mathbb{Z}}$. If $\text{\normalfont HSIC}(X,Y) >0$, it holds that
		$$\mathbb{P}(n\text{\normalfont HSIC}_{n}(X_{\cdot n},Y_{\cdot n})>Q_n^*(1-\alpha)) \rightarrow 1,$$
		independently of $k_n$. Moreover, let $F_{M}$ denote the cumulative distribution function of $||M\mu + N_{\mathcal{HS}}||^2_{\mathcal{HS}}$ where  
		$\mu$ and $M$ are defined as in Theorem \ref{Th3} and let $Q(1-\alpha)$ denote the $(1-\alpha)$th quantile of $F_{0}$. Suppose that the covariance operator of $N_{\mathcal{HS}}$ has one or more positive eigenvalues, $\mu\neq 0$ and $\text{\normalfont HSIC}(X,Y) =0$. Then:
		\begin{enumerate}[
			label = (\roman*),
			labelindent=\parindent,
			leftmargin=*,
			labelsep= 0em,
			align=left
			]
			\item If $\sqrt{n}/k_n\rightarrow M$, $\mathbb{P}(n\text{\normalfont HSIC}_{n}(X_{\cdot n},Y_{\cdot n})>Q_n^*(1-\alpha)) \rightarrow 1-F_{M}(Q(1-\alpha))$, 
			\item If $\sqrt{n}/k_n\rightarrow\infty$, $\hspace{0.1mm}\mathbb{P}(n\text{\normalfont HSIC}_{n}(X_{\cdot n},Y_{\cdot n})>Q_n^*(1-\alpha)) \rightarrow 1$.
		\end{enumerate}
		Furthermore, $F_{M}$ is absolutely continuous for all $M$ and $F_0(Q(1-\alpha)) = 1-\alpha$.
	\end{theorem}

	\section{A simulation study with a critical perspective}
	\label{sec4}
	
	To illustrate the finite sample performance of the tests of $\text{\normalfont HSIC}(X,Y) = 0$, calibrated with the bootstrap scheme presented in Section \ref{sec3.3}, as well as the inherent difficulties that arise when 
	we combine small and moderate sample sizes with non-Euclidean data and dependent samples, 
	we describe in this section the results of a succint simulation study based on functional data, where $\mathcal{Y} = \mathcal{X} = L^2(0,1)$, considering four different combinations of kernels:
	\begin{itemize}
		\item  MDD: $c_{\mathcal{X}}(x,x') = ||x||_{\mathcal{X}} + ||x'||_{\mathcal{X}} - ||x-x'||_{\mathcal{X}}$ and $c_{\mathcal{Y}}(y,y') = \langle y,y' \rangle_{\mathcal{Y}}$,
		\item  $\text{KCMD}_{\text{G}}$: $c_{\mathcal{X}}(x,x') = \exp(-||x-x'||_{\mathcal{X}}^2/\sigma^2_{c_{\mathcal{X}}})$ and $c_{\mathcal{Y}}(y,y') = \langle y,y' \rangle_{\mathcal{Y}}$,
		\item  DCOV: $c_{\mathcal{X}}(x,x') = ||x||_{\mathcal{X}} + ||x'||_{\mathcal{X}} - ||x-x'||_{\mathcal{X}}$ and $c_{\mathcal{Y}}(y,y') = ||y||_{\mathcal{Y}} + ||y'||_{\mathcal{Y}} - ||y-y'||_{\mathcal{Y}}$,
		\item $\text{\normalfont HSIC}_{\text{G}}$: $c_{\mathcal{X}}(x,x') = \exp(-||x-x'||_{\mathcal{X}}^2/\sigma^2_{c_{\mathcal{X}}})$ and $c_{\mathcal{Y}}(y,y') = \exp(-||y-y'||_{\mathcal{Y}}^2/\sigma^2_{c_{\mathcal{Y}}})$,
	\end{itemize}
	where the Gaussian kernel parameters $\sigma^2_{c_{\mathcal{X}}}$ and $\sigma^2_{c_{\mathcal{Y}}}$ are set to the median of the in-sample squared distances between pairwise distinct observations of $X$ and $Y$, respectively (median heuristic). The combinations in MDD and $\text{KCMD}_{\text{G}}$ characterize whether $Y$ is mean independent of $X$, whereas those in DCOV and $\text{\normalfont HSIC}_{\text{G}}$ characterize independence. Moreover, the choice of kernels in MDD and DCOV allows us to represent the distance-based family of tests since the corresponding values of $\text{\normalfont HSIC}_{n}(X,Y)$ coincide with the empirical estimators of the martingale difference divergence \citep{Lee2020} and the usual version of the distance covariance \citep{Lyons2013}, respectively. The relationship between kernel-based and distance-based methods is thoroughly studied in \cite{Sejdinovic2013}.
	
	The data is discretized in practice on an equally spaced grid with 1001 points to approximate the corresponding inner products and norms using the trapezoidal rule. Moreover, for all the iterative data generating processes we discuss in this section, we initialize the data at zero and discard the first 100 curves. The empirical rejection rates of the tests, calibrated with nominal level $\alpha = 0.05$, are computed  with $1000$ Monte Carlo replications, $1000$ bootstrap resamples and sample sizes of $n\in \{100,250,1000\}$.
	
	\subsection{Level sensibility}
	\label{sec4.1}
	
	To assess the level sensibility of the tests to the dependence structure and the choice of $\ell_n$, we rely on the data generating process described in Section 4.1 in \cite{Meintanis2022}, where $\{X_i\}_{i = 1}^n$ and $\{Y_i\}_{i = 1}^n$ are i.i.d. samples of a first order Hilbertian autoregressive process $\{S_i\}_{i\in \mathbb{Z}}$ which verifies
	$$S_{i}(t') = \int_0^1 \gamma_1\min\{t,t'\}S_{i-1}(t)dt + \mathcal{E}_i(t'),$$
	where $\{\mathcal{E}_i\}_{i\in \mathbb{Z}}$ is an i.i.d. collection of standard Wiener processes and  $\gamma_1 \in \{0,0.75,1.5,2.25\}$. The integral operator with kernel $\gamma_1 \min\{t,t'\}$ has operator norm $4\gamma_1/\pi^2$ which, for our choice of $\gamma_1$, takes, approximately, the values $0,0.3,0.61$ and $0.91$, the last of which lies close to the stationarity boundary defined by the unitary norm of the autoregressive operator. The empirical rejection rates are displayed in Table \ref{tab1}, where $\ell_n$ is set to either $1$ (independent bootstrap multipliers) or the closest integer to $2n^{1/5}$, $5n^{1/5}$ and $10n^{1/5}$. From these results we can conclude that, for small and moderate sample sizes ($n\le 250$), the test calibration becomes very imprecise as the data dependence increases. Moreover, the choice of $\ell_n$, which has a noticeable impact in the Type I error of the tests, should account for the sample dependence structure: strong dependence requires larger values of $\ell_n$ to preserve significance levels, although excessively large values of $\ell_n$ yield tests that are overly conservative. For larger sample sizes ($n=1000$), the calibration becomes less sensitive to the selection of $\ell_n$, although in the case of $\gamma_1 = 2.25$ some of the tests require a value of $\ell_n$ even greater than $10n^{1/5}$ to be properly calibrated. Interestingly, independence tests appear to be more conservative than mean dependence tests, which suggest that Gaussian wild bootstrap schemes may not be the best calibration methods for the former. However, alternative procedures such as those proposed in \cite{Chwialkowski2014} and \cite{Meintanis2022} require the additional assumption of independence between $\{Y_i\}_{i\in \mathbb{Z}}$ and  $\{X_i\}_{i\in \mathbb{Z}}$ to be consistent, which is much stronger than the (mean) independence of $Y_i$ and $X_i$ for all $i\in \mathbb{Z}$.
	
	\begin{table}[t]
		\caption{Empirical rejection rates, approximated with $1000$ Monte Carlo replications, of the MMD, $\text{KCMD}_{\text{G}}$, DCOV and $\text{\normalfont HSIC}_{\text{G}}$ tests of null $\text{\normalfont HSIC}(X,Y)$ for the simulation scenario of Section \ref{sec4.1}, under which $X$ and $Y$ are i.i.d. We consider different combinations of the sample size ($n\in \{100,250,1000\}$), the bootstrap bandwidth ($\ell_n\in \{1,2n^{1/5},5n^{1/5},10n^{1/5}\}$) and the parameter controlling the norm of the autoregression operator defining $X$ and $Y$ ($\gamma_1\in\{0,0.75,1.5,2.25\}$). All tests are calibrated with nominal level $\alpha= 0.05$ using $B=1000$ bootstrap resamples.}
		\centering
		\smallskip
		\setlength{\tabcolsep}{4.5pt}
		\renewcommand{\arraystretch}{1.0}
		\hrule
		\smallskip
		\begin{tabular}{@{}llc*{12}{c}@{}}\label{tab1}
			& & $\ell_n$ &
			\multicolumn{3}{c}{$1$} &
			\multicolumn{3}{c}{$2n^{1/5}$} &
			\multicolumn{3}{c}{$5n^{1/5}$} &
			\multicolumn{3}{c}{$10n^{1/5}$} \\
			\cmidrule(lr){3-3} \cmidrule(lr){4-6} \cmidrule(lr){7-9} \cmidrule(lr){10-12} \cmidrule(l){13-15}
			$\gamma_1$ & Test  & $n$ &
			$100$ & $250$ & $1000$ &
			$100$ & $250$ & $1000$ &
			$100$ & $250$ & $1000$ &
			$100$ & $250$ & $1000$ \\
			\cmidrule(r){1-1}\cmidrule(lr){2-2}\cmidrule(lr){3-3} \cmidrule(lr){4-4}\cmidrule(lr){5-5}\cmidrule(lr){6-6} \cmidrule(lr){7-7} \cmidrule(lr){8-8} \cmidrule(lr){9-9} \cmidrule(lr){10-10}\cmidrule(lr){11-11}\cmidrule(lr){12-12} \cmidrule(lr){13-13}\cmidrule(lr){14-14}\cmidrule(l){15-15}
			0 & MMD & & 0.057 & 0.041 & 0.041 & 0.064 & 0.047 & 0.045 & 0.065 & 0.048 & 0.044 & 0.052 & 0.038 & 0.036 \\
			& $\text{KCMD}_{\text{G}}$ & & 0.053 & 0.041 & 0.056 & 0.063 & 0.054 & 0.049 & 0.061 & 0.054 & 0.050 & 0.039 & 0.036 & 0.038 \\
			& DCOV & & 0.040 & 0.035 & 0.041 & 0.034 & 0.035 & 0.043 & 0.024 & 0.024 & 0.037 & 0.006 & 0.006 & 0.021 \\
			& $\text{\normalfont HSIC}_{\text{G}}$ & & 0.041 & 0.032 & 0.052 & 0.027 & 0.027 & 0.048 & 0.008 & 0.017 & 0.035 & 0.003 & 0.003 & 0.019 \vspace{2mm}\\
			
			0.75 & MMD & & 0.072 & 0.060 & 0.079 & 0.065 & 0.046 & 0.049 & 0.066 & 0.058 & 0.051 & 0.056 & 0.049 & 0.048 \\
			& $\text{KCMD}_{\text{G}}$ & & 0.069 & 0.063 & 0.064 & 0.063 & 0.045 & 0.054 & 0.063 & 0.051 & 0.050 & 0.046 & 0.040 & 0.044 \\
			& DCOV & & 0.056 & 0.055 & 0.066 & 0.046 & 0.039 & 0.054 & 0.033 & 0.025 & 0.044 & 0.016 & 0.013 & 0.037 \\
			& $\text{\normalfont HSIC}_{\text{G}}$ & & 0.042 & 0.044 & 0.053 & 0.029 & 0.034 & 0.042 & 0.013 & 0.014 & 0.030 & 0.001 & 0.003 & 0.014 \vspace{2mm}\\

			1.5 & MMD & & 0.186 & 0.187 & 0.208 & 0.091 & 0.077 & 0.081 & 0.076 & 0.062 & 0.066 & 0.075 & 0.069 & 0.065 \\
			& $\text{KCMD}_{\text{G}}$ & & 0.166 & 0.157 & 0.188 & 0.080 & 0.062 & 0.075 & 0.077 & 0.062 & 0.062 & 0.065 & 0.054 & 0.058 \\
			& DCOV & & 0.173 & 0.188 & 0.207 & 0.080 & 0.069 & 0.069 & 0.057 & 0.050 & 0.057 & 0.031 & 0.040 & 0.056 \\
			& $\text{\normalfont HSIC}_{\text{G}}$ & & 0.123 & 0.134 & 0.160 & 0.055 & 0.055 & 0.051 & 0.033 & 0.041 & 0.041 & 0.015 & 0.022 & 0.024 \vspace{2mm}\\

			2.25 & MMD & & 0.609 & 0.670 & 0.668 & 0.303 & 0.294 & 0.242 & 0.183 & 0.160 & 0.111 & 0.147 & 0.115 & 0.067 \\
			& $\text{KCMD}_{\text{G}}$ & & 0.642 & 0.701 & 0.728 & 0.312 & 0.283 & 0.228 & 0.167 & 0.152 & 0.098 & 0.135 & 0.098 & 0.056 \\
			& DCOV & & 0.704 & 0.792 & 0.845 & 0.326 & 0.290 & 0.250 & 0.172 & 0.148 & 0.104 & 0.138 & 0.098 & 0.067 \\
			& $\text{\normalfont HSIC}_{\text{G}}$ & & 0.769 & 0.871 & 0.919 & 0.324 & 0.297 & 0.244 & 0.150 & 0.138 & 0.102 & 0.096 & 0.085 & 0.058 \\
		\end{tabular}
		
	\smallskip
	\hrule	
	\end{table}

	\subsection{Time series autodependence}
	\label{sec4.2}
	
	One of the most prominent applications of distance-based methods in the time series literature concerns the development of autodependence measures, which quantify the dependence across temporal lags within a single stochastic process. To distinguish true dependence from spurious associations, these measures must be accompanied by confidence regions or, equivalently, by well-calibrated hypothesis tests that assess the significance of the corresponding estimators. To illustrate the ability of kernel-based tests to correctly detect the lack of (mean) independence between lagged observations of a time series, we generate $\{Y_i\}_{i=1}^n$ according to the self-excited functional threshold autoregressive model \citep{Li2024} given by
	$$Y_i(t') =  \int_0^1 1.5\min\{t,t'\}Y_{i-2}(t)dt \!\left(\ind{||Y_{i-2}||_{\mathcal{Y}}\le 1}-\ind{||Y_{i-2}||_{\mathcal{Y}}>1}\right)  + \mathcal{E}_i(t'),$$
	where $\{\mathcal{E}_i\}_{i\in \mathbb{Z}}$ is defined as in Section \ref{sec4.1} and we let $X_i = Y_{i-\text{lag}}$ for all $i$ and $\text{lag}\in \{1,2,4,6\}$. Table \ref{tab2} contains the empirical rejection rates obtained by setting $\ell_n$ to the closest integer to $2n^{1/5}$. It is easy to verify that $Y_i$ is independent of $Y_{i-1}$ for all $i$, which is correctly reflected by the fact that all four tests preserve the nominal level $\alpha = 0.05$. When $\text{lag}=2$, the tests exhibit non-trivial power that increases with the sample size and, remarkably, the tests based on the Gaussian kernel achieve higher power than those based on the distance kernel. The exact opposite occurs when $\text{lag}=4$ whereas, for $\text{lag}=6$, all tests display very limited power, even in the case of $n=1000$.

	\begin{table}[t]
		\caption{Empirical rejection rates, approximated with $1000$ Monte Carlo replications, of the MMD, $\text{KCMD}_{\text{G}}$, DCOV and $\text{\normalfont HSIC}_{\text{G}}$ tests of null $\text{\normalfont HSIC}(X,Y)$ for the time series autoregression scenario defined in Section \ref{sec4.2}, under which $X_{i} = Y_{i-\text{lag}}$ for all $i$. We consider an autoregression lag of
		$1,2,4$ and $6$ and a sample size of $n\in \{100,250,1000\}$. All tests are calibrated with nominal level $\alpha= 0.05$ using $B=1000$ bootstrap resamples.}
		\centering
		\smallskip
		\setlength{\tabcolsep}{4.5pt}
		\renewcommand{\arraystretch}{1.0}
		\hrule
		\smallskip
		\begin{tabular}{@{}lc*{12}{c}@{}}		\label{tab2}
			& $\text{lag}$ &
			\multicolumn{3}{c}{$1$} &
			\multicolumn{3}{c}{$2$} &
			\multicolumn{3}{c}{$4$} &
			\multicolumn{3}{c}{$6$} \\
			\cmidrule(lr){2-2} \cmidrule(lr){3-5} \cmidrule(lr){6-8} \cmidrule(lr){9-11} \cmidrule(l){12-14}
			Test & $n$ &
			$100$ & $250$ & $1000$ &
			$100$ & $250$ & $1000$ &
			$100$ & $250$ & $1000$ &
			$100$ & $250$ & $1000$ \\
			\cmidrule(r){1-1}\cmidrule(lr){2-2}\cmidrule(lr){3-3} \cmidrule(lr){4-4}\cmidrule(lr){5-5}\cmidrule(lr){6-6} \cmidrule(lr){7-7} \cmidrule(lr){8-8} \cmidrule(lr){9-9} \cmidrule(lr){10-10}\cmidrule(lr){11-11}\cmidrule(lr){12-12} \cmidrule(lr){13-13}\cmidrule(l){14-14}
			MMD & & 
			0.050 & 0.051 & 0.053 & 
			0.432 & 0.924 & 1.000 & 
			0.118 & 0.241 & 0.741 & 
			0.072 & 0.088 & 0.172 \\
			$\text{KCMD}_{\text{G}}$ & & 
			0.056 & 0.053 & 0.055 & 
			0.781 & 1.000 & 1.000 & 
			0.076 & 0.149 & 0.565 & 
			0.070 & 0.057 & 0.108 \\
			DCOV & & 
			0.026 & 0.044 & 0.052 & 
			0.412 & 0.950 & 1.000 & 
			0.106 & 0.225 & 0.739 & 
			0.056 & 0.077 & 0.159 \\
			$\text{\normalfont HSIC}_{\text{G}}$ & & 
			0.062 & 0.054 & 0.075 & 
			0.716 & 0.998 & 1.000 & 
			0.052 & 0.147 & 0.586 & 
			0.047 & 0.058 & 0.092 \\
		\end{tabular}		
		\smallskip
		\hrule	
	\end{table}

	\subsection{Power comparison}
	\label{sec4.3}
	To gain further insight into the specific advantages of choosing between distance and Gaussian kernels, we
	generate $\{Y_i\}_{i=1}^n$ according to the concurrent regression model
	$$Y_i(t) = \frac{1}{3}\Big(\gamma_2X_i(t) + (1-\gamma_2)X'_i(t) \Big) + 2\Big(\gamma_3\sin(2X_i(t)) + (1-\gamma_3)\sin(2X''_i(t)) \Big) +\Big(\gamma_4X_i(t) + (1-\gamma_4)X'''_i(t) \Big)\mathcal{E}_{\mathcal{Y}i}(t),$$
	where $\gamma_2, \gamma_3, \gamma_4\in \{0,1\}$, $\{X_i\}_{i=1}^n,\{X'_i\}_{i=1}^n,\{X''_i\}_{i=1}^n$ and $\{X'''_i\}_{i=1}^n$ are i.i.d. samples from the functional GARCH(1,1) model introduced in Example 2 in \cite{Cerovecki2019}, which verify
	$$\begin{cases}
		X_i(t') \hspace{0.7mm}= \sigma_i(t')\mathcal{E}_{\mathcal{X}i}(t'), \\
		\sigma^2_i(t') = 0.1+(t'-0.5)^2 + \!\displaystyle\int_0^1\!\left(0.2+(t-0.5)^2 + (t'-0.5)^2\right)X^2_{i-1}(t)dt + \!\displaystyle\int_0^1\!\left(0.4+(t-0.5)^2 + (t'-0.5)^2\right)\sigma^2_{i-1}(t)dt,
	\end{cases}$$
	and $\{\mathcal{E}_{\mathcal{X}i}\}_{i\in \mathbb{Z}}$ and $\{\mathcal{E}_{\mathcal{Y}i}\}_{i\in \mathbb{Z}}$ are i.i.d. collections of Gaussian random variables with zero mean and covariance function $\text{Cov}(\mathcal{E}_{\mathcal{X}i}(t),\mathcal{E}_{\mathcal{X}i}(t')) = \text{Cov}(\mathcal{E}_{\mathcal{Y}i}(t),\mathcal{E}_{\mathcal{Y}i}(t')) = \exp(-|t-t'|/2)$. In particular, we denote by DGP 1 the model corresponding to $\gamma_2 = \gamma_3 =\gamma_4 = 0$ ($X$ and $Y$ are independent), by DGP 2 the model associated with $\gamma_2 = 1$ and $\gamma_3=\gamma_4 = 0$ (a linear term on $\mathbb{E}[Y|X]$), by DGP 3 the model induced by $\gamma_3 = 1$ and $\gamma_2=\gamma_4 = 0$ (a nonlinear term in $\mathbb{E}[Y|X]$) and by DGP 4 the model resulting from setting $\gamma_4=1$ and $\gamma_2 = \gamma_3 = 0$ (conditional heteroscedasticity). From the empirical rejection rates displayed in Table \ref{tab3}, which were computed with $\ell_n$ set to the closest integer to $2n^{1/5}$, it can be concluded that all tests preserve the nominal level under the common null hypothesis in DGP 1, although the convergence of the empirical level to $\alpha=0.05$ seems to be quite slow. Moreover, the tests based on the Gaussian kernel with the median heuristic achieve greater power against nonlinear alternatives associated with the conditional expectation than those based on the distance kernel, whereas the latter outperform the former when we consider linear alternatives. The mean independence tests, MMD and $\text{KCMD}_{\text{G}}$, do preserve, as expected, the nominal level in DGP 4 since, in that case, $Y$ is mean independent of $X$. In addition, the performance of DCOV seems to be substantially worse than that of $\text{\normalfont HSIC}_{\text{G}}$ when the dependence between $Y$ and $X$ is not captured by the conditional expectation.

	\begin{table}[t]
		\caption{Empirical rejection rates, approximated with $1000$ Monte Carlo replications, of the MMD, $\text{KCMD}_{\text{G}}$, DCOV and $\text{\normalfont HSIC}_{\text{G}}$ tests of null $\text{\normalfont HSIC}(X,Y)$ for the four different data generating processes (DGP) defined in Section \ref{sec4.3}. $X$ and $Y$ are independent under DGP $1$, $\mathbb{E}[Y|X]$ is linear on $X$ under  DGP $2$ and nonlinear under DGP $3$ and, finally, $Y$ is mean independent, but not independent, of $X$ under DGP $4$. Results are reported for a sample size of $n\in \{100,250,1000\}$. All tests are calibrated with nominal level $\alpha= 0.05$ using $B=1000$ bootstrap resamples.}
		\centering
		\smallskip
		\setlength{\tabcolsep}{4.5pt}
		\renewcommand{\arraystretch}{1.0}
		\hrule
		\smallskip
		\begin{tabular}{@{}lc*{12}{c}@{}}\label{tab3}
			& DGP &
			\multicolumn{3}{c}{$1$} &
			\multicolumn{3}{c}{$2$} &
			\multicolumn{3}{c}{$3$} &
			\multicolumn{3}{c}{$4$} \\
			\cmidrule(lr){2-2} \cmidrule(lr){3-5} \cmidrule(lr){6-8} \cmidrule(lr){9-11} \cmidrule(l){12-14}
			Test &  $n$ &
			$100$ & $250$ & $1000$ &
			$100$ & $250$ & $1000$ &
			$100$ & $250$ & $1000$ &
			$100$ & $250$ & $1000$ \\
			\cmidrule(r){1-1}\cmidrule(lr){2-2}\cmidrule(lr){3-3} \cmidrule(lr){4-4}\cmidrule(lr){5-5}\cmidrule(lr){6-6} \cmidrule(lr){7-7} \cmidrule(lr){8-8} \cmidrule(lr){9-9} \cmidrule(lr){10-10}\cmidrule(lr){11-11}\cmidrule(lr){12-12} \cmidrule(lr){13-13}\cmidrule(l){14-14}
			MMD & & 
			0.033 & 0.043 & 0.040 & 
			0.657 & 0.932 & 0.994 & 
			0.466 & 0.771 & 0.986 & 
			0.032 & 0.028 & 0.029 \\
			$\text{KCMD}_{\text{G}}$ & & 
			0.024 & 0.033 & 0.037 & 
			0.480 & 0.808 & 0.989 & 
			0.668 & 0.883 & 0.995 & 
			0.023 & 0.026 & 0.035 \\
			DCOV & & 
			0.012 & 0.022 & 0.036 & 
			0.630 & 0.963 & 1.000 & 
			0.447 & 0.836 & 0.999 & 
			0.030 & 0.167 & 0.986 \\
			$\text{\normalfont HSIC}_{\text{G}}$ & & 
			0.014 & 0.029 & 0.052 & 
			0.415 & 0.886 & 1.000 & 
			0.702 & 0.948 & 1.000 & 
			0.384 & 0.941 & 1.000 \\
		\end{tabular}
		\smallskip
		\hrule	
	\end{table}

	\section{Discussion and further research}
	\label{sec5}
	
	We have reviewed the fundamental theory behind the Hilbert-Schmidt independence criterion and its applications on independence and mean independence tests, providing some new results that broaden the scope of this methodology. We have also presented new asymptotic results for the empirical estimator of the HSIC based on stationary, ergodic and near epoch dependent data on Polish spaces, including a consistent bootstrap scheme which, according to simulations, yields tests that preserve nominal levels in practice when the data dependence is, indeed, weak. However, bootstrap calibration has been shown to be quite sensitive to the choice of $\ell_n$, particularly as data dependence increases. Therefore, a highly relevant open problem is the development of a data-driven procedure for the automatic selection of this bootstrap hyperparameter so that the asymptotic validity of the calibration is preserved.
	
	The results in this paper can be extended to $q$-variable mutual independence tests based on product kernels. Let $X^1\in \mathcal{X}^1,\ldots, X^q\in \mathcal{X}^q$ and let $c_{\mathcal{X}^k}$ be a kernel on $\mathcal{X}^k\times \mathcal{X}^k$ with canonical feature map $\Phi_{\mathcal{X}^k}$, where $X^k$ is Polish space for all $k=1,\ldots,q$. Based on the collection of RKHS $\{\mathcal{C}_{\mathcal{X}^k}\}_{k=1}^q$, we can consider the $q$-fold tensor product Hilbert space $\mathcal{C}_{\mathcal{X}^1} \otimes \cdots \otimes\mathcal{C}_{\mathcal{X}^q}$, which is constructed, iteratively, as $\mathcal{C}_{\mathcal{X}^1} \otimes \cdots \otimes\mathcal{C}_{\mathcal{X}^k} = (\mathcal{C}_{\mathcal{X}^1} \otimes \cdots \otimes\mathcal{C}_{\mathcal{X}^{k-1}})\otimes\mathcal{C}_{\mathcal{X}^k}$ for all $k\ge 2$. Then, the $q$-variable HSIC associated to $X^1\ldots, X^q$ is defined as 
	$$\text{\normalfont HSIC}(X^1,\ldots,X^q) = \!\left|\left|\mathbb{E}\Big[\Phi_{\mathcal{X}^1}(X^1) \otimes \cdots \otimes \Phi_{\mathcal{X}^q}(X^q) \Big] - \mathbb{E}\Big[\Phi_{\mathcal{X}^1}(X^1)\Big] \otimes \cdots \otimes \mathbb{E}\Big[\Phi_{\mathcal{X}^q}(X^q)\Big] \right|\right|^2_{\mathcal{C}_{\mathcal{X}^1}\otimes \cdots \otimes \hspace{0.4mm} \mathcal{C}_{\mathcal{X}^q}}.$$
	As noted in \cite{Szabo2018}, it is generally not sufficient to have one characteristc kernel for each variable to completely characterize mutual independence. If all kernels are universal \citep{Sriperumbudur2011}, however, the characterization holds. The corresponding empirical analogue admits the expansion
	$$ \text{\normalfont HSIC}_n(X^1,\ldots,X^q) = \frac{1}{n^2}\sum_{i,j=1}^n \prod_{k=1}^q  c_{\mathcal{X}^k}(X_i^k,X_j^k) -\frac{2}{n}\sum_{i=1}^n \prod_{k=1}^q \frac{1}{n}\sum_{j=1}^n c_{\mathcal{X}^k}(X_i^k,X_j^k) +  \prod_{k=1}^q \frac{1}{n^2}\sum_{i,j=1}^n  c_{\mathcal{X}^k}(X_i^k,X_j^k).$$
	This estimator, like $\text{\normalfont HSIC}_n(X,Y)$, is nothing but the squared norm of the difference of a sample average of tensor products and a tensor product of sample averages. Thus, both the technical conditions (T1)--(T6) and the proofs of Theorems \ref{Th1}--\ref{Th5} can be adapted to show its consistency, the limit distribution under the null and local alternatives and the validity of the multiplier bootstrap scheme based on the analogous bootstrap statistic
	\begin{align*}
		\text{\normalfont HSIC}^*_n(X^1,\ldots,X^q) &= \frac{1}{n^2}\sum_{i,j=1}^n (r_{in}-\bar{r}_{\cdot n})(r_{jn}-\bar{r}_{\cdot n})\prod_{k=1}^q  c_{\mathcal{X}^k}(X_i^k,X_j^k)  -\frac{2}{n}\sum_{i=1}^n(r_{in}-\bar{r}_{\cdot n}) \prod_{k=1}^q \frac{1}{n}\sum_{j=1}^n (r_{jn}-\bar{r}_{\cdot n})c_{\mathcal{X}^k}(X_i^k,X_j^k)\\
		&\quad +  \prod_{k=1}^q \frac{1}{n^2}\sum_{i,j=1}^n  (r_{in}-\bar{r}_{\cdot n})(r_{jn}-\bar{r}_{\cdot n})c_{\mathcal{X}^k}(X_i^k,X_j^k).
	\end{align*}
	
	In some cases, particularly for time series modelling, it might be interesting to test whether $Y_i$ is (mean) independent of $\{X_j\}_{j\le i}$ for all $i\in \mathbb{Z}$. In line with the proposals in \cite{Fokianos2018}, \cite{Hlavka2020} and \cite{Jiang2023}, it is possible to adapt the results and proofs in this document to show that a test based on a weighted sum of $\text{\normalfont HSIC}_{n-j}(Y_{0},X_{-j})$ can consistently detect serial (mean) independence, i.e, pairwise (mean) independence of $Y_i$ and $X_{i-j}$ for all $j\ge 0$ and all $i\in \mathbb{Z}$.
	
	\section*{Acknowledgements}
	
	This work is supported by grant ED431C 2025/03, funded by the Xunta de Galicia, and also from projects PID2020-116587GB-I00 and PID2024-158017NB-I00, funded by MICIU/AEI/10.13039/501100011033 and by the European Regional Development Fund.

	\appendix
	\section{Proofs of results in Section 2}
	\label{app1}
	
	\begin{proof}[\textup{\textbf{Proof of Proposition \ref{pr1}}}]
		By construction, $c_{\mathcal{S}}$ is a kernel. Moreover, since $\Xi$ and $\Phi_{\tilde{\mathcal{S}}}$ are (Borel) measurable, $\Phi_{\mathcal{S}}(\star) = \Phi_{\tilde{\mathcal{S}}}(\Xi(\star))$ is measurable. 
		In addition, by Purves's Theorem \citep{Purves1966}, the inverse $\Xi^{-1}:\Xi(\tilde{\mathcal{S}})\rightarrow \mathcal{S}$ is also measurable. 
		As $c_{\tilde{\mathcal{S}}}$ is characteristic, for any pair of random variables $S,S'\in \mathcal{\tilde{S}}$ such that $\Phi_{\tilde{\mathcal{S}}}(\Xi(S))$ and $\Phi_{\tilde{\mathcal{S}}}(\Xi(S'))$ are Pettis integrable, $\mathbb{E}[\Phi_{\tilde{\mathcal{S}}}(\Xi(S))] = \mathbb{E}[\Phi_{\tilde{\mathcal{S}}}(\Xi(S'))]$ implies that $\Xi(S)$ and $\Xi(S')$ are identically distributed . Moreover, the measurability of $\Xi^{-1}$ implies that the Borel $\sigma$-algebra of $\mathcal{S}$ coincides with the preimage of the Borel $\sigma$-algebra of $\tilde{\mathcal{S}}$ and, therefore, for all borel measurable set $B_{\tilde{\mathcal{S}}}\subset \tilde{\mathcal{S}}$,
		$$\mathbb{P}\big(S\in \Xi^{-1}(B_{\tilde{\mathcal{S}}})\big) = \mathbb{P}\big(S'\in \Xi^{-1}(B_{\tilde{\mathcal{S}}})\big) \iff \mathbb{P}\big(\Xi(S)\in B_{\tilde{\mathcal{S}}}\big) = \mathbb{P}\big(\Xi(S)\in B_{\tilde{\mathcal{S}}}\big).$$
		That is, $S$ and $S'$ must be identically distributed too and, therefore, $c_{\mathcal{S}}$ is characteristic.
		
		To prove (ii), let $d_{\mathcal{S}}$ be a metric inducing the topology of $\mathcal{S}$, which exists since $\mathcal{S}$ is completely metrizable. We can assume that $d_{\mathcal{S}}$ is bounded by some constant $0<||d_{\mathcal{S}}||_\infty<\infty$ since $\min\{||d_{\mathcal{S}}||_\infty,d_{\mathcal{S}}\}$ is another metric that induces the same topology. Let $\{s_k\}_{k\in \mathbb{N}^+}$ be a dense subset of $\mathcal{S}$, which exists since $\mathcal{S}$ is separable. Define
		$\Xi(s) = \{w_kd_{\mathcal{S}}(s,s_k)\}_{k\in \mathbb{N}^+}$ where $\{w_k\}_{k\in \mathbb{N}^+}$ is a sequence of positive and square summable weights. Then,
		$$\sum_{k=1}^\infty w^2_kd^2_{\mathcal{S}}(s,s_k) \le ||d_{\mathcal{S}}||^2_\infty\sum_{k=1}^\infty w^2_k < \infty,$$
		and, thus, $\Xi(s)$ belongs to $\ell^2$ for all $s\in \mathcal{S}$. Continuity follows from the reverse triangle inequality, since
		$$||\Xi(s) - \Xi(s')||^2_{\ell^2} = \sum_{k=1}^\infty w^2_k|d_{\mathcal{S}}(s,s_k) - d_{\mathcal{S}}(s',s_k)|^2 \le d^2_{\mathcal{S}}(s,s')\sum_{k=1}^\infty w^2_k.$$
		Finally, since $\{s_k\}_{k\in \mathbb{N}^+}$ is dense, if $s\neq s'$ then there must be some $k\in \mathbb{N}^+$ such that $d(s,s_k)\neq d(s',s_k)$. As $w_k>0$, $\Xi(s)\neq \Xi(s')$ and, therefore, $\Xi$ is injective.
	\end{proof}
	
	\begin{proof}[\textup{\textbf{Proof of Proposition \ref{pr2}}}]
		
		We show that $X$ and $Y$ are independent if and only if $ \mathbb{E}[\langle\Phi_{\mathcal{Y}}(Y),\zeta_{\mathcal{Y}}\rangle_{\mathcal{C}_{\mathcal{Y}}}|X] = \mathbb{E}[\langle\Phi_{\mathcal{Y}}(Y),\zeta_{\mathcal{Y}}\rangle_{\mathcal{C}_{\mathcal{Y}}}]$ a.s. for all $\zeta_{\mathcal{Y}}\in \mathcal{C}_{\mathcal{Y}}$.  In such case, since $ \text{\normalfont HSIC}(X,Y) = \text{\normalfont HSIC}(X,\Phi_{\mathcal{Y}}(Y);c_{\mathcal{X}}, \langle\star,\cdot\rangle_{\mathcal{C}_{\mathcal{Y}}})$ by definition, it follows from Proposition \ref{pr3} that $\text{\normalfont HSIC}(X,Y) = 0$ if and only if $X$ and $Y$ are independent.
		
		By the reproducing property, $\mathbb{E}[\langle\Phi_{\mathcal{Y}}(Y),\zeta_{\mathcal{Y}}\rangle_{\mathcal{C}_{\mathcal{Y}}}|X] = \mathbb{E}[\langle\Phi_{\mathcal{Y}}(Y),\zeta_{\mathcal{Y}}\rangle_{\mathcal{C}_{\mathcal{Y}}}]$ a.s. if and only if $\mathbb{E}[\zeta_{\mathcal{Y}}(Y)|X] = \mathbb{E}[\zeta_{\mathcal{Y}}(Y)]$ a.s. for all $\zeta_{\mathcal{Y}}\in \mathcal{C}_{\mathcal{Y}}$. Of course this trivially holds if $X$ and $Y$ are independent. Let us fix $\zeta_{\mathcal{Y}}$ and let us assume that $\mathbb{E}[\zeta_{\mathcal{Y}}(Y)|X] = \mathbb{E}[\zeta_{\mathcal{Y}}(Y)]$ a.s. Then, by definition of conditional expectation with respect to $\sigma(X)$, 
		\begin{equation}\label{Pr2.1}
			\mathbb{E}\!\left[\zeta_{\mathcal{Y}}(Y) - \mathbb{E}[\zeta_{\mathcal{Y}}(Y)]|X\right] = 0 \ \ \text{a.s.} \implies \mathbb{E}\!\left[(\zeta_{\mathcal{Y}}(Y) - \mathbb{E}[\zeta_{\mathcal{Y}}(Y)])\ind{X\in B_{\mathcal{X}}}\right] = 0,
		\end{equation}
		for all $B_{\mathcal{X}}\in \mathcal{B}(X)$, the borel $\sigma$-algebra of $\mathcal{X}$. Let $\mathbb{P}(X \in B_{\mathcal{X}})>0$. Then,  $\mathbb{P}_{B_{\mathcal{X}}}(B_{\Omega}) = \mathbb{P}(B_{\Omega}\cap \{X\in B_{\mathcal{X}}\})/\mathbb{P}(X \in B_{\mathcal{X}})$ for all $B_{\Omega}\in \mathcal{F}$, where $\{X\in B_{\mathcal{X}}\}$ is shorthand for $\{\omega \in \Omega\mid X(\omega)\in B_{\mathcal{X}}\}$, defines a probability measure on $(\Omega,\mathcal{F})$ which is absolutely continuous with respect to $\mathbb{P}$ with Radon-Nikodym derivative $\ind{X\in B_{\mathcal{X}}}/\mathbb{P}(X \in B_{\mathcal{X}})$. Therefore, 
		\begin{align}
			\int \zeta_{\mathcal{\mathcal{Y}}}(Y(w)) \frac{\ind{X\in B_{\mathcal{X}}}(\omega)}{\mathbb{P}(X \in B_{\mathcal{X}})}d\mathbb{P}(w) &= \int \zeta_{\mathcal{\mathcal{Y}}}(Y(w)) d\mathbb{P}_{B_{\mathcal{X}}}(w) = \int \zeta_{\mathcal{\mathcal{Y}}}(y) d \Big(\mathbb{P}_{B_{\mathcal{X}}}\circ Y^{-1}\Big)(y), \label{Pr2.2}
		\end{align}
		where the latter is a Lebesgue integral with respect to the pushforward probability measure $\mathbb{P}_{B_{\mathcal{X}}}\circ Y^{-1}$ defined as $\mathbb{P}_{B_{\mathcal{X}}}\circ Y^{-1}(B_{\mathcal{Y}}) = \mathbb{P}_{B_{\mathcal{X}}}(Y\in B_{\mathcal{Y}})$ for all $B_{\mathcal{Y}}\in \mathcal{B}(\mathcal{Y})$, the borel $\sigma$-algebra of $\mathcal{Y}$. Thus,
		given a random variable $\tilde{Y}\in \mathcal{Y}$ with probability law $\mathbb{P}_{B_{\mathcal{X}}}\circ Y^{-1}$ and putting (\ref{Pr2.1}) and (\ref{Pr2.2}) together,
		$$ \mathbb{E}[\zeta_{\mathcal{Y}}(\tilde{Y})]\mathbb{P}(X \in B_{\mathcal{X}}) = \mathbb{E}\!\left[\zeta_{\mathcal{Y}}(Y)\ind{X\in B_{\mathcal{X}}}\right] = \mathbb{E}[\zeta_{\mathcal{Y}}(Y)]\mathbb{P}(X \in B_{\mathcal{X}}),$$
		which implies that $ \mathbb{E}[\zeta_{\mathcal{Y}}(\tilde{Y})] =  \mathbb{E}[\zeta_{\mathcal{Y}}(Y)]$.
		As we assume that $\Phi_{\mathcal{Y}}(Y)$ is Pettis integrable, it follows (from definition of Pettis integrability) that $\Phi_{\mathcal{Y}}(Y)\ind{X\in B_{\mathcal{X}}}$ is also Pettis integrable and, in consequence, $\Phi_{\mathcal{Y}}(\tilde{Y})$ is Pettis integrable with $\mathbb{E}[\Phi_{\mathcal{Y}}(\tilde{Y})] = \mathbb{E}\big[\Phi_{\mathcal{Y}}(Y)\ind{X\in B_{\mathcal{X}}}\big]/\mathbb{P}(X\in B_{\mathcal{X}})$. Furthermore, we have proved that $\langle \mathbb{E}[\Phi_{\mathcal{Y}}(\tilde{Y})],\zeta_{\mathcal{Y}}\rangle_{\mathcal{Y}} =  \langle \mathbb{E}[\Phi_{\mathcal{Y}}(Y)],\zeta_{\mathcal{Y}}\rangle_{\mathcal{Y}}$ for all $\zeta_{\mathcal{Y}}\in \mathcal{C}_{\mathcal{Y}}$ and, thus, $ \mathbb{E}[\Phi_{\mathcal{Y}}(Y)] =  \mathbb{E}[\Phi_{\mathcal{Y}}(\tilde{Y})]$. Since $c_{\mathcal{Y}}$ is characteristic, it follows that $Y$ and $\tilde{Y}$ are identically distributed: for all $B_{\mathcal{Y}}\in \mathcal{B}(\mathcal{Y})$, $\mathbb{P}(Y\in B_{\mathcal{Y}}) = \mathbb{P}_{B_{\mathcal{X}}}(Y\in B_{\mathcal{Y}}) = \mathbb{P}(Y\in B_{\mathcal{Y}}, X \in B_{\mathcal{X}})/\mathbb{P}(X\in B_{\mathcal{X}})$. Equivalently, $$\mathbb{P}\!\left(Y\in B_{\mathcal{Y}}\right)\mathbb{P}(X \in B_{\mathcal{X}}) = \mathbb{P}(Y\in B_{\mathcal{Y}}, X \in B_{\mathcal{X}}),$$
		for all $B_{\mathcal{X}}\in \mathcal{B}(X)$ such that $\mathbb{P}(X\in B_{\mathcal{X}}) >0$ and all $B_{\mathcal{Y}}\in \mathcal{B}(\mathcal{Y})$. The same result trivially holds if $\mathbb{P}(X\in B_{\mathcal{X}}) = 0$ and, ultimately, $X$ and $Y$ are independent.
	\end{proof}
	
	\begin{proof}[\textup{\textbf{Proof of Proposition \ref{pr3}}}]
		Without generality loss, let us assume that $\mathbb{E}[Y] = 0$. If we further assume that, for all $y'\in \mathcal{Y}$, $\mathbb{E}[\langle Y,y'\rangle_{\mathcal{Y}}|X] = 0$ a.s., it follows from the law of iterated expectations that, for any $y'\in \mathcal{Y}$ and any $\zeta_{\mathcal{X}}\in \mathcal{C}_{\mathcal{X}}$,
		\begin{align*}
			0 &= \mathbb{E}\!\left[\zeta_{\mathcal{X}}(X)\mathbb{E}[\langle Y, y'\rangle_{\mathcal{Y}}|X]\right] =  \mathbb{E}\!\left[\zeta_{\mathcal{X}}(X)\langle Y, y'\rangle_{\mathcal{Y}}\right] =  \mathbb{E}\!\left[\langle \Phi_{\mathcal{X}}(X),\zeta_{\mathcal{X}}\rangle_{\mathcal{C}_{\mathcal{X}}} \langle Y, y'\rangle_{\mathcal{Y}}\right]\\
			&= \mathbb{E}\!\left[ \langle \Phi_{\mathcal{X}}(X) \otimes Y, \zeta_{\mathcal{X}}\otimes y' \rangle_{\mathcal{HS}}\right] =  \langle \mathbb{E}[\Phi_{\mathcal{X}}(X) \otimes Y], \zeta_{\mathcal{X}}\otimes y' \rangle_{\mathcal{HS}}.
		\end{align*}
		Let $h\in \mathcal{HS}$. Since sums of tensor products are dense on $\mathcal{HS}$ by construction, if follows that, for any $\epsilon>0$, $||h - \sum_{i=1}^k \zeta_{\mathcal{X}i} \otimes y'_i||_{\mathcal{HS}} \le \epsilon$ for some $\{\zeta_{\mathcal{X}i},y'_i\}_{i=1}^k$ and $k\in \mathbb{N}^+$. Thus, by the Cauchy-Schwarz inequality,
		\begin{align*}
			\left|\langle \mathbb{E}[\Phi_{\mathcal{X}}(X) \otimes Y], h \rangle_{\mathcal{HS}}\right| &\le \left|\left\langle \mathbb{E}[\Phi_{\mathcal{X}}(X) \otimes Y], h -\sum_{i=1}^k \zeta_{\mathcal{X}i} \otimes y'_i \right\rangle_{\mathcal{HS}}\right|\le \epsilon||\mathbb{E}[\Phi_{\mathcal{X}}(X) \otimes Y]||_{\mathcal{HS}}.
		\end{align*}
		Letting $\epsilon \rightarrow 0$, it follows that $\langle \mathbb{E}[\Phi_{\mathcal{X}}(X) \otimes Y], h \rangle_{\mathcal{HS}} = 0$. As this holds for any $h$, $\mathbb{E}[\Phi_{\mathcal{X}}(X) \otimes Y] = 0$. Analogously, $\mathbb{E}[\Phi_{\mathcal{X}c}(X) \otimes Y] = 0$ and, taking square norms, $\text{\normalfont HSIC}(X,Y,c_{\mathcal{X}}, \langle\star,\cdot\rangle_{\mathcal{Y}}) = 0$.
		
		Let us assume now that $\text{\normalfont HSIC}(X,Y,c_{\mathcal{X}}, \langle\star,\cdot\rangle_{\mathcal{Y}}) = 0$ or, equivalently, $\mathbb{E}[\Phi_{\mathcal{X}c}(X) \otimes Y] = 0$. By previous arguments, $\mathbb{E}[ \zeta_{\mathcal{X}}(X) \langle Y, y'\rangle_{\mathcal{Y}}]=0$ for all $\zeta_{\mathcal{X}},y'$. Let $\langle Y,y'\rangle_{\mathcal{Y}}^+$ and $\langle Y,y'\rangle_{\mathcal{Y}}^-$ denote the positive and negative parts of $\langle Y,y'\rangle_{\mathcal{Y}}$. If $\langle Y,y'\rangle = 0$ a.s., then $\mathbb{E}[\langle Y,y'\rangle|X] = 0$ a.s. holds trivially. Suppose $\langle Y, y'\rangle$ is not almost surely zero, so that neither its positive nor its negative part is almost surely zero. Furthermore, $0= \mathbb{E}\big[\langle Y,y'\rangle^+_{\mathcal{Y}}\big] - \mathbb{E}\big[\langle Y,y'\rangle^-_{\mathcal{Y}}\big]$ and, thus, $\mathbb{E}\big[\langle Y,y'\rangle^+_{\mathcal{Y}}\big] = \mathbb{E}\big[\langle Y,y'\rangle^-_{\mathcal{Y}}\big] >0$. It follows that
		\begin{equation}
			\frac{\mathbb{E}\big[ \zeta_{\mathcal{X}}(X) \langle Y, y'\rangle^+_{\mathcal{Y}}\big]}{\mathbb{E}\big[\langle Y,y'\rangle^+_{\mathcal{Y}}\big]} = \frac{\mathbb{E}\big[ \zeta_{\mathcal{X}}(X) \langle Y, y'\rangle^-_{\mathcal{Y}}\big]}{\mathbb{E}\big[\langle Y,y'\rangle^-_{\mathcal{Y}}\big]}, \label{Pr3.1}
		\end{equation}
		where both $\mathbb{E}\big[ \zeta_{\mathcal{X}}(X) \langle Y, y'\rangle^+_{\mathcal{Y}}\big]$ and $\mathbb{E}\big[ \zeta_{\mathcal{X}}(X) \langle Y, y'\rangle^-_{\mathcal{Y}}\big]$ exist and are finite as $\Phi_{\mathcal{X}}(X)\otimes Y$ is assumed to be Pettis integrable and, thus, 
		$$ \mathbb{E}\big[|\zeta_{\mathcal{X}}(X)|\langle Y, y'\rangle^+_{\mathcal{Y}}\big] + \mathbb{E}\big[|\zeta_{\mathcal{X}}(X)|\langle Y, y'\rangle^-_{\mathcal{Y}}\big] = \mathbb{E}\big[|\zeta_{\mathcal{X}}(X)\langle Y, y'\rangle_{\mathcal{Y}}|\big] <\infty.$$
		Let $\tilde{\mathbb{P}}^+_{y'}(B_{\mathcal{X}}) = \mathbb{E}\big[\langle Y, y'\rangle^+_{\mathcal{Y}}\ind{X\in B_{\mathcal{X}}}\big]/\mathbb{E}\big[\langle Y,y'\rangle^+_{\mathcal{Y}}\big]$  and $\tilde{\mathbb{P}}^-_{y'}(B_{\mathcal{X}}) = \mathbb{E}\big[\langle Y, y'\rangle^-_{\mathcal{Y}}\ind{X\in B_{\mathcal{X}}}\big]/\mathbb{E}\big[\langle Y,y'\rangle^-_{\mathcal{Y}}\big]$ for all $B_{\mathcal{X}}\in \mathcal{B}(\mathcal{X})$.  It is easy to verify that $\tilde{\mathbb{P}}^+_{y'}$ and $\tilde{\mathbb{P}}^-_{y'}$ are well-defined probability measures on $(\mathcal{X}, \mathcal{B}(\mathcal{X}))$. Moreover, letting $\tilde{X}_{y'}^+$ and $\tilde{X}_{y'}^-$ be two random variables whose probability laws are $\tilde{\mathbb{P}}^+_{y'}$ and $\tilde{\mathbb{P}}^-_{y'}$, respectively, it follows that, for any measurable function $l:\mathcal{X}\rightarrow \mathbb{R}$ such that $\mathbb{E}\big[l(X)\langle Y,y'\rangle^+_{\mathcal{Y}}\big]$ and $\mathbb{E}\big[l(X)\langle Y,y'\rangle^-_{\mathcal{Y}}\big]$ are well-defined and finite,
		$\mathbb{E}\big[l(\tilde{X}_{y'}^+)\big] = \mathbb{E}\big[l(X)\langle Y,y'\rangle^+_{\mathcal{Y}}\big]/\mathbb{E}\big[\langle Y,y'\rangle^+_{\mathcal{Y}}\big]$ and $\mathbb{E}\big[l(\tilde{X}_{y'}^-)\big] = \mathbb{E}\big[l(X)\langle Y,y'\rangle^-_{\mathcal{Y}}\big]/\mathbb{E}\big[\langle Y,y'\rangle^-_{\mathcal{Y}}\big]$.
		This follows inmediatly for simple, measurable functions and it can be extended by the monotone convergence theorem to positive functions and also to general measurable functions by decomposition into positive and negative parts. In particular, (\ref{Pr3.1}) can be re-expressed as
		\begin{equation}
			\mathbb{E}\big[ \zeta_{\mathcal{X}}(\tilde{X}^+_{y'})\big] = \mathbb{E}\big[ \zeta_{\mathcal{X}}(\tilde{X}^-_{y'})\big] \ \ \text{for all }\zeta_{\mathcal{X}}\in \mathcal{C}_{\mathcal{X}},\label{Pr3.2}
		\end{equation}
		so $\mathbb{E}\big[ \langle \Phi_{\mathcal{X}}(\tilde{X}^+_{y'}), \zeta_{\mathcal{X}} \rangle_{\mathcal{C}_{\mathcal{X}}}\big]  = \mathbb{E}\big[ \langle \Phi_{\mathcal{X}}(\tilde{X}^-_{y'}), \zeta_{\mathcal{X}} \rangle_{\mathcal{C}_{\mathcal{X}}}\big]$
		holds for all $\zeta_{\mathcal{X}}\in \mathcal{C}_{\mathcal{X}}$. Since $\Phi_{\mathcal{X}}(X)\otimes Y$ is Pettis integrable, it follows that $\Phi_{\mathcal{X}}(X)\otimes Y \mathbbm{1}_{B_{\Omega}}$ is also Pettis integrable for all $B_{\Omega}\in \mathcal{F}$. Moreover, the map $\mathbb{E}\big[ \langle \Phi_{\mathcal{X}}(X), \cdot\rangle_{\mathcal{C}_{\mathcal{X}}}\langle Y, y\rangle_{\mathcal{Y}} \mathbbm{1}_{B_{\Omega}}\big]:\mathcal{C}_{\mathcal{X}} \rightarrow \mathbb{R}$ defined as $\mathbb{E}\big[ \langle \Phi_{\mathcal{X}}(X), \cdot\rangle_{\mathcal{C}_{\mathcal{X}}}\langle Y, y\rangle_{\mathcal{Y}}\mathbbm{1}_{B_{\Omega}} \big](\zeta_{\mathcal{X}}) = \mathbb{E}\big[ \zeta_{\mathcal{X}}(X) \langle Y, y'\rangle_{\mathcal{Y}}\mathbbm{1}_{B_{\Omega}}\big]$ is clearly linear and, given that
		$$ \sup_{|| \zeta_{\mathcal{X}}||_{\mathcal{C}_{\mathcal{X}}}=1}\!\left|\mathbb{E}\!\left[ \zeta_{\mathcal{X}}(X) \langle Y, y'\rangle_{\mathcal{Y}}\mathbbm{1}_{B_{\Omega}}\right]\right| =  \sup_{|| \zeta_{\mathcal{X}}||_{\mathcal{C}_{\mathcal{X}}}=1}\!\left|\left\langle \mathbb{E}\big[ \Phi_{\mathcal{X}}(X)\otimes Y \mathbbm{1}_{B_{\Omega}}\big], \zeta_{\mathcal{X}}\otimes y' \right\rangle_{\mathcal{HS}}\right|\le ||y'||_{\mathcal{Y}}\!\left|\left|\mathbb{E}\big[ \Phi_{\mathcal{X}}(X)\otimes Y \mathbbm{1}_{B_{\Omega}}\big]\right|\right|_{\mathcal{HS}},$$
		it is also bounded. By virtue of the Riesz representation theorem, there exists an unique $\mathbb{E}\big[\Phi_{\mathcal{X}}(X) \langle Y, y'\rangle_{\mathcal{Y}}\mathbbm{1}_{B_{\Omega}}\big]\in \mathcal{C}_{\mathcal{X}}$, for all $B_{\Omega}\in \mathcal{F}$, such that $\langle \mathbb{E}\big[\Phi_{\mathcal{X}}(X) \langle Y, y'\rangle_{\mathcal{Y}}\mathbbm{1}_{B_{\Omega}}\big], \zeta_{\mathcal{X}}\rangle_{\mathcal{X}} = \mathbb{E}\big[ \zeta_{\mathcal{X}}(X) \langle Y, y'\rangle_{\mathcal{Y}}\mathbbm{1}_{B_{\Omega}}\big]$ for all $\zeta_{\mathcal{X}}\in \mathcal{C}_{\mathcal{X}}$. In particular, $\langle Y, y'\rangle^+_{\mathcal{Y}} = \langle Y, y'\rangle_{\mathcal{Y}} \ind{\langle Y, y'\rangle_{\mathcal{Y}} \ge0}$ and $\langle Y, y'\rangle^-_{\mathcal{Y}} = \langle Y, y'\rangle_{\mathcal{Y}} \ind{\langle Y, y'\rangle_{\mathcal{Y}} \le0}$ and, therefore,
		$\mathbb{E}\big[ \Phi_{\mathcal{X}}(\tilde{X}^+_{y'})\big] = \mathbb{E}\big[\Phi_{\mathcal{X}}(X) \langle Y, y'\rangle^+_{\mathcal{Y}}\big]/\mathbb{E}\big[\langle Y,y'\rangle^+_{\mathcal{Y}}\big]$ and  $\mathbb{E}\big[ \Phi_{\mathcal{X}}(\tilde{X}^-_{y'})\big] = \mathbb{E}\big[\Phi_{\mathcal{X}}(X) \langle Y, y'\rangle^-_{\mathcal{Y}}\big]/\mathbb{E}\big[\langle Y,y'\rangle^-_{\mathcal{Y}}\big]$ are well-defined, $ \Phi_{\mathcal{X}}(\tilde{X}^+_{y'})$ and $ \Phi_{\mathcal{X}}(\tilde{X}^-_{y'})$ are Pettis integrable and, furthermore, (\ref{Pr3.2}) is equivalent to
		$$\mathbb{E}\big[ \Phi_{\mathcal{X}}(\tilde{X}^+_{y'})\big]  = \mathbb{E}\big[ \Phi_{\mathcal{X}}(\tilde{X}^-_{y'})\big].$$
		Since $c_{\mathcal{X}}$ is characteristic, $\tilde{\mathbb{P}}^+_{y'} = \tilde{\mathbb{P}}^-_{y'}$ and, thus,
		$ \mathbb{E}\big[\langle Y, y'\rangle_{\mathcal{Y}}\ind{X\in B_{\mathcal{X}}}\big] = 0$ for any measurable set $B_{\mathcal{X}}\in \mathcal{B}(\mathcal{X})$. By definition, and uniqueness, of conditional expectations on $\sigma(X)$ , $\mathbb{E}[\langle Y, y'\rangle_{\mathcal{Y}}|X] = 0$ a.s. 
		
		If $\mathbb{E}[Y|X]$ is well-defined as a Pettis integral-based conditional expectation, $\mathbb{E}\big[\mathbb{E}[Y|X]\ind{X\in B_{\mathcal{X}}}\big] = \mathbb{E}\big[Y\ind{X\in B_{\mathcal{X}}}\big]$ holds in Pettis's sense for all $B_{\mathcal{X}}\in \mathcal{B}(\mathcal{X})$. Furthermore, it follows from the definition of Pettis integrals that
		$$\mathbb{E}\big[\langle\mathbb{E}[Y|X], y'\rangle_{\mathcal{Y}}\ind{X\in B_{\mathcal{X}}}\big] = \big\langle \mathbb{E}\big[\mathbb{E}[Y|X]\ind{X\in B_{\mathcal{X}}}\big], y'\big\rangle_{\mathcal{Y}} =  \big\langle \mathbb{E}\big[Y\ind{X\in B_{\mathcal{X}}}\big], y'\big\rangle_{\mathcal{Y}} = \mathbb{E}\big[\langle Y, y'\rangle_{\mathcal{Y}}\ind{X\in B_{\mathcal{X}}}\big],$$
		and, by the uniqueness of $\mathbb{E}[\langle Y, y'\rangle_{\mathcal{Y}}|X]$, we deduce that 
		and $\mathbb{E}[\langle Y, y'\rangle_{\mathcal{Y}}|X] = \langle \mathbb{E}[Y|X], y'\rangle_{\mathcal{Y}}$ a.s. Then, we have proved that $\text{\normalfont HSIC}(X,Y,c_{\mathcal{X}}, \langle\star,\cdot\rangle_{\mathcal{Y}}) = 0$ if and only if $\langle \mathbb{E}[Y|X], y'\rangle_{\mathcal{Y}} = 0$ a.s. for all $y'\in \mathcal{Y}$ and, ultimately, if and only if $\mathbb{E}[Y|X] = 0$ a.s.
	\end{proof}
	
	\section{Proofs of results in Section 3}
	\label{app2}
	
	\subsection{Preliminary results}
	\label{app2.1}
	
	\begin{lemma}\label{Lm2}
		Let $\mathcal{S}_1,\mathcal{S}_2$ be Polish spaces and let $S_1,S_2$ be random variables taking values in $\mathcal{S}_1$ and $\mathcal{S}_2$, respectively. Then, there exists a measurable function $G:\mathcal{S}_1\times [0,1]\rightarrow \mathcal{S}_2$ and random variable $U$, independent of $S_1$ and uniformly distributed in $[0,1]$, such that $(S_1,S_2) = (S_1,G(S_1,U))$ a.s.
	\end{lemma}	
	\begin{proof}[\textup{\textbf{Proof of Lemma \ref{Lm2}}}]
		Let $\mathcal{S} = \mathcal{S}_1\times \mathcal{S}_2$, which is also a Polish space. For some fixed $s_1\in \mathcal{S}_1$ and $s_2\in \mathcal{S}_2$, let $S = (s_1,S_2)$ and $S'= (S_1,s_2)$. By Theorem 1 in \cite{Skorokhod1977}, there is some $G'$ and a uniformly distributed random variable $U$, independent of $S'$, such that $S = G'(S',U)$ a.s. Let $\pi_{\mathcal{S}_2}(s_1',s_2') = s_2'$ and $G(s'_1,u) = \pi_{\mathcal{S}_2}(G'((s'_1,s_2),u))$ for all $s'_1\in \mathcal{S}_1$, $s'_2\in \mathcal{S}_2$ and $u\in [0,1]$. Then, $S_2 = \pi_{\mathcal{S}_2}(G'(S',U)) = G(S_1,U)$ a.s. and, thus, $(S_1,S_2) = (S_1,G(S',U))$ a.s. Finally, note that $U$ is independent of $S_1$ since the latter is a measurable function of $S'$.
	\end{proof}
	
	\begin{proposition}
		\label{pr4}
		Assume $\{X_i\}_{i\in \mathbb{Z}}$ and $\{Y_i\}_{i\in \mathbb{Z}}$ are $\mathbb{P}$ $\text{\normalfont NED}$ on $\{\eta_i\}_{i\in \mathbb{Z}}$, where the latter is defined as in Section \ref{sec3.1}. If (T1) and (T2) hold,  $\{\Phi_{\mathcal{X}c}(X_i)\}_{i\in \mathbb{Z}}$, $\{\Phi_{\mathcal{Y}c}(Y_i)\}_{i\in \mathbb{Z}}$ and $\{\Phi_{\mathcal{X}c}(X_i)\otimes\Phi_{\mathcal{Y}c}(Y_i) \}_{i\in \mathbb{Z}}$ are  $L^1$ $\text{\normalfont NED}$ on $\{\eta_i\}_{i\in \mathbb{Z}}$ with respect to the norms in $\mathcal{C}_{\mathcal{X}}$, $\mathcal{C}_{\mathcal{Y}}$ and $\mathcal{HS}$, with coefficients $\{\tilde{a}_j\}_{j=0}^\infty$ defined as
		$$ \tilde{a}_j =M_{\tilde{a}}\!\left(P_{j}^{\frac{1+\delta}{2+\delta}} + \varrho_{}(\epsilon_{j})\right),$$
		for some positive constant $M_{\tilde{a}}$.
	\end{proposition}
	\begin{proof}[\textup{\textbf{Proof of Proposition \ref{pr4}}}]
		Let $\{\eta_i\}_{i\in \mathbb{Z}}\equiv (\{\eta_i\}_{i\in \{-j,\ldots, j\}}, \{\eta_i\}_{i\in \mathbb{Z}\backslash\{-j,\ldots, j\}})$, where $\equiv $ denotes isomorphic identification and $\{\eta_i\}_{i\in \{-j,\ldots, j\}}$ and $\{\eta_i\}_{i\in \mathbb{Z}\backslash\{-j,\ldots, j\}}$ take values in the corresponding Polish spaces. By endowing them with any metric inducing its topology and by virtue of Lemma \ref{Lm2}, $\{\eta_i\}_{i\in \mathbb{Z}}$ can be also identified with
		$(\{\eta_i\}_{i\in \{-j,\ldots, j\}}, G(\{\eta_i\}_{i\in \{-j,\ldots, j\}}, U))$ for some random variable $U$ wich is uniformly distributed in $[0,1]$ and independent of $\{\eta_i\}_{i\in \{-j,\ldots, j\}}$. Let $U'$ be another uniformly distributed random variable independent of $\{\eta_i\}_{i\in \mathbb{Z}}$ and $U$, let $\{\eta'_i\}_{i\in \mathbb{Z}}\equiv (\{\eta_i\}_{i\in \{-j,\ldots, j\}}, G(\{\eta_i\}_{i\in \{-j,\ldots, j\}}, U'))$ (which is identically distributed to $\{\eta_i\}_{i\in \mathbb{Z}}$) and let $(X'_0,Y'_0)$ be the analogue of $(X_0,Y_0)$ based on $\{\eta'_i\}_{i\in \mathbb{Z}}$. Then,
		\begin{align}
			&\ \mathbb{E}\!\left[|| \Phi_{\mathcal{X}c}(X_0)\otimes\Phi_{\mathcal{Y}c}(Y_0) - \mathbb{E}[\Phi_{\mathcal{X}c}(X_0)\otimes\Phi_{\mathcal{Y}c}(Y_0)|\eta_{-j},\ldots,\eta_j]||_{\mathcal{HS}}\right]\nonumber\\
			&\le \mathbb{E}\!\left[|| \Phi_{\mathcal{X}c}(X_0)\otimes\Phi_{\mathcal{Y}c}(Y_0) - \Phi_{\mathcal{X}c}(X'_0)\otimes\Phi_{\mathcal{Y}c}(Y'_0)||_{\mathcal{HS}}\right]\nonumber\\
			&\quad +\mathbb{E}\!\left[|| \Phi_{\mathcal{X}c}(X'_0)\otimes\Phi_{\mathcal{Y}c}(Y'_0) - \mathbb{E}[\Phi_{\mathcal{X}c}(X_0)\otimes\Phi_{\mathcal{Y}c}(Y_0)|\eta_{-j},\ldots,\eta_j]||_{\mathcal{HS}}\right]\nonumber\\
			&\le \mathbb{E}\!\left[|| \Phi_{\mathcal{X}c}(X_0)\otimes\Phi_{\mathcal{Y}c}(Y_0) - \Phi_{\mathcal{X}c}(X'_0)\otimes\Phi_{\mathcal{Y}c}(Y'_0)||_{\mathcal{HS}}\right]\nonumber\\
			&\quad +\mathbb{E}\!\left[|| \mathbb{E}[\Phi_{\mathcal{X}c}(X'_0)\otimes\Phi_{\mathcal{Y}c}(Y'_0) - \Phi_{\mathcal{X}c}(X_0)\otimes\Phi_{\mathcal{Y}c}(Y_0)|\eta_{-j},\ldots,\eta_j,U']||_{\mathcal{HS}}\right]\nonumber\\
			&\le \mathbb{E}\!\left[|| \Phi_{\mathcal{X}c}(X_0)\otimes\Phi_{\mathcal{Y}c}(Y_0) - \Phi_{\mathcal{X}c}(X'_0)\otimes\Phi_{\mathcal{Y}c}(Y'_0)||_{\mathcal{HS}}\right]\nonumber\\
			&\quad +\mathbb{E}\!\left[\mathbb{E}[|| \Phi_{\mathcal{X}c}(X'_0)\otimes\Phi_{\mathcal{Y}c}(Y'_0) - \Phi_{\mathcal{X}c}(X_0)\otimes\Phi_{\mathcal{Y}c}(Y_0)||_{\mathcal{HS}} |\eta_{-j},\ldots,\eta_j,U']\right]\nonumber\\
			&\le 2\mathbb{E}\!\left[|| \Phi_{\mathcal{X}c}(X_0)\otimes\Phi_{\mathcal{Y}c}(Y_0) - \Phi_{\mathcal{X}c}(X'_0)\otimes\Phi_{\mathcal{Y}c}(Y'_0)||_{\mathcal{HS}}\right], \label{Pr1.1}
		\end{align}
		where we have implicitly used that, if (T1) and (T2) hold, $ \Phi_{\mathcal{X}c}(X_0)\otimes\Phi_{\mathcal{Y}c}(Y_0)$ and $ \Phi_{\mathcal{X}c}(X'_0)\otimes\Phi_{\mathcal{Y}c}(Y'_0)$ are Bochner integrable (see Remark \ref{remark2}) and, therefore, $\mathcal{HS}$-valued conditional expectations of these variables exist.
		
		Let $A = A_{\mathcal{X}}\cup A_{\mathcal{Y}}$ where $A_{\mathcal{X}} = \{d_{\mathcal{X}}(X_0,\Pi_{\mathcal{X}j}(\eta_{-j},\ldots, \eta_j))>\epsilon_{j}\}\cup \{d_{\mathcal{X}}(X'_0,\Pi_{\mathcal{X}j}(\eta_{-j},\ldots, \eta_j))>\epsilon_{j}\}$, with $\epsilon_{j}$ defined as in (\ref{5}), and $A_{\mathcal{Y}} = \{ d_{\mathcal{Y}}(Y_0,\Pi_{\mathcal{Y}j}(\eta_{-j},\ldots, \eta_j))>\epsilon_{j}\}\cup \{d_{\mathcal{Y}}(Y'_0,\Pi_{\mathcal{Y}j}(\eta_{-j},\ldots, \eta_j))>\epsilon_{j}\}$. 
		Since $\{\eta'_i\}_{i\in \mathbb{Z}}$ and $\{\eta_i\}_{i\in \mathbb{Z}}$ are identically distributed and, by construction, $\eta'_i = \eta_i$ for all $i\in \{-j,\ldots,j\}$, it follows that
		$\mathbb{P}(A)\le \mathbb{P}(A_{\mathcal{X}}) + \mathbb{P}(A_{\mathcal{Y}})\le 4P_{j}$, where $P_{j}$ is also defined as in (\ref{5}). By Minkowski's inequality and the same distribution of $(X'_0,Y'_0)$ and $(X_0,Y_0)$,
		\begin{align}
			&\ \mathbb{E}\!\left[|| \Phi_{\mathcal{X}c}(X_0)\otimes\Phi_{\mathcal{Y}c}(Y_0) - \Phi_{\mathcal{X}c}(X'_0)\otimes\Phi_{\mathcal{Y}c}(Y'_0)||_{\mathcal{HS}}\ind{A}\right]
			\le2\mathbb{E}\!\left[|| \Phi_{\mathcal{X}c}(X)\otimes\Phi_{\mathcal{Y}c}(Y) ||^{2+\delta}_{\mathcal{HS}}\right]^\frac{1}{2+\delta}\mathbb{P}(A)^{\frac{1+\delta}{2+\delta}},\label{Pr1.2}
		\end{align}
		where the $(2+\delta)$th order moment of $\Phi_{\mathcal{X}c}(X)\otimes\Phi_{\mathcal{Y}c}(Y)$ is finite by (T2). Furthermore, letting $A^c$ denote the complement of $A$, it follows from the Cauchy-Schwarz inequality that
		\begin{align} 
			&\ \mathbb{E}\!\left[|| \Phi_{\mathcal{X}c}(X_0)\otimes\Phi_{\mathcal{Y}c}(Y_0) - \Phi_{\mathcal{X}c}(X'_0)\otimes\Phi_{\mathcal{Y}c}(Y'_0)||_{\mathcal{HS}}\ind{A^c}\right]\nonumber\\
			&\le \mathbb{E}\!\left[||\Phi_{\mathcal{X}c}(X_0)||_{\mathcal{C}_{\mathcal{X}}}d^{\frac{1}{2}}_{\Phi_{\mathcal{Y}}}(Y_0,Y_0')\ind{A^c}\right] + \mathbb{E}\!\left[||\Phi_{\mathcal{Y}c}(Y'_0)||_{\mathcal{C}_{\mathcal{Y}}}d^{\frac{1}{2}}_{\Phi_{\mathcal{X}}}(X_0,X_0')\ind{A^c}\right]\nonumber\\ 
			&\le \mathbb{E}\!\left[||\Phi_{\mathcal{X}c}(X_0)||^2_{\mathcal{C}_{\mathcal{X}}}\right]^\frac{1}{2}\mathbb{E}\!\left[d_{\Phi_{\mathcal{Y}}}(Y_0,Y_0')\ind{A^c}\right]^\frac{1}{2} + \mathbb{E}\!\left[||\Phi_{\mathcal{Y}c}(Y_0)||^2_{\mathcal{C}_{\mathcal{Y}}}\right]^\frac{1}{2}\mathbb{E}\!\left[d_{\Phi_{\mathcal{X}}}(X_0,X_0')\ind{A^c}\right]^\frac{1}{2}.\label{Pr1.3}
		\end{align}
		As $d_{\Phi_{\mathcal{Y}}}$ and $d_{\Phi_{\mathcal{X}}}$ are squared distances, $d_{\Phi_{\mathcal{X}}}(X_0,X_0') \le 2d_{\Phi_{\mathcal{X}}}(X_0,\Pi_{\mathcal{X}j}(\eta_{-j},\ldots, \eta_j)) + 2d_{\Phi_{\mathcal{X}}}(X'_0,\Pi_{\mathcal{X}j}(\eta_{-j},\ldots, \eta_j))$ and $d_{\Phi_{\mathcal{X}}}(Y_0,Y_0') \le 2d_{\Phi_{\mathcal{Y}}}(Y_0,\Pi_{\mathcal{Y}j}(\eta_{-j},\ldots, \eta_j)) + 2d_{\Phi_{\mathcal{Y}}}(Y'_0,\Pi_{\mathcal{Y}j}(\eta_{-j},\ldots, \eta_j))$.
		Thus, we have, by (T1),
		\begin{align}
			&\ \mathbb{E}\!\left[d_{\Phi_{\mathcal{Y}}}(Y_0,Y_0')\ind{A^c}\right]^\frac{1}{2} + \mathbb{E}\!\left[d_{\Phi_{\mathcal{X}}}(X_0,X_0')\ind{A^c}\right]^\frac{1}{2}\nonumber\\
			&\le 2\left(\mathbb{E}\!\left[d_{\Phi_{\mathcal{Y}}}(Y_0,\Pi_{\mathcal{Y}j}(\eta_{-j},\ldots, \eta_j))\ind{A^c}\right]^\frac{1}{2} + \mathbb{E}\!\left[d_{\Phi_{\mathcal{X}}}(X_0,\Pi_{\mathcal{X}j}(\eta_{-j},\ldots, \eta_j)) \ind{A^c}\right]^\frac{1}{2}\right)\le 4\varrho_{}(\epsilon_{j}).\label{Pr1.4}
		\end{align}
		
		Putting (\ref{Pr1.1})--(\ref{Pr1.4}) together, there is some constant $M_1>0$ such that
		$$\mathbb{E}\!\left[|| \Phi_{\mathcal{X}c}(X_0)\otimes\Phi_{\mathcal{Y}c}(Y_0) - \mathbb{E}[\Phi_{\mathcal{X}c}(X_0)\otimes\Phi_{\mathcal{Y}c}(Y_0)|\eta_{-j},\ldots,\eta_j]||_{\mathcal{HS}}\right]\\
		\le M_1\left(P_{j}^{\frac{1+\delta}{2+\delta}} + \varrho_{}(\epsilon_{j})\right).$$
		By similar arguments, there are some constants $M_2,M_3>0$ such that
		$$\mathbb{E}\!\left[||  \Phi_{\mathcal{X}c}(X_0) - \mathbb{E}[\Phi_{\mathcal{X}c}(X_0)|\eta_{-j},\ldots,\eta_j]||_{\mathcal{C}_{\mathcal{X}}}\right] \le 2\mathbb{E}\!\left[d^{\frac{1}{2}}_{\Phi_{\mathcal{X}}}(X_0,X_0')\right]
		\le M_2\left(P_{j}^{\frac{1+\delta}{2+\delta}} + \varrho_{}(\epsilon_{j})\right),$$
		and, furthermore,
		$$\mathbb{E}\!\left[||  \Phi_{\mathcal{Y}c}(Y_0) - \mathbb{E}[\Phi_{\mathcal{Y}c}(Y_0)|\eta_{-j},\ldots,\eta_j]||_{\mathcal{C}_{\mathcal{Y}}}\right] \le 2\mathbb{E}\!\left[d^{\frac{1}{2}}_{\Phi_{\mathcal{Y}}}(Y_0,Y_0')\right]
		\le M_3\left(P_{j}^{\frac{1+\delta}{2+\delta}} + \varrho_{}(\epsilon_{j})\right).$$
		Thus, $\{\Phi_{\mathcal{X}c}(X_i)\}_{i\in \mathbb{Z}}$, $\{\Phi_{\mathcal{Y}c}(Y_i)\}_{i\in \mathbb{Z}}$ and $\{\Phi_{\mathcal{X}c}(X_i)\otimes\Phi_{\mathcal{Y}c}(Y_i) \}_{i\in \mathbb{Z}}$ are $L^1$ $\text{\normalfont NED}$ on $\{\eta_i\}_{i\in \mathbb{Z}}$, with respect to the usual distances in $\mathcal{C}_{\mathcal{X}}$, $\mathcal{C}_{\mathcal{Y}}$ and $\mathcal{HS}$, considering conditional expectations on $(\eta_{-j},\ldots, \eta_j)$ as the corresponding $\Pi_{\mathcal{C}_{\mathcal{X}}j}, \Pi_{\mathcal{C}_{\mathcal{Y}}j}$ and $\Pi_{\mathcal{HS}j}$. The latter are well-defined as functions on the Polish space $E^{2j+1}$ since conditional expectations can be interpreted as Bochner integrals with respect to (regular) conditional probability measures.
	\end{proof}	
	
	\begin{theorem}\label{Th11}
		Let $\mathcal{S}$ be a separable Hilbert space and let $\{S_i\}_{i \in \mathbb{Z}}\in \mathcal{S}^{\mathbb{Z}}$ be $L^1$ $\text{\normalfont NED}$ with respect to $||\cdot||_{\mathcal{S}}$ on a strictly stationary and $\beta$-mixing process $\{\eta_i\}_{i\in \mathbb{Z}}$ whose random variables take values in some Polish space $E$. Let $\{\beta_j\}_{j=1}^\infty$ and $\{a_j\}_{j=1}^\infty$ denote the $\beta$-mixing and $L^1$ NED coefficients, respectively. Assume that, for some $\delta>0$, 
		$$\sum_{j=1}^\infty (a_j^{{\delta}/{(1+\delta)}} + ja_j + \beta_j^{{\delta}/{(2+\delta)}} + j\beta_j)<\infty$$
		and also that $\mathbb{E}[||S_0||_{\mathcal{S}}^{2+\delta}]<\infty$.  Moreover, for all $n\in\mathbb{N}^+$, let $\{r_{in}\}_{i=1}^{n}$ be a collection of real-valued Gaussian random variables, independent of $\{\eta_{i}\}_{i\in \mathbb{Z}}$, such that $\mathbb{E}[r_{in}] = 0$ and $\text{Cov}(r_{in},r_{jn}) = \rho_n(|i-j|)$, where $\rho_n:\mathbb{N} \rightarrow [0,1]$ for all $n\in \mathbb{N}^+$, $\rho_n(0) = 1$, $\rho_n(k)\rightarrow 1$ as $n \rightarrow\infty$ for all $k\ge1$ and $\sum_{k=1}^n \rho_n(k) = O(\ell_n)$ for some $\{\ell_n\}_{n\in \mathbb{N}^+}$ such that $\ell_n \rightarrow \infty$ and $\ell_n =o(n)$. Then, letting $\bar{r}_{\cdot n}$ denote the sample average of $\{r_{in}\}_{i=1}^{n}$,
		$$  \frac{1}{\sqrt{n}}\sum_{i=1}^{n}(r_{in} - \bar{r}_{\cdot n})S_i \xrightarrow{\mathcal{D^*}} N_{\mathcal{S}} \ \text{in }\mathbb{P},$$
		where $N_{\mathcal{S}}$ is a centred Gaussian distribution with covariance operator $\Gamma_{\mathcal{S}}$ defined by
		$$\langle s,\Gamma(s')\rangle_{\mathcal{S}} = \sum\limits_{i\in \mathbb{Z}} \mathbb{E}\!\left[\langle S_0 - \mathbb{E}[S_0], s \rangle_{S} \langle S_i - \mathbb{E}[S_i], s' \rangle_{S}  \right],$$
		for all $s,s' \in \mathcal{S}$.   
	\end{theorem}	
	
	\begin{proof}[\textup{\textbf{Proof of Theorem \ref{Th11}}}] 
		
		Denoting the sample average of  $\{S_i\}_{i \in \mathbb{Z}}$ by $\bar{S}$, we have
		\begin{equation}\label{Th11.1}
			\frac{1}{\sqrt{n}}\sum_{i=1}^{n}(r_{in} - \bar{r}_{\cdot n})S_i =  \frac{1}{\sqrt{n}}\sum_{i=1}^{n}r_{in} (S_i - \bar{S}) = \frac{1}{\sqrt{n}}\sum_{i=1}^{n}r_{in} (S_i - \mathbb{E}[S_i]) + \sqrt{n}\bar{r}_{\cdot n}(\mathbb{E}[S_i]-\bar{S}).
		\end{equation}
		Thus, without loss of generality, we can assume that $\mathbb{E}[S_i] = 0$. By Theorem 1.1. in \cite{Dehling2015}, $\sqrt{n}\bar{S} = O_{\mathbb{P}}(1)$ and, by definition of $\{r_{in}\}_{i=1}^{n}$, $\mathbb{E}[\bar{r}_{\cdot n}] =0 $ and, 
		$$\text{Var}(\bar{r}_{\cdot n}) = \frac{1}{n^2}\sum_{i,j=1}^n \rho_n(|i-j|) = \frac{1}{n} + \frac{2}{n^2}\sum_{k=1}^{n-1}(n-k)\rho_n(k)\le \frac{1}{n} + \frac{2}{n}\sum_{k=1}^{n-1}\rho_n(k) = O\left(\frac{1}{n}\right) + O\left(\frac{\ell_n}{n}\right) \rightarrow 0.$$
		Then, by Markov's inequality we have, for all $\epsilon>0$,
		\begin{equation}\label{Th11.2}
			\mathbb{P}\!\left(\sqrt{n}||\bar{r}_{\cdot n}\bar{S}||_{\mathcal{S}}>\epsilon|\{\eta_i\}_{i\in \mathbb{Z}} \right)\le \frac{n||\bar{S}||^2_{\mathcal{S}}}{\epsilon^2}\text{Var}(\bar{r}_{\cdot n}) = o_{\mathbb{P}}(1).
		\end{equation}
		
		Let $M_n\rightarrow \infty$ such that $M_n^4\ell_n = o(n)$ (for example, we can take $M_n^4 =\sqrt{n/\ell_n}$). Define $\tilde{S}_{in} = S_i\min\{1,M_n/||S_i||_{\mathcal{S}}\}$ so that $||\tilde{S}_{in}||_{\mathcal{S}}\le \min\{M_n,||S_i||_{\mathcal{S}}\}$ and $||\tilde{S}_{in} - S_i||_{\mathcal{S}}\le \max\{||S_i||_{\mathcal{S}}-M_n,0\}\le ||S_i||_{\mathcal{S}}$. In addition, let $\tilde{S}_{inc} = \tilde{S}_{in} - \mathbb{E}[\tilde{S}_{in}]$. Then, by the dominated convergence theorem, $\mathbb{E}[||\tilde{S}_{inc} - S_i||^{2+\delta}_{\mathcal{S}}] = o(1)$. This artificial variables are also used in the proof of Theorem 1.2. in \cite{Dehling2015} and it holds that, for all $n$, $\{\tilde{S}_{inc}\}_{i=1}^\infty$ is $L^1$ $\text{\normalfont NED}$ on $\{\eta_i\}_{i\in \mathbb{Z}}$ with coefficients $\{a_j\}_{j=0}^\infty$, which trivially implies that $\{S_i-\tilde{S}_{inc}\}_{i=1}^\infty$ is also $L^1$ $\text{\normalfont NED}$ with coefficients $\{2a_j\}_{j=0}^\infty$. Moreover, it follows from the summability assumptions on $\{a_j\}_{j=0}^\infty$ and $\{\beta_j\}_{j=0}^\infty$, as well as the first part of Lemma 2.4. in \cite{Dehling2015}, that
		\begin{align}
			\mathbb{E}\!\left[\mathbb{E}\!\left[\left|\left|\frac{1}{\sqrt{n}}\sum_{i=1}^{n}r_{in} (S_i - \tilde{S}_{inc})\right|\right|^2_{\mathcal{S}} | \{\eta_i\}_{i\in \mathbb{Z}}\right]\right] &= \frac{1}{n} \sum_{i,j=1}^n \rho_n(|i-j|)\mathbb{E}\!\left[\langle S_i - \tilde{S}_{inc},S_j - \tilde{S}_{jnc}\rangle_{\mathcal{S}} \right]\nonumber \\
			&\le \mathbb{E}\!\left[||S_i - \tilde{S}_{inc}||^2_{\mathcal{S}} \right] + 2\sum_{k=1}^{n-1}\!\left|\mathbb{E}\!\left[\langle S_0 - \tilde{S}_{0nc},S_k - \tilde{S}_{knc}\rangle_{\mathcal{S}} \right]\right|\nonumber\\
			&\le  \mathbb{E}\!\left[||S_i - \tilde{S}_{inc}||^2_{\mathcal{S}} \right] + 2\sum_{k=1}^{n-1} 2\mathbb{E}\!\left[||S_0 - \tilde{S}_{0nc}||^{2+\delta}_{\mathcal{S}} \right]^\frac{2}{2+\delta}\beta^{\frac{\delta}{2+\delta}}_{\lfloor k/3\rfloor}\nonumber\\
			&\quad +2\sum_{k=1}^{n-1} 4\mathbb{E}\!\left[||S_0 - \tilde{S}_{0nc}||^{2+\delta}_{\mathcal{S}} \right]^\frac{1}{1+\delta}(2a_{\lfloor k/3\rfloor})^{\frac{\delta}{1+\delta}} =o(1)\nonumber,
		\end{align}
		where we have implicitly used that $\rho_n$ is uniformly bounded by $1$. This implies, by Markov's inequality, that
		\begin{equation}\label{Th11.3}
			\mathbb{P}\!\left(\left|\left|\frac{1}{\sqrt{n}}\sum_{i=1}^{n}r_{in} (S_i - \tilde{S}_{inc})\right|\right|_{\mathcal{S}}>\epsilon|\{\eta_i\}_{i\in \mathbb{Z}} \right) = o_{\mathbb{P}}(1).
		\end{equation}
		
		Let now us show that 
		\begin{equation}\label{Th11.4}
			\Gamma_n = \mathbb{E}\!\left[\left(\frac{1}{\sqrt{n}}\sum_{i=1}^{n}r_{in}\tilde{S}_{inc}\otimes \frac{1}{\sqrt{n}}\sum_{i=1}^{n}r_{in}\tilde{S}_{inc}\right) | \{\eta_i\}_{i\in \mathbb{Z}} \right] = \frac{1}{n}\sum_{i=1}^n  \tilde{S}_{inc}\otimes\tilde{S}_{inc} + \frac{2}{n}\sum_{k=1}^{n-1}\sum_{i=1}^{n-k} \rho_n(k) \tilde{S}_{inc}\otimes \tilde{S}_{(i+k)nc} \rightarrow \Gamma,
		\end{equation}	
		in trace norm, in probability. To do so, since both $\Gamma_n$ and $\Gamma$ are trace class covariance operators, it suffices to show, by arguments based on almost sure convergence along subsequences, that $\langle s,\Gamma_n(s)\rangle_{\mathcal{S}}\rightarrow\langle s,\Gamma(s)\rangle_{\mathcal{S}}$ in probability for all $s\in \mathcal{S}$ such that $||s||_{\mathcal{S}} = 1$ and $\text{tr}(\Gamma_n)\rightarrow\text{tr}(\Gamma)$ in probability, where $\text{tr}$ denotes the trace of a linear operator. 
		
		If $||s||_{\mathcal{S}} = 1$, it is trivial to show that $\{\langle \tilde{S}_{inc},s\rangle_{\mathcal{S}}\}_{i=1}^\infty$ is $L^1$ $\text{\normalfont NED}$ with the same coefficients as $\{\tilde{S}_{inc}\}_{i=1}^\infty$. Thus, it follows from by the first part of Lemma 2.4. in \cite{Dehling2015} (replacing the incorrect exponent $\delta/(2+\delta)$ in the $L^1$ NED coefficients by $\delta/(1+\delta)$, as in Lemma 2.18. in \citeauthor{Borovkova2001}, \citeyear{Borovkova2001}) and the Cauchy-Schwarz inequality that, for all $k\ge 1$,
		\begin{align*}
			\frac{n-k}{n}\rho_n(k)\ind{k\le n-1}\mathbb{E}\!\left[\langle\tilde{S}_{0nc},s\rangle_{\mathcal{S}} \langle \tilde{S}_{knc},s\rangle_{\mathcal{S}}\right] &\le 2\mathbb{E}\!\left[||\tilde{S}_{0nc}||^{2+\delta}_{\mathcal{S}} \right]^\frac{2}{2+\delta}\beta^{\frac{\delta}{2+\delta}}_{\lfloor k/3\rfloor} + 4\mathbb{E}\!\left[||\tilde{S}_{0nc}||^{2+\delta}_{\mathcal{S}} \right]^\frac{1}{1+\delta}a_{\lfloor k/3\rfloor}^{\frac{\delta}{1+\delta}}\nonumber\\
			&\le 2\mathbb{E}\!\left[||S_0||^{2+\delta}_{\mathcal{S}} \right]^\frac{2}{2+\delta}\beta^{\frac{\delta}{2+\delta}}_{\lfloor k/3\rfloor} + 4\mathbb{E}\!\left[||S_0||^{2+\delta}_{\mathcal{S}} \right]^\frac{1}{1+\delta}a_{\lfloor k/3\rfloor}^{\frac{\delta}{1+\delta}},
		\end{align*}
		which is summable. Then, by the dominated convergence theorem, 
		\begin{equation}
			\langle s,\mathbb{E}[\Gamma_n](s)\rangle_{\mathcal{S}}= \mathbb{E}\!\left[\langle \tilde{S}_{0nc},s\rangle_{\mathcal{S}}^2 \right]+ 2\sum_{k=1}^{\infty}\frac{(n-k)}{n} \rho_n(k)\ind{k\le n-1}\mathbb{E}\!\left[\langle\tilde{S}_{0nc},s\rangle_{\mathcal{S}} \langle \tilde{S}_{knc},s\rangle_{\mathcal{S}}\right] \rightarrow\langle s,\Gamma(s)\rangle_{\mathcal{S}}, \label{Th11.6}
		\end{equation}
		In addition,
		\begin{align}
			&\ \mathbb{E}\!\left[\langle s,(\Gamma_n-\mathbb{E}[\Gamma_n])(s)\rangle_{\mathcal{S}}^2\right]\nonumber\\
			&= \mathbb{E}\!\left[\left|\frac{1}{n}\sum_{i=1}^n  \left(\langle \tilde{S}_{inc},s\rangle^2_{\mathcal{S}}-\mathbb{E}\!\left[\langle \tilde{S}_{inc},s\rangle^2_{\mathcal{S}}\right]\right)+ \frac{2}{n}\sum_{k=1}^{n-1}\sum_{i=1}^{n-k} \rho_n(k)\left( \langle \tilde{S}_{inc},s\rangle_{\mathcal{S}}\langle \tilde{S}_{(i+k)nc},s\rangle_{\mathcal{S}} - \mathbb{E}\!\left[\langle \tilde{S}_{inc},s\rangle_{\mathcal{S}}\langle \tilde{S}_{(i+k)nc},s\rangle_{\mathcal{S}}\right]\right)\right|^2\right]\nonumber\\
			&\le \frac{2}{n^2}\sum_{i,j=1}^n \left|\text{Cov}\!\left(\langle \tilde{S}_{inc},s\rangle_{\mathcal{S}}^2,\langle \tilde{S}_{jnc},s\rangle_{\mathcal{S}}^2\right)\right|\nonumber \\ 
			&\quad + \frac{4}{n^2}\sum_{k=1}^{n-1}\sum_{i=1}^{n-k}\sum_{{k'}=1}^{n-1}\sum_{{i'}=1}^{n-{k'}}\rho_n(k)\rho_n(k') \left|\text{Cov}\!\left(\langle \tilde{S}_{inc},s\rangle_{\mathcal{S}}\langle \tilde{S}_{(i+k)nc},s\rangle_{\mathcal{S}},\langle \tilde{S}_{{i'}nc},s\rangle_{\mathcal{S}}\langle \tilde{S}_{({i'}+{k'})nc},s\rangle_{\mathcal{S}}\right)\right|. \label{Th11.7}
		\end{align}
		As $|\langle \tilde{S}_{inc},s\rangle_{\mathcal{S}}|\le ||\tilde{S}_{inc}||_{\mathcal{S}}\le 2M_n$ holds for all $i$ and the map $t\mapsto t^2$ restricted to $|t|\le 2M_n$ is Lipschitz continuous with Lipschitz constant $4M_n$, it follows from Lemma 2.2. in \cite{Dehling2015} that $\{\langle \tilde{S}_{inc},s\rangle^2_{\mathcal{S}}\}_{i=1}^\infty$ is $L^1$ $\text{\normalfont NED}$ with coefficients $\{4M_na_j\}_{j=0}^\infty$. Then, by the second part of Lemma 2.4. in \cite{Dehling2015},
		\begin{align}
			\frac{2}{n^2}\sum_{i,j=1}^n \left|\text{Cov}\!\left(\langle \tilde{S}_{inc},s\rangle_{\mathcal{S}}^2,\langle \tilde{S}_{jnc},s\rangle_{\mathcal{S}}^2\right)\right| &= \frac{2}{n} \text{Var}\!\left(\langle \tilde{S}_{inc},s\rangle_{\mathcal{S}}^2\right) + \frac{2}{n}\sum_{k=1}^{n-1}\frac{(n-k)}{n}\!\left|\text{Cov}\!\left(\langle \tilde{S}_{0nc},s\rangle_{\mathcal{S}}^2,\langle \tilde{S}_{knc},s\rangle_{\mathcal{S}}^2\right)\right|\nonumber\\
			&\le \frac{2}{n} \mathbb{E}\!\left[\langle \tilde{S}_{inc},s\rangle_{\mathcal{S}}^2\right] + \frac{2}{n}\sum_{k=1}^{n-1}\!\left(2M_n^2\beta_{\lfloor k/3\rfloor}+4M_n(4M_na_{\lfloor k/3\rfloor})\right)= O\left(\frac{M_n^2}{n}\right). \label{Th11.8}
		\end{align}
		To bound the second term after the last inequality in (\ref{Th11.7}), we would need to derive a bound for each one of the covariances involving four inner products for every possible ordering of the indices $i,i',i+k,i'+k'$. To keep the proof short, we only analyze the case $i<i'<i+k<i'+k'$, since any alternative ordering can be treated in a similar manner. Let us denote $l = i'- i$, $l' = k-l$ and $l''= k' - l'$ so that $i' = i+l$, $i+k = i+l + l'$ and $i'+k' = i+l + l' + l''$. The sequence of $\beta$-mixing coefficients is non-increasing by definition. The same can be assumed for $\{a_{j}\}_{j =0}^\infty$ since, by the tower property of conditional expectations and the triangle inequality,
		$$\mathbb{E}\!\left[||S_0 - \mathbb{E}[S_0|\eta_{-j},\ldots,\eta_{j}]||_{\mathcal{S}}\right] \le 2\mathbb{E}\!\left[||S_0 - \mathbb{E}[S_0|\eta_{-j'},\ldots,\eta_{j'}]||_{\mathcal{S}}\right]\le 4\mathbb{E}\!\left[||S_0 -  \Pi_{\mathcal{S}j'}(\eta_{-j'},\ldots, \eta_{j'})||_{\mathcal{S}}\right],$$
		holds for all $j'\le j$ and, thus, $\{\min_{j'\le j}4a_j'\}_{j=0}^\infty$ is a valid sequence of $L^1$ NED coefficients for $\{S_i\}_{i\in \mathbb{Z}}$ which is non-increasing and satisfies the same summability assumptions as $\{a_{j}\}_{j =0}^\infty$. Then, since $\mathbb{E}[\langle \tilde{S}_{0nc},s\rangle_{\mathcal{S}}] =0$, it follows from Lemma 2.18. and the proof of Lemma 2.21. in \cite{Borovkova2001} (see also the proof of Lemma 2.24. there), taking $\delta= \infty$ in the latter result, that, for some positive constant $\tilde{M}_1$,
		\begin{align}
			&\left|\text{Cov}\!\left(\langle \tilde{S}_{inc},s\rangle_{\mathcal{S}}\langle \tilde{S}_{(i+l+l')nc},s\rangle_{\mathcal{S}},\langle \tilde{S}_{(i+l)nc},s\rangle_{\mathcal{S}}\langle \tilde{S}_{(i+l+l'+l'')nc},s\rangle_{\mathcal{S}}\right)\right|\nonumber\\
			&\le \left|\text{Cov}\!\left(\langle \tilde{S}_{inc},s\rangle_{\mathcal{S}}\langle \tilde{S}_{(i+l)nc},s\rangle_{\mathcal{S}},\langle \tilde{S}_{(i+l+l')nc},s\rangle_{\mathcal{S}}\langle \tilde{S}_{(i+l+l'+l'')nc},s\rangle_{\mathcal{S}}\right)\right| \nonumber\\
			&\quad+ \left|\mathbb{E}\!\left[\langle \tilde{S}_{inc},s\rangle_{\mathcal{S}}\langle \tilde{S}_{(i+l+l')nc},s\rangle_{\mathcal{S}}\right]\mathbb{E}\!\left[\langle \tilde{S}_{(i+l)nc},s\rangle_{\mathcal{S}}\langle \tilde{S}_{(i+l+l'+l'')nc},s\rangle_{\mathcal{S}}\right]\right| \nonumber\\
			&\quad +\left|\mathbb{E}\!\left[\langle \tilde{S}_{inc},s\rangle_{\mathcal{S}}\langle \tilde{S}_{(i+l)nc},s\rangle_{\mathcal{S}}\right]\mathbb{E}\!\left[\langle \tilde{S}_{(i+l+l')nc},s\rangle_{\mathcal{S}}\langle \tilde{S}_{(i+l+l'+l'')nc},s\rangle_{\mathcal{S}}\right]\right| \nonumber\\
			&\le \tilde{M}_1 M_n^4 \Big((\beta_{\lfloor l'/3\rfloor} + a_{\lfloor l'/3\rfloor}) +(\beta_{\lfloor (l+l')/3\rfloor} + a_{\lfloor (l+l')/3\rfloor}) (\beta_{\lfloor (l'+l'')/3\rfloor} + a_{\lfloor (l'+l'')/3\rfloor}) +(\beta_{\lfloor l/3\rfloor} + a_{\lfloor l/3\rfloor}) (\beta_{\lfloor l''/3\rfloor} + a_{\lfloor l''/3\rfloor})\Big)\nonumber\\
			&\le 2\tilde{M}_1 M_n^4 \Big((\beta_{\lfloor l'/3\rfloor} + a_{\lfloor l'/3\rfloor}) +(\beta_{\lfloor l/3\rfloor} + a_{\lfloor l/3\rfloor}) (\beta_{\lfloor l''/3\rfloor} + a_{\lfloor l''/3\rfloor})\Big),\label{Th11.9}
		\end{align}
		Let $p=\max\{l, l', l''\} = \max\{i'-i, i+k-i', i'+k'-i-k\}$. There is some positive constant $\tilde{M}_2$ such that
		\begin{align}
			&\frac{4}{n^2}\sum_{1\le i < i'< i+k<i'+k'\le n}\rho_n(k)\rho_n(k') \left|\text{Cov}\!\left(\langle \tilde{S}_{inc},s\rangle_{\mathcal{S}}\langle \tilde{S}_{(i+k)nc},s\rangle_{\mathcal{S}},\langle \tilde{S}_{{i'}nc},s\rangle_{\mathcal{S}}\langle \tilde{S}_{({i'}+{k'})nc},s\rangle_{\mathcal{S}}\right)\right|\ind{p=i+k-i'} \nonumber\\
			&\le \frac{\tilde{M}_2M_n^4}{n^2}\sum_{i=1}^{n-3}\sum_{l=1}^{n-i-2}\sum_{l'=1}^{n-i-l-1}\sum_{l''=1}^{n-i-l-l'}\rho_n(l+l')\rho_n(l'+l'') (\beta_{\lfloor l'/3\rfloor} + a_{\lfloor l'/3\rfloor})\ind{p=l'}\nonumber\\
			&\quad + \frac{\tilde{M}_2M_n^4}{n^2}\sum_{i=1}^{n-3}\sum_{l=1}^{n-i-2}\sum_{l'=1}^{n-i-l-1}\sum_{l''=1}^{n-i-l'-l''}\rho_n(l+l')\rho_n(l'+l'') (\beta_{\lfloor l/3\rfloor} + a_{\lfloor l/3\rfloor})(\beta_{\lfloor l''/3\rfloor} + a_{\lfloor l''/3\rfloor})\ind{p=l'}\nonumber \\
			&\le \frac{\tilde{M}_2M_n^4}{n}\sum_{l'=1}^{n}(\beta_{\lfloor  l'/3\rfloor} + a_{\lfloor l'/3\rfloor})\left(\sum_{l=1}^{n-l'}\rho_n(l+l')\right)\left(\sum_{l''=1}^{l'}\rho_n(l'+l'')\right)\nonumber  \\
			&\quad + \frac{\tilde{M}_2M_n^4}{n}\sum_{l=1}^{n}\sum_{l''=1}^{n}(\beta_{\lfloor l/3\rfloor} + a_{\lfloor l/3\rfloor})(\beta_{\lfloor l''/3\rfloor} + a_{\lfloor l''/3\rfloor})\left(\sum_{l'=1}^{n-l}\rho_n(l+l')\rho_n(l'+l'')\right)\nonumber \\
			&\le \frac{\tilde{M}_2M_n^4}{n}\left(\sum_{l=1}^{n}\rho_n(l)\right)\left(\left(\sum_{l=1}^{n}l(\beta_{\lfloor l/3\rfloor} + a_{\lfloor l/3\rfloor})\right) + \left(\sum_{l=1}^{n}(\beta_{\lfloor l/3\rfloor} + a_{\lfloor l/3\rfloor})\right)^2\right)= O\left(\frac{M_n^4\ell_n}{n}\right). \label{Th11.10}
		\end{align}
		From Lemmas 2.18. and 2.21. in \cite{Borovkova2001} and the monotonicity of $\{a_{j}\}_{j=0}^\infty$ and  $\{\beta_{j}\}_{j=0}^\infty$, we deduce that
		\begin{align}
			&\left|\text{Cov}\!\left(\langle \tilde{S}_{inc},s\rangle_{\mathcal{S}}\langle \tilde{S}_{(i+l+l')nc},s\rangle_{\mathcal{S}},\langle \tilde{S}_{(i+l)nc},s\rangle_{\mathcal{S}}\langle \tilde{S}_{(i+l+l'+l'')nc},s\rangle_{\mathcal{S}}\right)\right|\nonumber \nonumber \\
			&\le \left|\mathbb{E}\!\left[\langle \tilde{S}_{inc},s\rangle_{\mathcal{S}}\langle \tilde{S}_{(i+l)nc},s\rangle_{\mathcal{S}}\langle \tilde{S}_{(i+l+l')nc},s\rangle_{\mathcal{S}}\langle \tilde{S}_{(i+l+l'+l'')nc},s\rangle_{\mathcal{S}}\right]\right|\nonumber \nonumber \\
			&\quad + \left|\mathbb{E}\!\left[\langle \tilde{S}_{inc},s\rangle_{\mathcal{S}}\langle \tilde{S}_{(i+l+l')nc},s\rangle_{\mathcal{S}}\right]\mathbb{E}\!\left[\langle \tilde{S}_{(i+l)nc},s\rangle_{\mathcal{S}}\langle \tilde{S}_{(i+l+l'+l'')nc},s\rangle_{\mathcal{S}}\right]\right|\nonumber \nonumber \\
			&\le \tilde{M}_3M_n^4(\beta_{\lfloor l /3\rfloor} + a_{\lfloor l/3\rfloor}) + \tilde{M}_3M_n^2(\beta_{\lfloor (l + l')/3\rfloor}+a_{\lfloor (l + l')/3\rfloor})(2M_n)^2\le 5\tilde{M}_3M_n^4(\beta_{\lfloor l /3\rfloor} + a_{\lfloor l/3\rfloor}),\label{Th11.11}
		\end{align}
		holds for some positive constant $\tilde{M}_3$ and, analogously,
		\begin{equation}\label{Th11.12}
			\left|\text{Cov}\!\left(\langle \tilde{S}_{inc},s\rangle_{\mathcal{S}}\langle \tilde{S}_{(i+l+l')nc},s\rangle_{\mathcal{S}},\langle \tilde{S}_{(i+l)nc},s\rangle_{\mathcal{S}}\langle \tilde{S}_{(i+l+l'+l'')nc},s\rangle_{\mathcal{S}}\right)\right| \le 5\tilde{M}_3 M_n^4 (\beta_{\lfloor l''/3\rfloor} + a_{\lfloor l''/3\rfloor}).
		\end{equation}
		Therefore, it is easy to verify that the asymptotic bound in (\ref{Th11.10}) also holds for $p = l'' = i'+k'-i-k$ and $p=l = i'-i$. Consequently,
		\begin{equation*}
			\frac{4}{n^2}\sum_{1\le i < i'< i+k<i'+k'\le n}\rho_n(k)\rho_n(k') \left|\text{Cov}\!\left(\langle \tilde{S}_{inc},s\rangle_{\mathcal{S}}\langle \tilde{S}_{(i+k)nc},s\rangle_{\mathcal{S}},\langle \tilde{S}_{{i'}nc},s\rangle_{\mathcal{S}}\langle \tilde{S}_{({i'}+{k'})nc},s\rangle_{\mathcal{S}}\right)\right| = O\left(\frac{M_n^4\ell_n}{n}\right).
		\end{equation*}
		The same asymptotic bound can be established for any other ordering of the indices $i,i',i+k,i+k'$; i.e, 
		\begin{equation}
			\frac{4}{n^2}\sum_{k=1}^{n-1}\sum_{i=1}^{n-k}\sum_{{k'}=1}^{n-1}\sum_{{i'}=1}^{n-{k'}}\rho_n(k)\rho_n(k') \left|\text{Cov}\!\left(\langle \tilde{S}_{inc},s\rangle_{\mathcal{S}}\langle \tilde{S}_{(i+k)nc},s\rangle_{\mathcal{S}},\langle \tilde{S}_{{i'}nc},s\rangle_{\mathcal{S}}\langle \tilde{S}_{({i'}+{k'})nc},s\rangle_{\mathcal{S}}\right)\right| = O\left(\frac{M_n^4\ell_n}{n}\right). \label{Th11.13}
		\end{equation}
		Putting (\ref{Th11.7}), (\ref{Th11.8}) and (\ref{Th11.13}) together, we conclude that $\mathbb{E}\big[\langle s,(\Gamma_n-\mathbb{E}[\Gamma_n])(s)\rangle_{\mathcal{S}}^2\big] \rightarrow 0$ which, in combination with Markov's inequality, Slutsky's theorem and (\ref{Th11.6}), implies that $\langle s,\Gamma_n(s)\rangle_{\mathcal{S}}\rightarrow\langle s,\Gamma(s)\rangle_{\mathcal{S}}$ in probability.
		
		By virtue of the dominated convergence theorem for sums and also Lemma 2.4. in \cite{Dehling2015}, we have
		$$\mathbb{E}\!\left[\text{tr}(\Gamma_n)\right]=\text{tr}(\mathbb{E}[\Gamma_n]) = \mathbb{E}\!\left[|| \tilde{S}_{0nc}||_{\mathcal{S}}^2 \right]+ 2\sum_{k=1}^\infty \frac{(n-k)}{n}\rho_n(k)\ind{k\le n-1}\mathbb{E}\!\left[\langle \tilde{S}_{0nc}, \tilde{S}_{knc}\rangle_{\mathcal{S}}\right] \rightarrow \text{tr}(\Gamma).$$
		Moreover, the covariances $\text{Cov}\big(\langle\tilde{S}_{inc},\tilde{S}_{(i+k)nc}\rangle_{\mathcal{S}},\langle \tilde{S}_{i'nc},\tilde{S}_{(i'+k')nc}\rangle_{\mathcal{S}}\big)$ satisfy bounds of the same type as those established in (\ref{Th11.9}), (\ref{Th11.11}), and (\ref{Th11.12}) (see Lemma 2.10. in \citeauthor{Dehling2015}, \citeyear{Dehling2015}). Consequently, the same arguments used to prove $\mathbb{E}\big[\langle s,(\Gamma_n-\mathbb{E}[\Gamma_n])(s)\rangle_{\mathcal{S}}^2\big] \rightarrow 0$ also yield $\mathbb{E}\big[\text{tr}(\Gamma_n - \mathbb{E}[\Gamma_n])^2\big]\rightarrow 0$. As a result, $\text{tr}(\Gamma_n)\rightarrow \text{tr}(\Gamma)$ in probability and, ultimately, (\ref{Th11.4}) holds in trace norm, in probability.
		
		For each subsequence $\{n_{k}\}_{k\in \mathbb{N^+}}$, there is a further subsequence $\{n_{k_l}\}_{l\in \mathbb{N^+}}$ such that (\ref{Th11.1}), (\ref{Th11.2}), (\ref{Th11.3}) and (\ref{Th11.4}) hold almost surely. Along this subsequence, $\Gamma_{n_{k_l}}\rightarrow \Gamma$ in trace norm a.s. and, since $\{r_{in}\}_{i=1}^n$ are assumed to be Gaussian, $\sum_{i=1}^{n_{k_l}}r_{in_{k_l}} S_i/\sqrt{n_{k_l}}$ is also Gaussian for all $l$. Then, by Slustsky's theorem,
		$\sum_{i=1}^{n_{k_l}}r_{in_{k_l}}\langle s,S_i\rangle_{\mathcal{S}}/\sqrt{n_{k_l}}$ converges in distribution to $N\big(0,\langle s,\Gamma(s)\rangle_{\mathcal{S}}\big)$ for all $s\in \mathcal{S}$,
		for almost all realizations of $\{\eta_i\}_{i\in \mathbb{Z}}$. This implies, by virtue of Lemma 3.1. and Remark 3.3. in \cite{Chen1998}, that
		$\sum_{i=1}^{n_{k_l}}r_{in_{k_l}}S_i/\sqrt{n_{k_l}}$ converges weakly to $N_{\mathcal{HS}}$ a.s. Applying Slutsky's theorem again, we conclude that $\sum_{i=1}^{n_{k_l}}(r_{in_{k_l}}-\bar{r}_{\cdot n_{k_l}})S_i/\sqrt{n_{k_l}}$ also converges weakly to $N_{\mathcal{HS}}$ a.s. Since $\{n_{k}\}_{k\in \mathbb{N^+}}$ is arbitrary, the weak convergence in probability of $\sum_{i=1}^{n}(r_{in} - \bar{r}_{\cdot n})S_i/\sqrt{n}$ to $N_{\mathcal{HS}}$ holds.
	\end{proof}	
	
	\subsection{Proof of Theorems 1-3}
	\label{app2.2}
	
	\begin{proof}[\textup{\textbf{Proof of Theorem \ref{Th1}}}]
		By assumption (see Remark \ref{remark2}), $\Phi_{\mathcal{X}}(X_i)$, $\Phi_{\mathcal{Y}}(Y_i)$ and $\Phi_{\mathcal{X}}(X_i)\otimes \Phi_{\mathcal{Y}}(Y_i)$ are Bochner integrable for all $i\in \mathbb{Z}$. Then, the following decomposition is immediate:
		\begin{align}
			&\ \frac{1}{n}\sum_{i=1}^n (\Phi_{\mathcal{X}}(X_i)- \bar{\Phi}_{\mathcal{X}}(X)) \otimes  (\Phi_{\mathcal{Y}}(Y_i)- \bar{\Phi}_{\mathcal{Y}}(Y))\nonumber \\
			&= \frac{1}{n}\sum_{i=1}^n (\Phi_{\mathcal{X}}(X_i)- \mathbb{E}[\Phi_{\mathcal{X}}(X)]) \otimes  (\Phi_{\mathcal{Y}}(Y_i)- \mathbb{E}[\Phi_{\mathcal{Y}}(Y)])-(\mathbb{E}[\Phi_{\mathcal{X}}(X)]- \bar{\Phi}_{\mathcal{X}}(X))\otimes (\mathbb{E}[\Phi_{\mathcal{Y}}(Y)] - \bar{\Phi}_{\mathcal{Y}}(Y)).\label{Th1.1}
		\end{align}
		Furthermore,
		\begin{align*}
			&\ \mathbb{E}\Big[||(\Phi_{\mathcal{X}}(X_i)- \mathbb{E}[\Phi_{\mathcal{X}}(X)]) \otimes  (\Phi_{\mathcal{Y}}(Y_i)- \mathbb{E}[\Phi_{\mathcal{Y}}(Y)])||_{\mathcal{HS}}\Big]\\
			&\le \mathbb{E}\!\left[||\Phi_{\mathcal{X}}(X)||_{\mathcal{C}_{\mathcal{X}}} ||\Phi_{\mathcal{Y}}(Y)||_{\mathcal{C}_{\mathcal{Y}}}\right] + 3\mathbb{E}\!\left[||\Phi_{\mathcal{X}}(X)||_{\mathcal{C}_{\mathcal{X}}}\big]\mathbb{E}\big[ ||\Phi_{\mathcal{Y}}(Y)||_{\mathcal{C}_{\mathcal{Y}}}\right] < \infty.
		\end{align*}
		
		Let $T^i$ denote the $i$-fold self composition of the shift operator $T$. Moreover, let $(\mathbf{X}_0,\mathbf{Y}_0):(\mathcal{X}\times \mathcal{Y})^\mathbb{Z}\rightarrow\mathcal{X}\times \mathcal{Y}$ be defined as $(\mathbf{X}_0,\mathbf{Y}_0)(\{(x_i,y_i)\}_{i\in \mathbb{Z}}) = (x_0,y_0)$, so that $(X_i,Y_i) = (\mathbf{X}_0,\mathbf{Y}_0)(T^i(\{(X_i,Y_i)\}_{i\in \mathbb{Z}}))$. It is clear that $(\mathbf{X}_0,\mathbf{Y}_0)$ is measurable by construction. In addition, since $\mathcal{HS}$ is a separable Hilbert space and $\{(X_i,Y_i)\}_{i\in \mathbb{Z}}$ is ergodic, we can apply the ergodic theorem for Bochner integrable functions taking values in reflexive Banach spaces (Theorem 6 in \citeauthor{Beck1957}, \citeyear{Beck1957}) to conclude that
		\begin{equation}
			\frac{1}{n}\sum_{i=1}^n \Phi_{\mathcal{X}c}(X_i) \otimes  \Phi_{\mathcal{Y}c}(Y_i) \xrightarrow{a.s.} \mathbb{E}[ \Phi_{\mathcal{X}c}(X) \otimes  \Phi_{\mathcal{Y}c}(Y)]. \label{Th1.2}
		\end{equation}
		Analogously, $\bar{\Phi}_{\mathcal{X}}(X) \xrightarrow{a.s.} \mathbb{E}[\Phi_{\mathcal{X}}(X)]$ and $\bar{\Phi}_{\mathcal{Y}}(Y) \xrightarrow{a.s.} \mathbb{E}[\Phi_{\mathcal{Y}}(Y)]$ and, thus,
		\begin{equation}
			||(\mathbb{E}[\Phi_{\mathcal{X}}(X)]- \bar{\Phi}_{\mathcal{X}}(X))\otimes (\mathbb{E}[\Phi_{\mathcal{Y}}(Y)] - \bar{\Phi}_{\mathcal{Y}}(Y))||_{\mathcal{HS}} \xrightarrow{a.s.} 0.\label{Th1.3}
		\end{equation} 
		
		Taking square norms in (\ref{Th1.1}) and combining (\ref{Th1.2}) and (\ref{Th1.3}) with the continuous mapping theorem and the representation in (\ref{4}), it follows that $\text{\normalfont HSIC}_{n}(X,Y) \xrightarrow{a.s.} \text{\normalfont HSIC}(X,Y).$
	\end{proof}
	
	\begin{proof}[\textup{\textbf{Proof of Theorem \ref{Th2}}}]
		It follows from (T1) and (T2) that $\{\Phi_{\mathcal{X}c}(X_i)\}_{i\in \mathbb{Z}}$, $\{\Phi_{\mathcal{Y}c}(Y_i)\}_{i\in \mathbb{Z}}$ and $\{\Phi_{\mathcal{X}c}(X_i)\otimes\Phi_{\mathcal{Y}c}(Y_i) \}_{i\in \mathbb{Z}}$ are $L^1$ $\text{\normalfont NED}$ on $\{\eta_i\}_{i\in \mathbb{Z}}$ with coefficients $\{\tilde{a}_j\}_{j=0}^\infty$ defined as in Proposition \ref{pr4}. By the summability and finite moment assumptions in (T2), we can apply Theorem 1.1. in \cite{Dehling2015} to conclude that
		$$\frac{1}{\sqrt{n}}\sum\limits_{i=1}^n (\Phi_{\mathcal{X}c}(X_i)\otimes \Phi_{\mathcal{Y}c}(Y_i) - \mathbb{E}[\Phi_{\mathcal{X}c}(X)\otimes \Phi_{\mathcal{Y}c}(Y)])\xrightarrow{\mathcal{D}} N_{\mathcal{HS}},$$
		where $N_{\mathcal{HS}}$ is zero-mean, Gaussian and its covariance operator $\Gamma$ verifies 
		$$\langle s, \Gamma (s')\rangle_{\mathcal{HS}} = \sum\limits_{i\in \mathbb{Z}} \mathbb{E}\!\left[\langle S_0, s \rangle_{\mathcal{HS}} \langle S_i, s' \rangle_{\mathcal{HS}}  \right],$$
		for all $s,s' \in \mathcal{HS}$, where $S_i = (\Phi_{\mathcal{X}c}(X_i)\otimes \Phi_{\mathcal{Y}c}(Y_i) - \mathbb{E}[\Phi_{\mathcal{X}c}(X)\otimes \Phi_{\mathcal{Y}c}(Y)])$. Moreover, it also follows from Theorem 1.1. in \cite{Dehling2015} that
		\begin{align*}
			\sqrt{n}||(\mathbb{E}[\Phi_{\mathcal{X}}(X)]- \bar{\Phi}_{\mathcal{X}}(X))\otimes (\mathbb{E}[\Phi_{\mathcal{Y}}(Y)]- \bar{\Phi}_{\mathcal{Y}}(Y))||_{\mathcal{HS}} &= \sqrt{n}||\mathbb{E}[\Phi_{\mathcal{X}}(X)]- \bar{\Phi}_{\mathcal{X}}(X)||_{\mathcal{C}_{\mathcal{X}}}||\mathbb{E}[\Phi_{\mathcal{Y}}(Y)]- \bar{\Phi}_{\mathcal{Y}}(Y)||_{\mathcal{C}_{\mathcal{Y}}}\nonumber \\
			&= \frac{1}{\sqrt{n}}\!\left|\left|\frac{1}{\sqrt{n}}\sum\limits_{i=1}^n\Phi_{\mathcal{X}c}(X_i)\right|\right|_{\mathcal{C}_{\mathcal{X}}}\!\left|\left|\frac{1}{\sqrt{n}}\sum\limits_{i=1}^n\Phi_{\mathcal{Y}c}(Y_i)\right|\right|_{\mathcal{C}_{\mathcal{Y}}}
			= O_{\mathbb{P}}\!\left(\frac{1}{\sqrt{n}}\right). \nonumber
		\end{align*}
		Therefore, the limit distributions in Theorem \ref{Th2} hold by virtue of the continuous mapping theorem, Slutsky's theorem and equations (\ref{4}) and (\ref{Th1.1}).
	\end{proof}
	
	\begin{proof}[\textup{\textbf{Proof of Theorem \ref{Th3}}}]
		
		On the one hand, the following decomposition trivially holds:
		\begin{align*}
			&\ \frac{1}{n}\sum\limits_{i=1}^n (\Phi_{\mathcal{X}}(X_{in}) - \mathbb{E}[\Phi_{\mathcal{X}}(X_{in})])\otimes (\Phi_{\mathcal{Y}}(Y_{in})- \mathbb{E}[\Phi_{\mathcal{Y}}(Y_{in})])\\
			&= \frac{1}{n}\sum\limits_{i=1}^n (\Phi_{\mathcal{X}}(X_{i}) - \mathbb{E}[\Phi_{\mathcal{X}}(X_{i})])\otimes (\Phi_{\mathcal{Y}}(Y_{i})- \mathbb{E}[\Phi_{\mathcal{Y}}(Y_{i})])\\
			&\quad +\frac{1}{n}\sum\limits_{i=1}^n \left(\frac{Z_{\mathcal{X}i}- \mathbb{E}[Z_{\mathcal{X}i}]}{k_n} + R_{\mathcal{X}in} - \mathbb{E}[R_{\mathcal{X}in}] \right)\otimes (\Phi_{\mathcal{Y}}(Y_{i})- \mathbb{E}[\Phi_{\mathcal{Y}}(Y_{i})])\\
			&\quad +\frac{1}{n}\sum\limits_{i=1}^n (\Phi_{\mathcal{X}}(X_{i}) - \mathbb{E}[\Phi_{\mathcal{X}}(X_{i})])\otimes \left(\frac{Z_{\mathcal{Y}i} - \mathbb{E}[Z_{\mathcal{Y}i}]}{k_n} + R_{\mathcal{Y}in} - \mathbb{E}[R_{\mathcal{Y}in}] \right)\\
			&\quad +\frac{1}{n}\sum\limits_{i=1}^n \left(\frac{Z_{\mathcal{X}i}- \mathbb{E}[Z_{\mathcal{X}i}]}{k_n} + R_{\mathcal{X}in} - \mathbb{E}[R_{\mathcal{X}in}] \right)\otimes \left(\frac{Z_{\mathcal{Y}i} - \mathbb{E}[Z_{\mathcal{Y}i}]}{k_n} + R_{\mathcal{Y}in} - \mathbb{E}[R_{\mathcal{Y}in}] \right) = I_{1n} + I_{2n} + I_{3n} + I_{4n},
		\end{align*}
		where $I_{1n},I_{2n},I_{3n}$  and $I_{4n}$ are implicitly defined. We have shown that $\sqrt{n}(I_{1n} - \mathbb{E}[\Phi_{\mathcal{X}c}(X)\otimes \Phi_{\mathcal{Y}c}(Y)]) \xrightarrow{\mathcal{D}} N_{\mathcal{HS}}$ in the proof of Theorem \ref{Th2}. Morover, since functionals of ergodic and stationary processes are, themselves, ergodic and stationary, it follows from the ergodic theorem in reflexive Banach spaces that
		$$\frac{1}{n}\sum\limits_{i=1}^n (Z_{\mathcal{X}i} - \mathbb{E}[Z_{\mathcal{X}i}]) \otimes \Phi_{\mathcal{Y}c}(Y_{i}) \xrightarrow{a.s} \mathbb{E}[(Z_{\mathcal{X}i} - \mathbb{E}[Z_{\mathcal{X}i}]) \otimes \Phi_{\mathcal{Y}c}(Y_{i})],$$
		since, by (T4), the limit is well-defined. Moreover,
		\begin{align*}
			\mathbb{E}\!\left[\left|\left|\frac{1}{n}\sum\limits_{i=1}^n ( R_{\mathcal{X}in} - \mathbb{E}[R_{\mathcal{X}in}])\otimes \Phi_{\mathcal{Y}c}(Y_{i})\right|\right|_{\mathcal{HS}}\right]
			\le \mathbb{E}\!\left[\left|\left| R_{\mathcal{X}in} - \mathbb{E}[R_{\mathcal{X}in}] \right|\right|_{\mathcal{C}_{\mathcal{X}}} \!\left|\left|\Phi_{\mathcal{Y}}(Y_{i})\right|\right|_{\mathcal{C}_{\mathcal{Y}}}\right]
			= o\left(\frac{1}{k_n}\right). \nonumber 
		\end{align*}
		Therefore, $k_nI_{2n} = \mathbb{E}[Z_{\mathcal{X}i}  \otimes \Phi_{\mathcal{Y}c}(Y_{i})] + o_{\mathbb{P}}(1)$. By the same arguments, $k_nI_{3n} = \mathbb{E}[\Phi_{\mathcal{X}c}(X_{i})\otimes Z_{\mathcal{Y}i}] + o_{\mathbb{P}}(1)$ and $k_nI_{4n} = o_{\mathbb{P}}(1)$.
		
		On the other hand, $\bar{\Phi}_{\mathcal{Y}}(Y) -	\mathbb{E}[\Phi_{\mathcal{Y}}(Y_{i})] = O_{\mathbb{P}}(n^{-1/2})$, as shown in the proof of Theorem \ref{Th2}. Moreover, by the ergodic mean theorem and the moment conditions on $R_{\mathcal{X}in}$ in (T4),
		\begin{align*}
			\bar{\Phi}_{\mathcal{Y}}(Y_{\cdot n}) -	\mathbb{E}[\Phi_{\mathcal{Y}}(Y_{in})] &= \bar{\Phi}_{\mathcal{Y}}(Y) -	\mathbb{E}[\Phi_{\mathcal{Y}}(Y)] + \frac{1}{k_n}\!\left(\bar{Z}_{\mathcal{Y}\cdot} - \mathbb{E}[Z_{\mathcal{Y}i}]\right) + \bar{R}_{\mathcal{Y}\cdot n}  - \mathbb{E}[R_{\mathcal{Y}in}]\\
			&= O_{\mathbb{P}}\!\left(\frac{1}{\sqrt{n}}\right) + O_{\mathbb{P}}\!\left(\frac{1}{k_n}\right) + o_{\mathbb{P}}\!\left(\frac{1}{k_n}\right),
		\end{align*}
		where $\bar{\Phi}_{\mathcal{Y}}(Y_{\cdot n})$, $\bar{Z}_{\mathcal{Y}\cdot}$ and $\bar{R}_{\mathcal{Y}\cdot n}$ denote the corresponding sample averages.
		Analogously, $\bar{\Phi}_{\mathcal{X}}(X_{\cdot n}) -	\mathbb{E}[\Phi_{\mathcal{X}}(X_{in})] = O_{\mathbb{P}}(1/\sqrt{n}) + O_{\mathbb{P}}(1/k_n)$ and, consequently,
		\begin{align*}
			&||(\mathbb{E}[\Phi_{\mathcal{X}}(X_{in})] - \bar{\Phi}_{\mathcal{X}}(X_{\cdot n}))\otimes (\mathbb{E}[\Phi_{\mathcal{Y}}(Y_{in})] - \bar{\Phi}_{\mathcal{Y}}(Y_{\cdot n}))||_{\mathcal{HS}}
			= o_{\mathbb{P}}\!\left(\frac{1}{\sqrt{n}}\right) + o_{\mathbb{P}}\!\left(\frac{1}{k_n}\right).
		\end{align*}
		
		The decomposition in (\ref{Th1.1}), replacing $(X_i,Y_i)$ with $(X_{in}, Y_{in})$, still applies and the previous results imply that
		\begin{align*}
			&\ \frac{1}{n}\sum_{i=1}^n (\Phi_{\mathcal{X}}(X_{in}) - \bar{\Phi}_{\mathcal{X}}(X_{\cdot n}))\otimes (\Phi_{\mathcal{Y}}(Y_{in})- \bar{\Phi}_{\mathcal{Y}}(Y_{\cdot n}))\\
			&= (I_{1n} - \mathbb{E}[\Phi_{\mathcal{X}c}(X)\otimes \Phi_{\mathcal{Y}c}(Y)]) + \mathbb{E}[\Phi_{\mathcal{X}c}(X)\otimes \Phi_{\mathcal{Y}c}(Y)] + \frac{\mu}{k_n}+ o_{\mathbb{P}}\!\left(\frac{1}{\sqrt{n}}\right) + o_{\mathbb{P}}\!\left(\frac{1}{k_n}\right) .
		\end{align*}
		Then, taking square norms in the last equality, the results in Theorem \ref{Th3} follow from (\ref{4}), the continous mapping theorem and Slutsky's theorem .
	\end{proof}
	
	\subsection{Proof of Theorems 4 and 5}
	\label{app2.3}
	
	\begin{proof}[\textup{\textbf{Proof of Theorem \ref{Th4}}}]
		By the ergodic mean theorem, $\mathbb{E}[\Phi_\mathcal{X}(X)] -\bar{\Phi}_\mathcal{X}(X) \rightarrow 0$ a.s. and $\mathbb{E}[\Phi_\mathcal{Y}(Y)] -\bar{\Phi}_\mathcal{Y}(Y)\rightarrow 0$ a.s. and, in particular, those two terms converge to zero as $n\rightarrow \infty$ for almost all realizations of $\{\eta_{i}\}_{i\in \mathbb{Z}}$. In addition, since $ \sum_{i=1}^{n}(r_{in}-\bar{r}_{\cdot n}) = 0$, the following decomposition follows by adding and substracting terms conveniently:
		\begin{align*}
			&\ \frac{1}{n}\sum_{i=1}^{n} (r_{in}-\bar{r}_{\cdot n})(\Phi_{\mathcal{X}}(X)- \bar{\Phi}_{\mathcal{X}}(X)) \otimes  (\Phi_{\mathcal{Y}}(Y_i)- \bar{\Phi}_{\mathcal{Y}}(Y))\\
			&= \frac{1}{n}\sum_{i=1}^{n} (r_{in}-\bar{r}_{\cdot n})(\Phi_{\mathcal{X}}(X_i)- \mathbb{E}[\Phi_{\mathcal{X}}(X)]) \otimes  (\Phi_{\mathcal{Y}}(Y_i)- \mathbb{E}[\Phi_{\mathcal{Y}}(Y)])\\
			&\quad +\frac{1}{n}\sum_{i=1}^{n}  (r_{in}-\bar{r}_{\cdot n})(\Phi_{\mathcal{X}}(X_i)-\mathbb{E}[\Phi_{\mathcal{X}}(X)])\otimes (\mathbb{E}[\Phi_{\mathcal{Y}}(Y)] - \bar{\Phi}_{\mathcal{Y}}(Y))\\
			&\quad +(\mathbb{E}[\Phi_{\mathcal{X}}(X)] - \bar{\Phi}_{\mathcal{X}}(X))\otimes \frac{1}{n}\sum_{i=1}^{n}  (r_{in}-\bar{r}_{\cdot n})(\Phi_{\mathcal{Y}}(Y_i) -\mathbb{E}[\Phi_{\mathcal{Y}}(Y)]) = I_{1n} + I_{2n} + I_{3n},
		\end{align*}
		where $I_{1n},I_{2n}$ and $I_{3n}$ are implicitly defined. Recall that $\{\Phi_{\mathcal{X}c}(X_i)\}_{i\in \mathbb{Z}}$,  $\{\Phi_{\mathcal{Y}c}(Y_i)\}_{i\in \mathbb{Z}}$ and $\{\Phi_{\mathcal{X}c}(X_i)\otimes\Phi_{\mathcal{Y}c}(Y_i) \}_{i\in \mathbb{Z}}$ are $L^1$ $\text{\normalfont NED}$ on $\{\eta_i\}_{i\in \mathbb{Z}}$ with coefficients defined as in Proposition \ref{pr4} which, if (T5) holds, verify the summability conditions required to apply Theorem \ref{Th11}. Therefore, for any subsequence $\{n_k\}_{k\in \mathbb{N^+}}$ there is a further subsequence $\{n_{k_l}\}_{l\in \mathbb{N^+}}$ such that, for almost all realizations of $\{\eta_{i}\}_{i\in \mathbb{Z}}$, the conditional law of $\sqrt{n_{k_l}}I_{1n_{k_l}}$ given $\{\eta_{i}\}_{i\in \mathbb{Z}}$ converges to that of $N_{\mathcal{HS}}$ and, moreover, $||\sum_{i=1}^{n_{k_l}}  (r_{in_{k_l}}-\bar{r}_{\cdot n_{k_l}})\Phi_{\mathcal{Y}c}(Y_i)||_{\mathcal{C}_{\mathcal{X}}}/\sqrt{n_{k_l}}$ and $||\sum_{i=1}^{n_{k_l}}  (r_{in_{k_l}}-\bar{r}_{\cdot n_{k_l}})\Phi_{\mathcal{X}c}(X_i)||_{\mathcal{C}_{\mathcal{Y}}}/\sqrt{n_{k_l}}$ are bounded in (conditional) probability by the continous mapping theorem. Consequently, by Slutsky's theorem, $\sqrt{n_{k_l}}(I_{1n_{k_l}}+I_{2n_{k_l}}+I_{3n_{k_l}})$ converges weakly to $N_{\mathcal{HS}}$ for almost all realizations of $\{\eta_{i}\}_{i\in \mathbb{Z}}$. Taking squared norms and applying the continuous mapping theorem once again, we obtain (i). 
		
		We can write $\left|\left| N_{\mathcal{HS}}\right|\right|^2_{\mathcal{HS}} = \sum_{j=1}^\infty \lambda_j |N_j|^2$ a.s., where $\{\lambda_j\}_{j=1}^\infty$ is the collection of eigenvalues of the covariance operator of $N_{\mathcal{HS}}$. If there is, at least, one positive eigenvalue, say $\lambda_1$, then $\left|\left| N_{\mathcal{HS}}\right|\right|^2_{\mathcal{HS}}$ is the sum of two well-defined random variables, one of them absolutely continous. Therefore, $\left|\left| N_{\mathcal{HS}}\right|\right|^2_{\mathcal{HS}}$ is absolutely continuous. Since quantiles of the conditional distribution of $n\text{\normalfont HSIC}^*_{n}(X_{i},Y_{i})$ converge to those of $\left|\left| N_{\mathcal{HS}}\right|\right|^2_{\mathcal{HS}}$ in probability (which is easy to show by a subsequence argument), denoting by $Q(1-\alpha)$ the $(1-\alpha)$th quantile of the distribution of $\left|\left| N_{\mathcal{HS}}\right|\right|^2_{\mathcal{HS}}$, we have
		\begin{align*}
			&\ \big|\mathbb{P}(n\text{\normalfont HSIC}_{n}(X,Y)>Q_n^*(1-\alpha)) - \mathbb{P}(n\text{\normalfont HSIC}_{n}(X,Y)>Q(1-\alpha))\big|\\
			&\le \mathbb{P}(Q_n^*(1-\alpha)\ge n\text{\normalfont HSIC}_{n}(X,Y)>Q(1-\alpha), Q_n^*(1-\alpha)>Q(1-\alpha))\\
			&\quad +\mathbb{P}(Q(1-\alpha)\ge n\text{\normalfont HSIC}_{n}(X,Y)>Q_n^*(1-\alpha),Q(1-\alpha)\ge Q_n^*(1-\alpha))\\
			&\le \mathbb{P}(Q_n^*(1-\alpha)>Q(1-\alpha)) + \mathbb{P}(Q(1-\alpha)>Q_n^*(1-\alpha)) \rightarrow 0,
		\end{align*}
		so (ii) follows from Theorem \ref{Th2}.
	\end{proof}
	
	\begin{proof}[\textup{\textbf{Proof of Theorem \ref{Th5}}}]

		Let us denote by $\hat{\Phi}_{\mathcal{X}c}(X_i) = \Phi_{\mathcal{X}}(X_i)- \bar{\Phi}_{\mathcal{X}}(X)$, $\hat{\Phi}_{\mathcal{Y}c}(Y_i) = \Phi_{\mathcal{Y}}(Y_i)- \bar{\Phi}_{\mathcal{Y}}(Y)$,  $\hat{\Phi}_{\mathcal{X}c}(X_{in}) = \Phi_{\mathcal{X}}(X_{in})- \bar{\Phi}_{\mathcal{X}}(X_{\cdot n})$ and $\hat{\Phi}_{\mathcal{Y}c}(Y_{in}) = \Phi_{\mathcal{Y}}(Y_{in})- \bar{\Phi}_{\mathcal{Y}}(Y_{\cdot n})$, where $\bar{\Phi}_{\mathcal{X}}(X_{\cdot n})$ and $\bar{\Phi}_{\mathcal{Y}}(Y_{\cdot n})$ denote sample averages of $\Phi_{\mathcal{X}}$ and $\Phi_{\mathcal{Y}}$ with respect to $\{X_{in}\}_{i=1}^n$ and $\{Y_{in}\}_{i=1}^n$, respectively. In the proof of Theorem \ref{Th4} it is shown that, conditionally on $\{\eta_{i}\}_{i\in \mathbb{Z}}$,
		\begin{equation} \label{Th5.1}
			\frac{1}{\sqrt{n}}\sum\limits_{i=1}^{n} (\tilde{r}_{in}-\bar{r}_{\cdot n})\hat{\Phi}_{\mathcal{X}c}(X_i) \otimes \hat{\Phi}_{\mathcal{Y}c}(Y_i) \xrightarrow{\mathcal{D}^*} N_{\mathcal{HS}} \ \text{in }\mathbb{P},
		\end{equation}
		independently of the value of $\text{\normalfont HSIC}(X,Y)$. Moreover, $\text{Var}(\bar{r}_{\cdot n}) = O(\ell_n/n)$ (see the proof of Theorem \ref{Th11}) and, therefore, $\mathbb{E}[|\bar{r}_{\cdot n}|] = O(\sqrt{\ell_n/n})$. Thus,
		\begin{align}
			&\ \mathbb{E}\!\left[\left|\left|\frac{1}{\sqrt{n}}\sum\limits_{i=1}^{n} \bar{r}_{\cdot n}(\hat{\Phi}_{\mathcal{X}c}(X_{in}) \otimes \hat{\Phi}_{\mathcal{Y}c}(Y_{in})-\hat{\Phi}_{\mathcal{X}c}(X_{i}) \otimes \hat{\Phi}_{\mathcal{Y}c}(Y_i))\right|\right|_{\mathcal{HS}} |\{\eta_{i}\}_{i\in \mathbb{Z}} \right]\nonumber \\ 
			&\le \mathbb{E}\!\left[|\bar{r}_{\cdot n}|\right] \frac{1}{\sqrt{n}}\sum\limits_{i=1}^{n}||\hat{\Phi}_{\mathcal{X}c}(X_{in}) \otimes \hat{\Phi}_{\mathcal{Y}c}(Y_{in})-\hat{\Phi}_{\mathcal{X}c}(X_{i}) \otimes \hat{\Phi}_{\mathcal{Y}c}(Y_i)||_{\mathcal{HS}} \nonumber \\
			&= O\left(\sqrt{\ell_n}\right) \frac{1}{n}\sum\limits_{i=1}^{n}||\hat{\Phi}_{\mathcal{X}c}(X_i)||_{\mathcal{C}_{\mathcal{X}}} ||\hat{\Phi}_{\mathcal{Y}c}(Y_{in})-\hat{\Phi}_{\mathcal{Y}c}(Y_i)||_{\mathcal{C}_{\mathcal{Y}}}+O\left(\sqrt{\ell_n}\right) \frac{1}{n}\sum\limits_{i=1}^{n}||\hat{\Phi}_{\mathcal{X}c}(X_{in})-\hat{\Phi}_{\mathcal{X}c}(X_i))||_{\mathcal{C}_{\mathcal{X}}}||\hat{\Phi}_{\mathcal{Y}c}(Y_i)||_{\mathcal{C}_{\mathcal{Y}}} \nonumber \\  
			&\quad +O\left(\sqrt{\ell_n}\right) \frac{1}{n}\sum\limits_{i=1}^{n}||\hat{\Phi}_{\mathcal{X}c}(X_{in})-\hat{\Phi}_{\mathcal{X}c}(X_i))||_{\mathcal{C}_{\mathcal{X}}}||\hat{\Phi}_{\mathcal{Y}c}(Y_{in}) - \hat{\Phi}_{\mathcal{Y}c}(Y_{i})||_{\mathcal{C}_{\mathcal{Y}}}= O_{\mathbb{P}}\!\left(\frac{\sqrt{\ell_n}}{k_n}\right),  \label{Th5.2}
		\end{align}
		where the last equality follows from (T3), (T4) and Markov's inequality since, for example,
		expanding $\hat{\Phi}_{\mathcal{X}c}$ and $\hat{\Phi}_{\mathcal{Y}c}$ and applying the triangle inequality we have
		\begin{align}
			&\ \frac{1}{n}\sum\limits_{i=1}^{n}||\hat{\Phi}_{\mathcal{X}c}(X_i)||_{\mathcal{C}_{\mathcal{X}}} ||\hat{\Phi}_{\mathcal{Y}c}(Y_{in})-\hat{\Phi}_{\mathcal{Y}c}(Y_i)||_{\mathcal{C}_{\mathcal{Y}}}\nonumber\\
			&\le \frac{1}{n}\sum\limits_{i=1}^{n}||\Phi_{\mathcal{X}}(X_i)||_{\mathcal{C}_{\mathcal{X}}} ||\Phi_{\mathcal{Y}}(Y_{in})-\Phi_{\mathcal{Y}}(Y_i)||_{\mathcal{C}_{\mathcal{Y}}}
			+\frac{3}{n}\sum\limits_{i=1}^{n}||\Phi_{\mathcal{X}}(X_i)||_{\mathcal{C}_{\mathcal{X}}}\frac{1}{n}\sum\limits_{i=1}^{n} ||\Phi_{\mathcal{Y}}(Y_{in})-\Phi_{\mathcal{Y}}(Y_i)||_{\mathcal{C}_{\mathcal{Y}}} \nonumber\\
			&= O_{\mathbb{P}}\!\left(\frac{1}{k_n}\right)+O_{\mathbb{P}}\!\left(1\right)O_{\mathbb{P}}\!\left(\frac{1}{k_n}\right). \nonumber  
		\end{align}
		Moreover, by previous arguments, condition (T6), the AM-GM inequality and the assumption $\sum_{k=1}^n \rho_n(k) = O(\ell_n)$, it follows that
		\begin{align}
			&\ \mathbb{E}\!\left[\left|\left|\frac{1}{\sqrt{n}}\sum\limits_{i=1}^{n} \tilde{r}_{in}(\hat{\Phi}_{\mathcal{X}c}(X_{in}) \otimes \hat{\Phi}_{\mathcal{Y}c}(Y_{in})-\hat{\Phi}_{\mathcal{X}c}(X_{i}) \otimes \hat{\Phi}_{\mathcal{Y}c}(Y_i))\right|\right|^2_{\mathcal{HS}} |\{\eta_{i}\}_{i\in \mathbb{Z}} \right]\nonumber \\
			&= \frac{1}{n}\sum\limits_{i,j=1}^{n}\rho_n(|i-j|)\langle\hat{\Phi}_{\mathcal{X}c}(X_{in}) \otimes \hat{\Phi}_{\mathcal{Y}c}(Y_{in})-\hat{\Phi}_{\mathcal{X}c}(X_{i}) \otimes \hat{\Phi}_{\mathcal{Y}c}(Y_i),\hat{\Phi}_{\mathcal{X}c}(X_{jn}) \otimes \hat{\Phi}_{\mathcal{Y}c}(Y_{jn})-\hat{\Phi}_{\mathcal{X}c}(X_{j}) \otimes \hat{\Phi}_{\mathcal{Y}c}(Y_j) \rangle_{\mathcal{HS}}\nonumber\\
			&\le \frac{1}{n}\sum\limits_{i,j=1}^{n}\rho_n(|i-j|)||\hat{\Phi}_{\mathcal{X}c}(X_i) \otimes \hat{\Phi}_{\mathcal{Y}c}(Y_{in})-\hat{\Phi}_{\mathcal{X}c}(X_{in}) \otimes \hat{\Phi}_{\mathcal{Y}c}(Y_i)||^2_{\mathcal{HS}}= O_{\mathbb{P}}\!\left(\frac{\ell_n}{k_n^2}\right).  \label{Th5.3}
		\end{align}
		
		On the one hand, if $\ell_n=o(k_n^2)$ (which holds if $\sqrt{n}/k_n = O(1)$), putting (\ref{Th5.1}), (\ref{Th5.2}) and (\ref{Th5.3}) together, it follows from the continuous mapping theorem and Slusky's theorem (by means of subsequence arguments, as in the proof of Theorem \ref{Th4}) that
		$$n\text{\normalfont HSIC}^*_{n}(X_{\cdot n},Y_{\cdot n}) \xrightarrow{\mathcal{D}^*} ||N_{\mathcal{HS}}||^2_{\mathcal{HS}} \ \text{in }\mathbb{P},$$
		which, in particular, implies that $Q_n^*(1-\alpha)\rightarrow Q(1-\alpha)$ in probability and, therefore, $Q_n^*(1-\alpha) = O_{\mathbb{P}}(1)$ for all $\alpha \in (0,1)$. On the other hand, if $k_n^2 = o(n)$ and since $(\ell_n/k_n^2)(k_n^2/n) = \ell_n/n = o(1)$, we have
		$$\frac{k^2_n}{n}n\text{\normalfont HSIC}^*_{n}(X_{\cdot n},Y_{\cdot n})\xrightarrow{\mathcal{D}^*} 0\ \text{in }\mathbb{P},$$
		so $({k^2_n}/{n})Q_n^*(1-\alpha) = o_{\mathbb{P}}(1)$ in this case. Thus, independently of $k_n$,  $\min\{1,({k^2_n}/{n})\}Q_n^*(1-\alpha) = O_{\mathbb{P}}(1)$ for all $\alpha\in (0,1)$.

		By Theorem \ref{Th3}, when $\text{\normalfont HSIC}(X,Y)>0$, it holds that $\text{\normalfont HSIC}_n(X_{\cdot n},Y_{\cdot n}) = \text{\normalfont HSIC}(X,Y) + o_{\mathbb{P}}(1)$. Then, as $k_n^2\rightarrow \infty$,
		$$\mathbb{P}\!\left(n\text{\normalfont HSIC}_n(X_{\cdot n},Y_{\cdot n}) > Q_n^*(1-\alpha)\right)=\mathbb{P}\!\left(\min\left\{1,\frac{k_n^2}{n}\right\}n\text{\normalfont HSIC}_n(X_{\cdot n},Y_{\cdot n}) > \min\left\{1,\frac{k_n^2}{n}\right\}Q_n^*(1-\alpha)\right) \rightarrow1.$$
		When $\text{\normalfont HSIC}(X,Y)=0$, $\sqrt{n}/k_n\rightarrow \infty$ and $\mu\neq 0$ , $n\text{\normalfont HSIC}_n(X_{\cdot n},Y_{\cdot n}) = n\left|\left|\mu \right|\right|^2_{\mathcal{HS}}/k_n^2 + o_{\mathbb{P}}(n/k_n^2)$ by Theorem \ref{Th3}. Thus, as $({k^2_n}/{n})Q_n^*(1-\alpha) = o_{\mathbb{P}}(1)$ in this scenario,
		$$\mathbb{P}\!\left(n\text{\normalfont HSIC}_n(X_{\cdot n},Y_{\cdot n}) > Q_n^*(1-\alpha)\right) = \mathbb{P}\!\left(k_n^2\text{\normalfont HSIC}_n(X_{\cdot n},Y_{\cdot n}) > \frac{k_n^2}{n}Q_n^*(1-\alpha)\right) \rightarrow1.$$
		Finally, if $\sqrt{n}/k_n\rightarrow M$ for $0\le M<\infty$,  $n\text{\normalfont HSIC}_n(X_{\cdot n},Y_{\cdot n})$ converges in distribution to $\left|\left| M\mu+N_{\mathcal{HS}}\right|\right|^2_{\mathcal{HS}}$ which is absolutely continuous provided that $N_{\mathcal{HS}}$ has one positive eigenvalue (see the proof of Theorem \ref{Th4}). Since $\ell_n=o(k_n^2)$ in this case, $Q_n^*(1-\alpha)\rightarrow Q(1-\alpha)$ in probability and, by the same arguments used to prove the second part of Theorem \ref{Th4}, 
		$$ \mathbb{P}\!\left(n\text{\normalfont HSIC}_n(X_{\cdot n},Y_{\cdot n}) > Q_n^*(1-\alpha)\right) \rightarrow 1-F_{M}(Q(1-\alpha)),$$
		where $F_{M}$ denotes the cumulative distribution function of $\left|\left| M\mu+N_{\mathcal{HS}}\right|\right|^2_{\mathcal{HS}}$ and $F_{M}(Q(1-\alpha)) = 1-\alpha$ in the particular case of $M= 0$.
	\end{proof}

	\section{Supplementary material}
	\label{app3}
	The supplementary material presents a real-data application of the proposed testing methodology.

	\bibliography{biblio}
	\bibliographystyle{cas-model2-names}
\end{document}


\begin{frontmatter}

\title{Supplementary material for ``Kernel-based independence and mean independence tests for weakly dependent data''} 

\author[1]{Daniel Diz-Castro}
\author[1]{Manuel Febrero-Bande}
\author[1]{Wenceslao Gonz\'alez-Manteiga}
 \affiliation[1]{organization={Department of Statistics, Mathematical Analysis and Optimization, Universidade de Santiago de Compostela},
             city={Santiago de Compostela},
             postcode={15782},
             country={Spain}}

\end{frontmatter}

\renewcommand{\thesection}{S.\arabic{section}}
\section*{Real-data application: the Hawaii Ocean Time Series}
\label{secS.1}

The Hawaii Ocean Time-series dataset, available from the Hawaii Ocean Time-series program's official website (http://hahana.soest.hawaii.edu/hot/hot-dogs/cextraction.html), includes repeated measurements over time of the chemical, physical, and biological profiles of ocean water as functions of hydrostatic pressure (and, therefore, depth) at the ALOHA station in Hawaii. In this material, we focus on temperature (\unit{\celsius}), salinity (psu, or practical salinity units, based on the PSS-78 scale), oxygen concentration (\SI{}{\micro\mol}/kg) and chloropigment concentration (\SI{}{\micro\g}/l), measured every 2 decibars of pressure from 2 to 200 decibars (approximately corresponding to depths between 1.98 and 198.6 meters). A total of $n=120$ profile measurements were collected from January 1, 2010, to December 31, 2022, although $3$ of them contained measurement errors in the form of negative values for some of the variables and were, therefore, discarded. The remaining $n=117$ curves, represented in Figure \ref{fig1}, comprise the data available for our analysis, which we interpret as functional data in the Hilbert space $L^2(2,200)$. All functional variables can be assumed to be bounded, so the moment conditions required for the theoretical results in the main document are automatically satisfied when using either the Gaussian or distance-induced kernels. In addition, the time span between profile measurements (generally over a month) is likely sufficiently large to justify a weak dependence structure such as $\mathbb{P}$ NED on an underlying $\beta$-mixing process.

\begin{figure}[t]
	\centering
	\includegraphics[width=15cm]{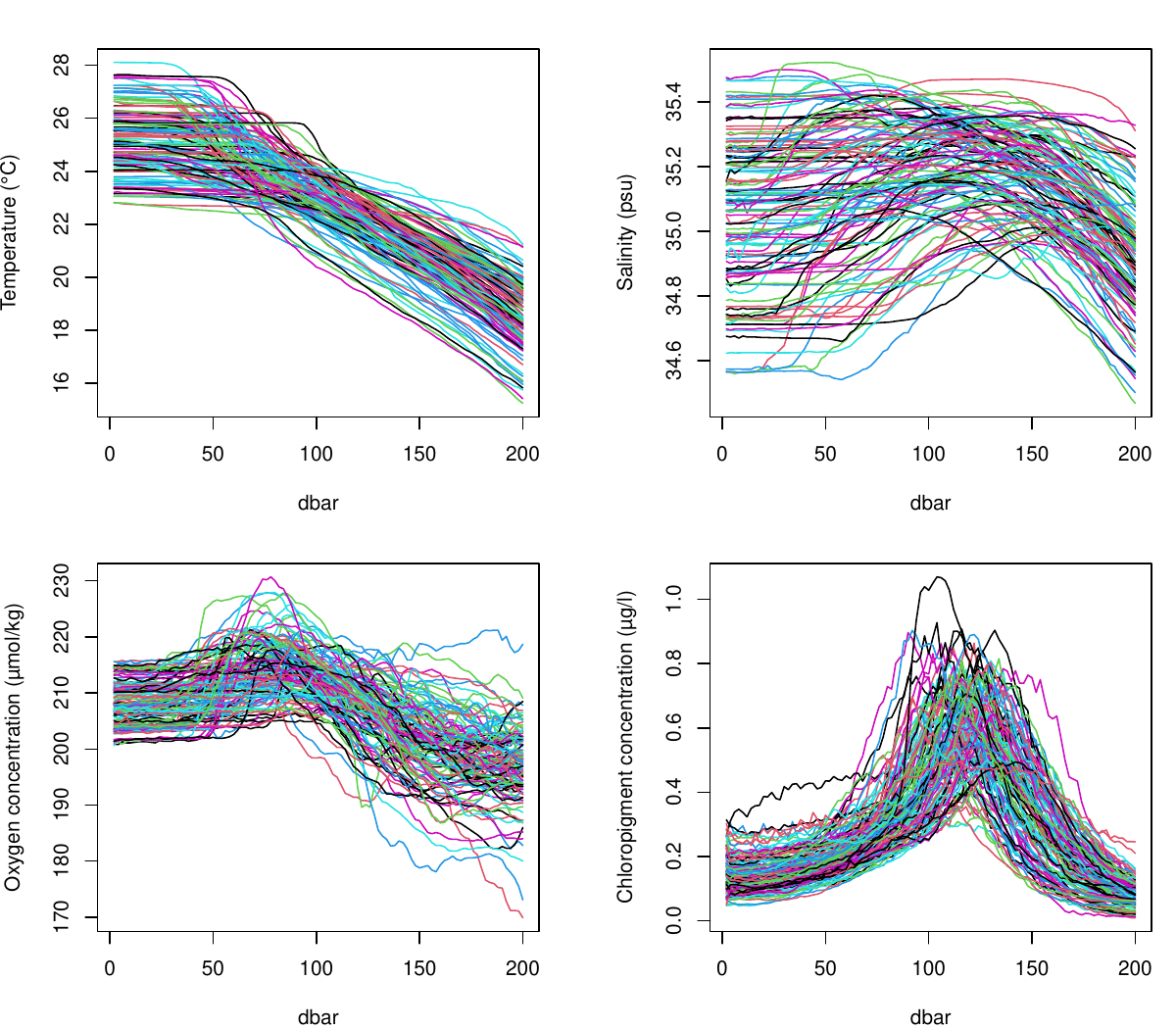}
	\caption{Plots of the oceanic temperature (top left), salinity (top right), oxygen concentration (bottom left) and chloropigment concentration (bottom right) profiles measured at the ALOHA station in Hawaii.}
	\label{fig1}
\end{figure}

Our first objective is to study the pairwise dependence relationships among the four selected variables. Heuristically, temperature is expected to influence the other three variables, as the solubility of different elements in water is usually affected by temperature and pressure, among other physical and chemical conditions. We also expect the oxygen concentration and chloropigment concentration to be closely related, since a high concentration of chloropigments indicates the presence of photosynthetic organisms, which influence oxygen concentration levels as both a source (through photosynthesis) and a sink (through respiration and bacterial decomposition of organic matter). For each pair of variables, we applied the bootstrap calibrated tests based on the kernel combinations denoted in the simulation study in the main document by DCOV (distance-based kernels) and $\text{HSIC}_{\text{G}}$ (Gaussian kernels) with $n_B=1000$ bootstrap replications and $l_n = 5$, the closest integer to $2n^{1/5}$ with $n=117$. The approximated $p$-values of the tests were:

\begin{itemize}
	\item Temperature and salinity: 0 (DCOV), 0.001 ($\text{HSIC}_{\text{G}}$).
	\item Temperature and oxygen concentration:  0 (DCOV), 0 ($\text{HSIC}_{\text{G}}$).
	\item Temperature and chloropigment concentration:  0.001 (DCOV), 0 ($\text{HSIC}_{\text{G}}$).
	\item Salinity and oxygen concentration:  0 (DCOV), 0.003 ($\text{HSIC}_{\text{G}}$).
	\item  Salinity and chloropigment concentration:  0.241 (DCOV), 0.250 ($\text{HSIC}_{\text{G}}$).
	\item  Oxygen concentration and chloropigment concentration:  0 (DCOV), 0 ($\text{HSIC}_{\text{G}}$).
\end{itemize}
The results from the DCOV and $\text{HSIC}_{\text{G}}$ tests are very similar in this case. Even accounting for the multiple comparison issue, it is clear that we can reject all pairwise independence assumptions at the level $\alpha = 0.05$ except for the pair (salinity, chloropigment concentration). Note that the measured salinity varies from 34.5 to 35.6 psu, which is a rather small range. This may explain the possible lack of dependence between the measured salinity and chloropigment concentration profiles, as photosynthetic organisms are expected to be adapted to the usual salinity of their environment.

Next, we consider the chloropigment concentration profile as a response variable to
illustrate how kernel-based mean independence tests can be used as an initial screening tool for regression modeling by identifying covariates with a significant marginal effect.
In particular, we apply the tests based on the kernel combinations denoted by MDD and KCMD in the simulation study in main document to try to detect whether the chloropigment concentration profile is mean independent of any of the other three variables. 
The resulting approximated $p$-values, obtained with $n_B = 1000$ and $l_n = 5$, were:

\begin{itemize}
	\item  chloropigment concentration on temperature:  0 (MDD), 0 ($\text{HSIC}_{\text{G}}$).
	\item  chloropigment concentration on salinity:  0.187 (MDD), 0.199 ($\text{HSIC}_{\text{G}}$).
	\item  chloropigment concentration on oxygen concentration:  0 (MDD), 0 ($\text{HSIC}_{\text{G}}$).
\end{itemize}
Both versions of the tests give, once again, almost identical results. Based on them, we can conclude that the temperature and oxygen concentration have a significant marginal effect on the conditional expectation of the oceanic chloropigment concentration at level $\alpha = 0.05$. We find no evidence against the hypothesis of mean independence of the chloropigment concentration from the salinity; thus, in this case, the independence and mean independence tests agree. This need not always be the case, however: absence of evidence is not evidence of absence. Therefore, even if the independence tests do not reject the corresponding null hypothesis, mean independence may still be rejected. Finally, if we assume that any covariate in a meaningful regression model for the chloropigment concentration profile must have a nontrivial marginal effect, we may screen out the salinity profile and consider only the temperature and oxygen concentration profiles as potentially relevant covariates.